\documentclass[showpacs,aps,superscriptaddress,notitlepage,reprint,longbibliography]{revtex4-1}

\usepackage[utf8]{inputenc}
\usepackage[T1]{fontenc}
\usepackage[dvips]{graphicx}
\usepackage{amsmath,amssymb,amsthm,mathrsfs,amsfonts,dsfont,amscd,keyval}
\usepackage{mathtools,mathrsfs}
\usepackage{yfonts} 
\usepackage{subfigure, epsfig}
\usepackage{titletoc} 
\usepackage[subfigure]{tocloft}
\usepackage{titlesec}

\usepackage{ytableau} \ytableausetup{boxsize = 2pt}
\usepackage{braket}
\usepackage{tensor}
\usepackage{bm}
\usepackage{enumerate,enumitem}
\usepackage{color}
\usepackage{hyperref}
\usepackage{tabularx}
\usepackage{times}
\hypersetup{breaklinks=true}
\hypersetup{colorlinks=true, linkcolor=blue, citecolor=magenta, urlcolor=blue, linktocpage=true}
\usepackage[hyphenbreaks]{breakurl}
\usepackage{multirow}
\usepackage{float}
\usepackage{quiver}
\usepackage{tcolorbox}

\newenvironment{proofsketch}
  {\par\noindent{\it Proof Sketch:}\hspace{0.5em}\rm}
  {\hfill$\square$\par}  

\DeclareMathOperator{\Ima}{Im}

\newtheorem{m-theorem}{Theorem}

\theoremstyle{definition}
\newtheorem{definition}{Definition}[section]
\newtheorem*{example}{Example}
\newtheorem{proposition}[definition]{Proposition}
\newtheorem{claim}{Claim}

\theoremstyle{plain}
\newtheorem{theorem}[definition]{Theorem}
\newtheorem{lemma}[definition]{Lemma}
\newtheorem{corollary}[definition]{Corollary}

\theoremstyle{remark}
\newtheorem*{remark}{Remark}

\newcommand{\comments}[1]{} 

\newcommand{\tr}{\mathrm{Tr}}

\newcommand{\CQA}{\mathrm{CQA}}
\DeclareMathOperator{\U}{U}
\DeclareMathOperator{\SU}{SU}
\DeclareMathOperator{\Cay}{Cay}
\DeclareMathOperator{\Sym}{Sym}
\DeclareMathOperator{\Asym}{Asym}
\DeclareMathOperator{\YJM}{YJM}
\DeclareMathOperator{\Wick}{Wick}

\DeclareMathOperator{\Comm}{Comm}

\DeclareMathOperator*{\medoplus}{\text{\raisebox{0.25ex}{\scalebox{0.75}{$\bigoplus$}}}}

\newcommand{\Lbra}[1]{\left\langle #1 \right\rvert}
\newcommand{\Lket}[1]{\left\lvert #1 \right\rangle}
\DeclareMathOperator{\sbA}{\ytableaushort{{*(black)}{}, {}{}}} 
\DeclareMathOperator{\sbB}{\ytableaushort{{}{*(black)}, {}{}}}
\DeclareMathOperator{\sbC}{\ytableaushort{{}{}, {*(black)}{}}}
\DeclareMathOperator{\sbD}{\ytableaushort{{}{}, {}{*(black)}}}

\newsavebox{\mstrut}
\newcommand{\bbra}[1]{%
    \sbox{\mstrut}{\(#1\)}%
    \mathinner{\langle\kern-0.3\ht\mstrut\left\langle{#1}\right|}%
}
\newcommand{\kett}[1]{%
    \sbox{\mstrut}{\(#1\)}%
    \mathinner{\left|{#1}\right\rangle\kern-0.3\ht\mstrut\rangle}%
}

\makeatletter
\def\l@subsubsection#1#2{} 
\makeatother

\begin{document}

\title{Efficient quantum pseudorandomness under conservation laws}

\author{Zimu Li}
\affiliation{Yau Mathematical Sciences Center, Tsinghua University, Beijing 100084, China}

\author{Han Zheng}
\affiliation{Pritzker School of Molecular Engineering, The University of Chicago, Chicago, IL 60637, USA}
\affiliation{Department of Computer Science, The University of Chicago, Chicago, IL 60637, USA}

\author{Zi-Wen Liu}
\affiliation{Yau Mathematical Sciences Center, Tsinghua University, Beijing 100084, China}

\begin{abstract}

The efficiency of locally generating unitary designs, which capture statistical notions of quantum pseudorandomness, lies at the heart of wide-ranging areas in physics and quantum information technologies. 
While there are extensive potent methods and results for this problem, the evidently important setting where continuous symmetries or conservation laws (most notably $\mathrm{U}(1)$ and $\mathrm{SU}(d)$) are involved is known to present fundamental difficulties. In particular, even the basic question of whether any local symmetric circuit can generate $2$-designs efficiently (in time that grows at most polynomially in the system size) remains open with no circuit constructions provably known to do so, despite intensive efforts. In this work, we resolve this long-standing open problem  for both $\mathrm{U}(1)$ and $\mathrm{SU}(d)$ symmetries by explicitly constructing local symmetric quantum circuits which we prove to converge to symmetric unitary $2$-designs in polynomial time using a combination of representation theory, graph theory, and Markov chain methods.  As a direct application, our constructions can be used to efficiently generate near-optimal covariant quantum error-correcting codes, confirming a conjecture in [PRX Quantum 3, 020314 (2022)]. 

\end{abstract}

\maketitle


\section{Introduction and background}

The generation of random quantum states and processes not only represents a core task in quantum information science and technology, underpinning various pivotal applications spanning areas including quantum device benchmarking~\cite{Randomized2005,knill2008randomized,Bannai2021,Elben_2022}, tomography~\cite{Scott2008,Huang2020shadow,Bertoni2024,Elben_2022,Zhu2024Clifford}, error correction~\cite{brown2013short,Brown_2015,preskill2021approxeastinknill,kong2022near}, information theory~\cite{hayden2004randomizing,Hayden-Preskill2007,dupuis2010decouplingapproachquantuminformation,Dupuis_2014,HastingsSuperAdd,Haferkamp2024linear}, and machine learning~\cite{biamonte2017quantum,NetKet,Zheng2021SpeedingUL,Liu2022QNTK,Liu2023QNTK}, but also plays critical roles in diverse areas in fundamental physics~\cite{Elben_2022,Yoshida-Kitaev2017,RobertsChaos2017,Liu_2018,PhysRevLett.120.130502,brandao2021models,Nahum1,Nahum2,Bringewatt2024,Liu2024Mpemba,chang2024deep}. In this regard, a central problem is to produce global random unitaries with quantum circuits composed of local gates, each of which can only involve a limited amount of subsystems. Such local circuit models are a pillar of the study of quantum computing and physical dynamics---they reflect the fundamental locality feature associated with experimental implementations and physical interactions, with which a time scale naturally emerges.

As a consequence of the exponential size of the unitary group in the number of subsystems, the generation of (even reasonable approximate notions of) uniformly random unitaries, namely unitaries distributed according to the Haar measure, requires exponentially large circuit depth (or time). However, for both the usage in practical tasks as well as applications in physics contexts, we are usually interested in ``pseudorandom'' distributions rigorously characterized by \emph{unitary $k$-designs}, i.e.,~ensembles of unitaries that only need to match the Haar measure up to $k$-th moments. In particular, the most important landmark is $k=2$, where the ensembles pervade the space and begin to exhibit nontrivial quantum randomness  useful for applications and genuine quantum features such as quantum information scrambling and global entanglement~\cite{Scott2008,roberts2016chaos,Liu_2018,You_2018,Kuo2020,Znidaric2021,Znidaric2022a,kong2022near,Swingle2022}. From a series of efforts~\cite{OliveiraDahlstenPlenio08,Znidaric2008,Harrow2design2009,Low2010,Brown_2010} leading up to a seminal work by Brandão, Harrow and Horodecki~\cite{harrow2016local,BHH16prl}, it is known that random quantum circuits composed of local random gates generate unitary $k$-designs efficiently, i.e.,~in polynomial time in $k$ and the system size, establishing the foundation for their technological and physical applications. Subsequently, there has been an array of progress that improves the parameters ~\cite{PhysRevX.7.021006,Haferkamp2021,brian2022linear,Gao2022,harrow2023approximate,mittal2023local,Haah2024Pauli,Metger2024,Schuster2024lowDepth,Haferkamp2024linear}.

However, when physical symmetries or conservation laws are imposed, the situation becomes much more mystifying. Under continuous global symmetries, most importantly $\U(1)$ and $\SU(d)$ symmetries which are canonical representatives of Abelian and non-Abelian symmetries respectively, global unitaries that can be generated by local circuits are severely restricted~\cite{MarvianNature,MarvianSU2,SUd-k-Design2023,marvian2023nonuniversality}, in stark contrast to the scenario without symmetry where $2$-local unitaries are already universal~\cite{Vlasov2001,Brylinski2001,Sawicki2016,Oszmaniec2022}. This fundamental distinction imposes serious restrictions on the capability of local circuits to generate designs. 
For instance, it is known that for general qudits with local dimension $d \geq 3$, $2$-local $\SU(d)$-symmetric ensembles cannot even form $2$-designs~\cite{MarvianSUd}. Partial progress towards merely understanding the \emph{possibility} of locally generating certain orders of designs has taken intensive efforts.  Ref.~\cite{SUd-k-Design2023} constructs explicit $4$-local circuit ensembles capable of converging to $\SU(d)$-symmetric designs for $k = O(n^2)$ and proves that any constant locality is not sufficient for generating designs of arbitrarily high order. For the same quadratic scaling of $k$, Ref.~\cite{Marvian3local} optimizes the locality requirement, showing that the $3$-local circuits suffice. Furthermore, Refs.~\cite{MarvianDesign,mitsuhashi2024Designs,mitsuhashi2024Designs2} systematically study the relation between the largest possible $k$ with given localities under transversal $\U(1)$, $\SU(d)$, and some discrete symmetries.  

Remarkably, the further question of whether symmetric designs can be provably efficiently generated has resisted substantial efforts and remained open, with only intrinsic barriers understood from various angles (see also e.g.~Refs.~\cite{kong2022near,SUd-k-Design2023}).
In particular, the proof strategies that have been effective in bounding the convergence time cannot be readily adapted to the case with continuous symmetries. 
First, these strategies usually involve methods from many-body theory such as Knabe bounds~\cite{Knabe1988,Gosset2016}, Nachtergaele's martingale method~\cite{Nachtergaele1996,Cirac2006}, and Lipkin--Meshkov--Glick Hamiltonian~\cite{LMG1965A,Znidaric2008,Brown_2010}, which can help prove $O(1/n)$ spectral gaps, while there are strong evidences~\cite{U(1)Design2023,SUd-k-Design2023} that $2$-local unitary ensembles under both $\U(1)$ and $\SU(2)$ symmetries have spectral gaps lower than $O(1/n)$, leading to polynomially prolonged convergence times. Moreover, the decomposition of the Hilbert space into smaller subspaces respecting the symmetry, which is a characteristic feature induced by symmetries, also prevents the use of another technique---approximate orthogonality of frame operators~\cite{harrow2016local,Haferkamp2021,Harrow2023orthogonality,Metger2024,Haferkamp2024linear}---for evaluating the convergence time, as the dimensionality of these subspaces may scale independently with the qudit local dimension. Subspace decomposition also hinders the computation of frame potential~\cite{junyu2020chargescrambler,hunter2019unitary,brian2022linear,SUd-k-Design2023}. Even in the most basic case of $k=2$, the challenges imposed by continuous symmetries are daunting. In the symmetry-free case, the convergence of unitary ensembles can be translated into the mixing of several independent random walks, which can be upper-bounded using tools from the thoery of Markov chain theory~\cite{Oliveira2design2007a,Oliveira2design2007b,Harrow2design2009,Diniz_2011,Brown_2015,harrow2023approximate}. However, the symmetry constraints obscure this reduction. For instance, the Pauli operators typically used for such reductions are no longer applicable due to the symmetry constraints. To conclude, given the broad practical motivations and theoretical depth evident from the above discussion,  the efficient generation of designs under continuous symmetries emerges as an especially important and interesting open problem.


\section{Summary of main results}
In this work, we make progress on this open problem by proving polynomial upper-bounds on the convergence times to $2$-designs under both $\U(1)$ and $\SU(d)$ symmetries. Combining techniques from representation theory, graph theory and Markov chain, we devise radically new approaches which manage to overcome the aforementioned difficulties. Roughly speaking, we explicit write down the representations of our circuit models by employing $S_n$ representation theory~\cite{Okounkov1996,Sagan01,Goodman2009} and study the convergence times by comparing with random walks on certain \emph{Cayley graphs} $\mathcal{G}(\mathcal{T})$ over the symmetric group $S_n$ given generating sets $\mathcal{T}$~\cite{Aldous1992,CaputoProof2010,Bacher1994,Flatto1985,Friedman2000,DiaconisComparison1993,Levin2009, chang2024deep}. 

Concretely, we consider the so-called Convolutional Quantum Alternating circuits proposed in Refs.~\cite{Zheng2021SpeedingUL,SUd-k-Design2023}, which we denote by  $\mathcal{E}_{\CQA,\times}$  for brevity, where ``$\times$'' can stand for either $\SU(d)$ or $\U(1)$ symmetry. Similar notations are also applied hereafter. Let $\tau$ denote a transposition/SWAP acting on 2 qudits and let $\mathcal{T}$ be a generating subset of $S_n$. In defining $\mathcal{E}_{\CQA,\times}$, we sample from two gates sets: (a) unitary time evolutions $e^{-i\theta \tau}$ for $\theta \in [0,2\pi]$ and $\tau \in \mathcal{T}$, and (b) unitary time evolutions of local Hamiltonians specified in \eqref{eq: wick-contraction-superoperator-U1} and \eqref{eq: wick-contraction-superoperator-SUd} ($2$-local for $\U(1)$ symmetry and $4$-local for $\SU(d)$ symmetry). Our main results can be informally summarized as follows: 
\begin{m-theorem}[Informal] \label{thm: CQA-second-largest-eigenvals}
	The local circuit ensembles defined by $\mathcal{E}_{\CQA,\times}$ can converge to symmetric $2$-designs within a polynomial depth. Specifically, for $n$ qubits under $\U(1)$ symmetry,
	\begin{enumerate}[itemsep=0mm, left=0mm]
		\item  If $\mathcal{T}$ consists of geometrically adjacent (nearest-neighbour) SWAPs on a 1D chain, the convergence time is $O(n^3(4n \log 2 + \log (1/\epsilon) ))$.
		\item If $\mathcal{T}$ consists of SWAPs on a complete graph (representing all-to-all interactions) or star graph, the convergence time is $O(n(4n \log 2 + \log (1/\epsilon) ))$.
	\end{enumerate}
	Under $\SU(d)$ symmetry with adjacent SWAPs, the convergence time is $O(n^3(4n \log d+ \log (1/\epsilon) ))$.
\end{m-theorem}

Here we study $\SU(d)$ symmetry for general $d$ (note the sharp difference for $d \geq 3$~\cite{Biedenharn1,Biedenharn2,Marin1,Marin2,Zheng2021SpeedingUL,MarvianSUd}) and only consider qubits for $\U(1)$ symmetry since the generalization to qudits is straightforward. For fixed localities, all other possible symmetric ensembles should similarly exhibit a polynomial scaling of convergence time (see Appendix \ref{sec:sketches}). We briefly introduce symmetric designs and articulate the associated fundamental difficulties in Section \ref{sec:basic}. Then we formally define the CQA ensemble and present our proof strategy in Section \ref{sec:CQAconvergence}.


\section{Unitary designs and symmetries}\label{sec:basic}

Let $\mathcal{E}$ be an ensemble (distribution) of unitaries acting on the Hilbert space $\mathcal{H}$. For any operator $M \in \operatorname{End}(\mathcal{H}^{\otimes k})$, the \emph{$k$-fold (twirling) channel} (or the \emph{$k$-th moment (super-)operator}) with respect to $\mathcal{E}$ acting on $M$ is defined by $T_k^{\mathcal{E}}(M) = \mathbb{E}_{\mathcal{E}}[ U^{\otimes k} M U^{\dagger \otimes k}]$.
Besides, given any compact group $G$, we use $T_k^G$ to denote the $k$-th moment operator defined by the Haar measure over $G$. An ensemble is called an \emph{(exact) unitary $k$-design} with respect to the group $G$ if $T^{\mathcal{E}}_k = T_k^G$. More generally, we call $\mathcal{E}$ an \emph{$\epsilon$-approximate $k$-design} if the strong notion of $\epsilon$-approximation in terms of complete positivity~\cite{harrow2016local,Gao2022,Metger2024,Haferkamp2024linear} holds, namely, $(1 - \epsilon) T_k^G \leq_{\mathrm{cp}} T_k^{\mathcal{E}} \leq_{\mathrm{cp}} (1 + \epsilon) T_k^G$, where $A \leq_{\mathrm{cp}} B$ means $B - A$ is completely positive and $c_{\mathrm{cp}}(\mathcal{E}, k)$ is the smallest constant $\epsilon$ achieving the above bound. Suppose $T^{\mathcal{E}}_k$ is Hermitian and positive semidefinite (PSD). It turns out that $c_{\mathrm{cp}}(\mathcal{E}, k)$ can be evaluated via the spectrum of $T^{\mathcal{E}}_k$. To be precise, let $\Delta(T^{\mathcal{E}}_k) = 1 - \lambda_2(T^{\mathcal{E}}_k)$ be the \emph{spectral gap} of $T^{\mathcal{E}}_k$ with $\lambda_2(T^{\mathcal{E}}_k)$ being its second largest eigenvalue. 
The spectral gap determines the rate at which the ensemble $\mathcal{E}$ converges to $k$-designs, and consequently, the required circuit depth for generating $\epsilon$-approximate a $k$-design. To be more precise, 
consider a circuit consisting of $p$ steps of random walks where in each step we sample a unitary from the ensemble $\mathcal{E}$. It can be shown~\cite{vanDam2002,Low2010,harrow2016local,SUd-k-Design2023} that when
\begin{align}\label{eq:ConvergenceDepth}
	p \geq \frac{1}{\Delta(T^{\mathcal{E}}_k)} (2kn\log d + \log 1/\epsilon),
\end{align}
this random circuit forms an $\epsilon$-approximate $k$-design (also see Ref.~\cite{Schuster2024lowDepth} for an improved dependence on $k$).
The canonical case where $G$ is the unitary group $\U(\mathcal{H}) \equiv \U(N) \equiv \U(d^n)$ of an $n$-qudit system with $\dim \mathcal{H} = N = d^n$ has been extensively studied in mathematics and quantum information literature. Prominent results include, e.g., efficient convergence to approximate $k$-designs by $2$-local random circuits~\cite{Dankert2026PRA,Gross2006,Znidaric2008,Harrow2design2009,Brown_2010,harrow2016local,hunter2019unitary,harrow2023approximate,mittal2023local,Metger2024,Haferkamp2024linear}, exact $k\leq 3$-designs induced by the Clifford group~\cite{DiVincenzo_2002,Gross2006,Dankert2026PRA,webb2015clifford,zhu2016clifford,zhu2017multiqubit,CliffordSampling2014,CliffordSampling2021} and exact designs of general orders~\cite{Guralnick2005,Bannai2018,Bannai2019,Bannai2021,PRXQuantum.2.030339}.

Let $\hat{\sigma} = \sum_{i_1, \cdots, i_k} \ket{i_{\sigma(1)}, \cdots i_{\sigma(k)}}\bra{i_1, \cdots, i_k}$ for $\sigma \in S_k$ and $i_j =0, \cdots, N-1$ denote permutation operators acting on the $k$-fold tensor $\mathcal{H}^{\otimes k}$. Obviously, $T^{\U(\mathcal{H})}_k(\hat{\sigma}) = \hat{\sigma}$. Moreover, it is well known that these permutations span all possible operators $M$ such that $T^{\U(\mathcal{H})}_k(M) = M$ by the double commutant theorem~\cite{Goodman2009,Tolli2009}. This invariance motivates the following complete set of projection super-operators, $\{\mathcal{P}_{\nu \zeta} := \Pi_\nu (\cdot) \Pi_{\zeta}: \nu, \zeta \vdash k \}$, where $\Pi_{\nu}, \Pi_{\zeta} $ denotes the projection  onto the irreducible representations (irreps) $S^\nu, S^\zeta$ of the symmetric group $S_k$ for the integer partitions $\nu, \zeta \vdash k$, respectively. We note that (see in Appendix~\ref{sec:kDesigns})
\begin{align}
	\mathcal{P}_{\nu \zeta}  \circ T^{\mathcal{E}}_k =  T^{\mathcal{E}}_k \circ \mathcal{P}_{\nu \zeta}.
\end{align}
In particular, this implies that $T^{\mathcal{E}}_k$ can be block-diagonalized by $\mathcal{P}_{\nu \zeta}$, namely $T^{\mathcal{E}}_k = \bigoplus_{\nu, \zeta \vdash k } \mathcal{P}_{\nu \zeta} \circ T^{\mathcal{E}}_k \circ \mathcal{P}_{\nu \zeta} \equiv T^{\mathcal{E}}_{k, \nu \zeta}$. For $\mathcal{E}$ to form a (approximate) unitary $k$-design, the unit eigenspaces of $T^{\mathcal{E}}_k$  must be spanned by these permutation operators $\hat{\sigma}$. By Schur's lemma, $\mathcal{P}_{\nu \zeta}(\hat{\sigma}) = \hat{\sigma} \delta_{\nu \zeta}$, so that $\mathcal{P}_{\nu \zeta}$ decompose the unit eigenspace into $\{ \Ima \mathcal{P}_{\nu \nu }; \nu \vdash k \}$ and it is well-known that the number of independent eigenvectors should be $k! = \sum_{\nu \vdash k} \dim(\Ima \mathcal{P}_{\nu \nu}) = \sum_{\nu \vdash k} (\dim S^\nu)^2$ when $k < d^n$~\cite{Rains1998}. It is important to note that the projection $\mathcal{P}_{\nu \zeta}$ does not guarantee all unit eigenvectors are separated, i.e. two eigenmatrices may lie in the same $\Ima \mathcal{P}_{\nu \nu}$.

For the design generation problem it is natural to introduce techniques from the study of spectral gaps and convergence times, which is a central topic in Markov chain theory~\cite{Diaconis1988,Levin2009}, by converting the moment operator $T_k^{\mathcal{E}}$ into Markov transition operators~\cite{Oliveira2design2007a,Oliveira2design2007b,Harrow2design2009,Diniz_2011,Brown_2015}. It is generally unclear how to further separate unit eigenvectors within $\Ima \mathcal{P}_{\nu \nu}$, which invalidate most well-known methods~\cite{DiaconisCheeger1991,DiaconisComparison1993} regardless of whether symmetries are involved or not. However, in the special case of $k=2$, $\nu \vdash 2$ can only correspond to either the symmetric (trivial) or the anti-symmetric (sign) irrep, both of which are one-dimensional. Letting $\Pi_{\pm}$ denote the corresponding projections, we define 
\begin{align}
	\begin{aligned}
		&\mathcal{P}_{\sbA} = \Pi_+ (\cdot) \Pi_+, \quad \mathcal{P}_{\sbD} = \Pi_- (\cdot) \Pi_-, \\
		&\mathcal{P}_{\sbB} = \Pi_+ (\cdot) \Pi_-, \quad \mathcal{P}_{\sbC} = \Pi_- (\cdot) \Pi_+,
	\end{aligned}
\end{align}
in which case the two orthonormal eigenmatrices are naturally split within $\Ima \mathcal{P}_{\sbA}$ and $\Ima \mathcal{P}_{\sbD}$ respectively. 

Besides using irreps, there is a simpler way to divide eigenvectors into different subspaces using Pauli basis in $2$-designs~\cite{Harrow2design2009,Diniz_2011,harrow2023approximate}. However, when conserved quantities are involved, this method is no longer applicable since the presence of conserved quantities shields certain degrees of freedom in the physical spaces. Let $G = \mathcal{U}_\times$ be the group consisting of unitaries respecting the symmetry. We typically assume the group elements as a direct sum of smaller unitaries $\bigoplus^J_{j=1} I_{m_{\lambda_j}} \otimes U_{\lambda_j}$ corresponding to the decomposition of Hilbert space $\mathcal{H} = \bigoplus^J_{j=1} \mathcal{H}_{j} \otimes \mathcal{H}^{\lambda_j}$, where $m_{\lambda_j} = \dim \mathcal{H}_j$ denotes the space of multiplicities and $\lambda_j$ labels  inequivalent charge sectors or irreps~\cite{kong2022near,MarvianNature,Zheng2021SpeedingUL,U(1)Design2023}. The Haar measure over $\mathcal{U}_\times$ is simply given by taking each $U_{\lambda_j} \in \U(\mathcal{H}^{\lambda_j})$ from the corresponding Haar measure. In such cases, the $k$-fold channel associated with $\mathcal{U}_\times$ involve a similar decomposition (for simplicity we ignore multiplicities and denote by $\lambda_{i_1}$ an arbitrary subspace label), 
\begin{align}\label{eq: k-fold-channel-symmetry-decomposition}
    & T^{\mathcal{U}_\times}_k (M) =  (\medoplus_{\substack{\lambda_{i_1}, \cdots, \lambda_{i_k} \\ \lambda_{j_1},\cdots, \lambda_{j_k}}} T^{\lambda_{i_1}, \cdots, \lambda_{i_k}}_{\lambda_{j_1},\cdots, \lambda_{j_k}})(M) \\
    & = (\medoplus_{\substack{\lambda_{i_1}, \cdots, \lambda_{i_k} \\ \lambda_{j_1},\cdots, \lambda_{j_k}}} \int_{\mathcal{U}_\times} dU  U^{\otimes^k_{\ell=1}}_{\lambda_{i_\ell}} \otimes \bar{U}^{\otimes^k_{\ell = 1}}_{\lambda_{j_\ell}}) (M)
    \equiv \mathcal{P}^{\Wick}_k \circ \widetilde{T}^{\mathcal{U}_\times}_k. \notag
\end{align}
Here we call $\mathcal{P}^{\Wick}_k$ the \emph{Wick projection}, which is a super-operator projection that aligns subspace labels and basis elements (see Appendix~\ref{sec:sketches} and \ref{sec:detailsT2CQA} for more details) in a way analogous to the Wick contraction in quantum field theory. The projection $\widetilde{T}^{\mathcal{U}_\times}_k$ can be taken arbitrarily as long as their composition yields $T^{\mathcal{U}_\times}_k$. The above formula emphasizes the significance of $\mathcal{P}^{\Wick}_k$, which will be introduced in detail and drastically simplifies the analysis. For example, when $k = 2$, for fixed subspace labels $i_1, i_2, j_1, j_2$ and its basis indices $a_1, a_2, b_1, b_2$, we can write $\mathcal{P}^{\Wick}_2 = \delta^{\{i_1, i_2\}}_{\{j_1, j_2\}} \delta^{\{a_1, a_2\}}_{\{b_1, b_2\}}$, which are nonzero when $\{i_1, i_2\} = \{j_1, j_2\}$ and $\{a_1, a_2\} = \{b_1, b_2\}$. Mathematically, $\mathcal{P}^{\Wick}_k$ can be constructed by elements in the so-called \emph{Gelfand-Tsetlin (GZ) algebra}~\cite{Okounkov1996,Tolli2009}. In particular, for our current focus on transversal $\U(1)$ and $\SU(d)$ symmetries, it is shown in Ref.~\cite{SUd-k-Design2023} and Appendix~\ref{sec:detailsT2CQA} that the Wick projections can be written as 
\begin{small}
    \begin{align}
       \hspace{-1mm}  & \mathcal{P}^{\Wick}_{k,\U(1)} \hspace{-1mm} =  \hspace{-1mm} \prod_{r,s} \hspace{-1mm} \int \hspace{-1mm} d\gamma \exp(i\gamma_{rs}Z_r Z_s)^{\otimes k}  \hspace{-1mm}  \otimes  \exp(-i\gamma_{rs}Z_r Z_s)^{\otimes k}, \label{eq: wick-contraction-superoperator-U1} \\
    \hspace{-1mm} &\mathcal{P}^{\Wick}_{k,\SU(d)} \hspace{-1mm} = \hspace{-1mm} \prod_{r,s} \hspace{-1mm} \int \hspace{-1mm} d\gamma \exp(i\gamma_{rs}X_r X_s)^{\otimes k} \hspace{-1mm}  \otimes  \exp(-i\gamma_{rs}X_r X_s)^{\otimes k} \hspace{-1mm}, \label{eq: wick-contraction-superoperator-SUd}
\end{align}
\end{small}
where $Z_i$ denotes Pauli-$Z$ acting on the $i$-th site and $X_i$ denote the $i$-th \emph{Young--Jucys--Murphy (YJM) operators}~\cite{Young1977,Jucys1974,Murphy1981,Okounkov1996}. By definition, these Pauli strings are $2$-local. The YJM operators are defined by SWAPs and hence their products are $4$-local (also see Appendix~\ref{sec:SnTheory} for more details). Finally, we note that $\mathcal{P}_{\ytableaushort{{} {}, {} {}}} \circ \mathcal{P}^{\Wick}_2 = \mathcal{P}^{\Wick}_2 \circ \mathcal{P}_{\ytableaushort{{} {}, {} {}}}$ with $\mathcal{P}_{\ytableaushort{{} {}, {} {}}}$ denoting any choice of invariant subspace projection. The proof, which holds for general $k$, is given in Appendix~\ref{sec:kDesigns}. 


\section{Symmetric CQA ensembles and convergence time}\label{sec:CQAconvergence}

By definition, the Wick projections for both $\U(1)$ and $\SU(d)$ can be constructed with locality-preserving operators in the formation of $2$-designs. They are also of great significance in facilitating our proof on the spectral gap. 
We include them into our symmetric circuit ensembles $\mathcal{E}_{\CQA,\times}$. Formally, we focus on the following second moment operators: 
\begin{align}
    & T^{\mathcal{E}_{\CQA}}_{\U(1)} := \mathcal{P}^{\Wick}_{2, \U(1)} \circ \Big(\frac{1}{|\mathcal{T}|}\sum_{\tau \in \mathcal{T}} T^\tau_2 \Big) \circ \mathcal{P}^{\Wick}_{2, \U(1)},  \label{eq: 2nd-channel-cqa-U1} \\
    & T^{\mathcal{E}_{\CQA}}_{\SU(d)} := \mathcal{P}^{\Wick}_{2, \SU(d)} \circ \Big(\frac{1}{|\mathcal{T}|}\sum_{\tau \in \mathcal{T}} T^\tau_2 \Big) \circ \mathcal{P}^{\Wick}_{2, \SU(d)}.  \label{eq: 2nd-channel-cqa-SUd}
\end{align}
The transposition $\tau$ respects both $\U(1)$ and $\SU(d)$ symmetries by definition, and $ T^\tau_2 = \mathbb{E}_\theta[ \exp(i\theta \tau)^{\otimes 2} \otimes \exp(-i \theta \tau)^{\otimes 2}]$. Physically, one can interpret the $T_2^\tau$'s as results of evenly drawing an index $j \in [n]$, a parameter $\theta \in [0, 2\pi]$, and implementing the $2$-local unitary $\exp(i\theta \tau)$ on the qudits. By Eqs.~\eqref{eq: wick-contraction-superoperator-U1} and~\eqref{eq: wick-contraction-superoperator-SUd}, the respective Wick projections can also be implemented in a similar way, which defines $\mathcal{E}_{\CQA,\times}$ (for more details, see Definition \ref{def:CQAEnsemble} and Appendix \ref{sec:sketches}).


We now sketch the strategy and essential ideas to rigorously upper bound the $2$-design convergence time of these ensembles. The proof details are left to Appendix~\ref{sec:details}. We first emphasize that when imposing a global symmetry including $\U(1)$ and $\SU(d)$, the dimension of unit eigenspace of $T^{\mathcal{U}_\times}_k$ defined in Eq.~\eqref{eq: k-fold-channel-symmetry-decomposition} scales with both $k$ and, crucially, the number of ways to decompose subspace along with their multiplicities. Consequently, we have to bound the spectral gaps for all possible cases. When $k=2$, the Wick projection reduces the formidable-looking procedure to three more accessible types: (i) $ S^{\lambda} \otimes S^{\mu} \otimes S^{\lambda} \otimes S^{\mu}$, (ii) $ S^{\lambda} \otimes S^{\mu} \otimes S^{\mu} \otimes S^{\lambda}$, and (iii) $ S^{\lambda} \otimes S^{\lambda} \otimes S^{\lambda} \otimes S^{\lambda}$, for $\lambda \neq \mu$. Note that, for simplicity, we ignore the multiplicities and $S^\lambda$ is used to denote the subspace in the decomposition of $\mathcal{H}$ with respect to either $\U(1)$ or $\SU(d)$ symmetry. We deal with these cases one by one.

As mentioned earlier, we bound the spectral gap by comparing with those of Cayley graphs over $S_n$ and their associated Cayley moment operators defined below. Let $\mathcal{T}$ be a generating set consisting of SWAPs in $S_n$. The problem of bounding the spectral gap of $\mathcal{G}(\mathcal{T})$ has been actively studied in spectral graph theory, and for some interesting cases, explicit bounds can be established~\cite{Aldous1992,CaputoProof2010}. In particular, after normalization, we have $\Delta(\mathcal{G}(\mathcal{T})) = \frac{2}{n-1}(1-\cos\frac{\pi}{n})$ for $\mathcal{T}$ consisting of nearest-neighbour SWAPs~\cite{Bacher1994}, and $\Delta(\mathcal{G}(\mathcal{T})) = \frac{2}{n-1}$ for all-to-all interactions~\cite{Flatto1985,Friedman2000} (also see Appendix~\ref{sec:Aldous} for more details). For simplicity, in the following, we only assume $\mathcal{T}$ including $n-1$ many nearest-neighbour SWAPs. Loosely speaking, type (i) and (ii) above are analogous to the product of $1$-designs, where the corresponding spectral gaps can be bounded easily with those of the Cayley graphs. This leads to $\Delta( T^{\mathcal{E}_{\CQA}}_{2, \times} \vert_{S^{\lambda} \otimes S^{\mu} \otimes S^{\lambda} \otimes S^{\mu}} ) = \Omega( \Delta(\mathcal{G}(\mathcal{T})) )$ and case (ii) follows similarly. The proofs can be found in Theorem \ref{thm:Case1&2}. 

To tackle type (iii), we first define $\mathcal{P}_{\ytableaushort{{} {}, {} {}}}$ for $(S^{\lambda})^{\otimes 4}$ and divide it into 4 invariant subspaces as we did earlier for a general Hilbert space $\mathcal{H}$. We define the \emph{Cayley moment operator}:
\begin{align}\label{eq: cayley-moment-operator}
		\begin{aligned}
			\hspace{-1mm} \Cay_{2, \times} & \hspace{-1mm} = \hspace{-1mm} \mathcal{P}^{\Wick}_{2, \times} \circ \frac{1}{8}(6 IIII +\tau \tau II +  II \tau \tau ) \circ \mathcal{P}^{\Wick}_{2, \times}. \hspace{-2mm} 
		\end{aligned}
\end{align}
where for brevity, we abuse the notation by denoting the restrictions of all super-operators to $(S^{\lambda})^{\otimes 4}$ by their original notations, and omit $\otimes$. For both $\SU(d)$ and $\U(1)$ symmetries, $\Cay_{2, \times}$ admits a unique unit eigenvalue within each $(S^{\lambda})^{\otimes 4}$. This is proved in Claim \ref{claim:Cayley-1} in Appendix~\ref{sec:detailsU(1)} using properties of Kronecker coefficients~\cite{Fulton1997,Sagan01} and Wick projections themselves. Furthermore, as a simple application of the Cauchy interlacing theorem~\cite{Horn2017}, we obtain $\Delta(\Cay_{2, \times}) = \Theta(\Delta(\mathcal{G}(\mathcal{T})))$ which legitimates the use of $\Cay_{2, \times}$ in comparing gaps. We sketch the proof strategy in the diagram on the right, where the arrows point towards larger spectral gaps and two-sided arrows means that both sides have the same scaling.

Recall that unit eigenvectors do not appear in the off-diagonal invariant subspaces $\Ima \mathcal{P}_{\ytableaushort{{} {*(black)}, {*(black)} {}}}$, so we compare by the largest eigenvalue $\lambda_1(T^{\mathcal{E}_{\CQA}}_{2, \times} {\ytableaushort{{} {*(black)}, {*(black)} {}}})$ in the top-left step and the comparison can be made by direct computation. The steps on the right leverage various techniques from the theory of Markov chains including: comparison via Dirichlet forms, defining induced Markov chains \cite{Levin2009} and comparing by congestion ratio of transition paths \cite{DiaconisComparison1993}. An extended proof outline can be found in Appendix~\ref{sec:sketches} with full details presented in Appendix~\ref{sec:details}. Combining these results for all invariant blocks, we can bound $\Delta(T^{\mathcal{E}_{\CQA}}_{2, \times})$ by those of the Cayley moment operators and the corresponding Cayley graphs. Using \eqref{eq:ConvergenceDepth}, we arrive at Theorem \ref{thm: CQA-second-largest-eigenvals}.
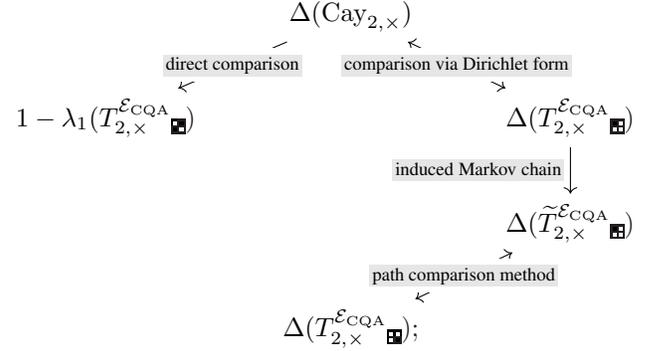
\[\begin{tikzcd}
	& {\Delta(\Cay_{2, \times})} \\
	{1-\lambda_1(T^{\mathcal{E}_{\CQA}}_{2, \times} {\ytableaushort{{} {*(black)}, {*(black)} {}}})} && {\Delta(T^{\mathcal{E}_{\CQA}}_{2, \times} {\sbA})} \\
	&& {\Delta(\widetilde{T}^{\mathcal{E}_{\CQA}}_{2, \times} {\sbA})} \\
	& {\Delta(T^{\mathcal{E}_{\CQA}}_{2, \times} {\sbD})}
	    \arrow["\text{\begin{tikzpicture}[baseline=(char.base)] \node[fill=gray!20, inner sep=1pt ] (char) {\text{direct comparison}}; \end{tikzpicture}}"{description}, from=1-2, to=2-1];
	\arrow["\text{\begin{tikzpicture}[baseline=(char.base)] \node[fill=gray!20, inner sep=1pt] (char) {\text{comparison via Dirichlet form}}; \end{tikzpicture}}"{description}, tail reversed, from=1-2, to=2-3]
	\arrow["\text{\begin{tikzpicture}[baseline=(char.base)] \node[fill=gray!20, inner sep=1pt] (char) {\text{induced Markov chain\vphantom{g}}}; \end{tikzpicture}}"', from=2-3, to=3-3]
	\arrow["\text{\begin{tikzpicture}[baseline=(char.base)] \node[fill=gray!20, inner sep=1pt] (char) {\text{path comparison method}}; \end{tikzpicture}}"{description}, tail reversed, from=4-2, to=3-3]
\end{tikzcd}\]


\section{Discussion}

By devising strategies radically different from those successful in proving convergence rates in the symmetry-free case, we resolve the long-standing open problem of efficient generation of pseudorandom ensembles under conservation laws by constructing explicit circuit models originated from the CQA ansatz which provably converge to symmetric $2$-designs in polynomial time. Our results pave the way for various physical and practical applications that rely on efficient constructions of unitary designs under conservation laws, such as covariant quantum error correction~\cite{preskill2021approxeastinknill,Hayden_2021,faist20,Woods2020continuousgroupsof,Zhou2021newperspectives,kong2022near,Yang22,Kubica21,liu2021approximate,liu2023quantum,SUd-k-Design2023Application}, geometric machine learning~\cite{Zheng2021SpeedingUL,ragone2022representation,LSKP23,Liu2022QNTK,Liu2023QNTK}, and understanding the role of symmetries in the physics of complex quantum systems, a critical topic in quantum many-body physics and gravity, from the perspectives of e.g.~thermodynamics~\cite{SUd-k-Design2023Application,majidy2023critical, Majidy2023review,chang2024deep} and scrambling~\cite{khemani2018operator, rakovszky2018diffusive, Huang2019OTOC,yoshida:soft,Nakata2023,Liu2024Mpemba}. 

We note that it is not evident how to extend our current analysis to designs of higher orders (extended discussions can be found in the Appendix). This extension would potentially be useful for understanding the complexity evolution of symmetric random circuits using arguments analogous to Ref.~\cite{brandao2021models}, which we deem an important avenue for future work. Furthermore, it is particularly worth noting that our constructions are friendly for experimental implementation, potentially enabling the aforementioned applications as well as the simulation of nonequilibrium many-body quantum dynamics and charged black holes, which is of great physical interest, in the near term.  


\section*{Acknowledgements}
ZL and HZ contribute equally in this work. ZL and ZWL are supported in part by a startup funding from YMSC, Tsinghua University, and NSFC under Grant No.~12475023.

\bibliography{circuit-symm}

\clearpage
\widetext
\appendix
\begingroup
\section*{Appendix}

\titleformat{\section}[block]{\large\bfseries\filcenter}{}{0pt}{}
\setcounter{tocdepth}{2} 

\startcontents[sections]
\titlecontents{section}[0pt]{\vspace{5mm}}{\thecontentslabel\hspace{1em}}{}{\titlerule*[1pc]{.}\contentspage}
\titlecontents{subsection}[1.5em]{}{\thecontentslabel\hspace{1em}}{}{\titlerule*[1pc]{.}\contentspage}
\printcontents[sections]{}{1}{\section*{}\vspace{-10mm}}

\endgroup


\section{Summary of notations}

\renewcommand{\arraystretch}{1.5} 

\begin{table}[H]
	\centering
	\caption{} 
	\begin{tabular}{c l }
		Notation & Definition  \\ [0.5ex]
		\hline
		$\lambda \vdash n$ & A partition of $n$ \\
		$S^\lambda$ & An $S_n$ irrep as a subspace of the Hilbert space of $n$ qudits \\ 
		$d_\lambda, m_\lambda$ & The dimension and multiplicity of $S^\lambda$ \\
		$p(n,d)$ & The number of inequivalent $S_n$ irreps of $n$ qudits \\ 
		$S^\mu$ & A $\U(1)$ charge sector \\ 
		$X_l$ & A YJM element  \\
		$\tau$ & An adjacent transposition/SWAP $(j,j+1)$ \\
		$T$ & A standard Young tableaux \\
		$\alpha_T$ & The content vector of the tableaux $T$ \\
		$r$ & The axial distance \\
		$\mathcal{U}_{\SU(d)}$ & The group of $\SU(d)$-symmetric unitaries \\
		$S\mathcal{U}_{\SU(d)}$ & The subgroup of $S\mathcal{U}_{\SU(d)}$ with trivial relative phases \\
		$\mathcal{V}_{r,\SU(d)}$ & The group generated by $r$-local $\SU(d)$-symmetric unitaries \\
		$\CQA_{\SU(d)}$ & The group generated by 2nd order YJM elements and $\tau_j$ \\
		$\CQA^{(k)}_{\SU(d)}$ & The group generated by $k$-th order YJM elements and $\tau_j$ \\
		$\mathcal{E}_{\CQA,\SU(d)}$ & The CQA random walk ensemble \\
		$T_{k,\SU(d)}^{\mathcal{U}_\times}$ & The $k$-th moment operator with respect to the group $\mathcal{U}_{\SU(d)}$ and its Haar measure \\
		$T_{k,\SU(d)}^{\CQA}$ & The $k$-th moment operator of the Haar measure of $\CQA_{\SU(d)}$ \\
		$T_{k,\SU(d)}^{\mathcal{E}_{\CQA}}$ & The $k$-th moment operator of the ensemble $\mathcal{E}_{\CQA,\SU(d)}$ \\
		$T_2^{\YJM}$ & The second moment operator defined by YJM elements \\
		$M_{2,\SU(d)}^{\mathcal{E}_{\CQA}}$ & The modified second moment operator of $T_{2,\SU(d)}^{\mathcal{E}_{\CQA}}$ \\
		$\Ima\mathcal{P}_{\sbA}$ & An invariant subspace defined by the projection $\mathcal{P}_{\sbA}$ \\
		$S_{\sbA}$ & The subspace in $\Ima\mathcal{P}_{\sbA}$ obtained by further projection through $T_2^{\YJM}$ \\
		$M^{\mathcal{E}_{\CQA}}_{2,\sbA}$ & The modified moment operator restricted to 	$S_{\sbA}$ \\
		$\Cay_{2,\SU(d)}$ & The Cayley moment operator \\
		$\Cay_{2,\SU(d),\Sym},\Cay_{2,\SU(d),\Asym}$ & Two invariant sub-blocks of the Cayley moment operator \\
		$\lambda_i(M)$ & The $i$-th eigenvalue of an operator $M$ \\
		$\Delta(M)$ & The spectral gap of an operator $M$ \\
		$(\mathcal{X},P,\pi)$ & A Markov chain with state space $\mathcal{X}$ and transition matrix $P$ \\
		$\pi_P$ & The stationary distribution of an irreducible reversible Markov chain \\
		$(\mathcal{S},\tilde{P})$ & An induced Markov chain of $P$ to a subspace $\mathcal{S} \subset \mathcal{X}$ \\
		$\mathcal{E}(P,f)$ & The Dirichlet form of a reversible $P$ and function $f$ defined on $\mathcal{X}$ \\
		$\gamma_{xy}$ & A directed path from state $x$ to $y$ in $\mathcal{X}$ \\
		$v_{ab}$ & A basis state from $S_{\sbA}$ \\
		$s_{ab}$ & A basis state from $S_{\sbD}$ \\
		\hline
	\end{tabular}
	\label{table:Notations}
\end{table}

We also have similar notations for the U$(1)$ case like U$(1)$-symmetric ensembles $\mathcal{E}_{\CQA, \U(1)}$ and moment operators $T_{2,\U(1)}^{\mathcal{E}_{\CQA}}$. When there is no ambiguity, these subscripts would be omitted or replaced by the symbol ``$\times$" for conciseness.


\section{Preliminaries}

Here we introduce various basic notions and facts about $S_n$ representation theory, $S_n$ Cayley graph theory, subspace decomposition and unitary $k$-designs under $\SU(d)$ and $\U(1)$ symmetries,  Markov process as well as our CQA model, to lay the foundation for later mathematical proofs. We also refer interested readers to Refs.~\cite{Fulton1997,Sagan01,Goodman2009,Tolli2009,Levin2009,Zheng2021SpeedingUL,SUd-k-Design2023} for more systematic presentations on these topics.

\subsection{Miscellaneous facts about $S_n$ representation theory}\label{sec:SnTheory}

Irreducible representations (irreps) of the symmetric group $S_n$  are in one-to-one correspondence with the so-called \emph{Young diagrams}. For instance, for $S_6$, the following two Young diagrams stand for the \emph{trivial representation} and the \emph{standard representation}, respectively:
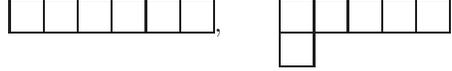
\begin{align*}
	\ytableausetup{boxsize=1.25em} \ydiagram{6}, \qquad \ydiagram{5,1},
\end{align*}
the direct sum of which is more familiar as the 6-dimensional \emph{defining representation} under which each $\sigma \in S_6$ permutes the components of vectors from $\mathbb{R}^6$. 

\begin{definition}\label{def:YoungDiagram}
	Formally, let $\lambda = (\lambda_1,\ldots,\lambda_r)$ be a collection of positive integers such that $\lambda_i \geq \lambda_{i+1}$ and $\sum_i \lambda_i = n$. Then $\lambda$ is called a \emph{partition} of the integer $n$, denoted by $\lambda \vdash n$. Obviously, $\lambda$ defines a Young diagram abstractly and the $S_n$ irrep corresponding to this Young diagram is always denoted as $S^\lambda$. The dimension of this irrep is given by the \emph{hook length formula}:
	\begin{align}
		\dim S^\lambda = \frac{n!}{\prod_{(x,y) \in \lambda} h_{x,y} },
	\end{align}
	where $(x,y)$ specifies a box from $\lambda$ by its \emph{row} and \emph{column numbers}, and the \emph{hook length} $h(x,y)$ counts the number of all boxes to the right of or below $(x,y)$ plus itself.
\end{definition}

Given an arbitrary $S_n$ irrep $S^\lambda$, there is a canonical way to label a basis, called the \emph{Gelfand--Tsetlin (GZ) basis} or the \emph{Young--Yamanouchi basis}, on the representation space, using \emph{standard Young tableau} $T$, which are defined by filling into each box of $\lambda$ a positive integer from $1,2,\ldots,n$ in an increasing order from left to right and top to bottom. We simply call it the \emph{Young basis} hereafter. For instance, the standard representation of $S_6$ mentioned above is 5-dimensional with 5 basis vectors labeled as
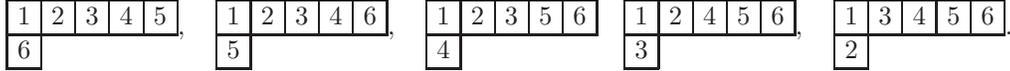
\begin{align*}
	\ytableaushort{12345, 6}, \quad  \ytableaushort{12346, 5}, \quad \ytableaushort{12356, 4} \quad \ytableaushort{12456, 3}, \quad \ytableaushort{13456, 2}.
\end{align*}
When we present our mathematical result in Section \ref{sec:sketches} \& \ref{sec:details}, the Young basis is intensively used for computations.

\begin{definition}\label{def:YJM}
	For $1 < k \leq n$, the \textit{Young--Jucys--Murphy element}, or YJM element for short, is defined as a (formal) sum of \emph{transpositions} or SWAPs 
	\begin{align}\label{eq:A-YJM}
		X_i = (1,i) + (2,i) + \cdots + (i-1,i).
	\end{align}
	We set $X_1 = 0$ as convention.
\end{definition} 

The YJM element is a central concept used in this work, which was developed by Young~\cite{Young1977}, Jucys~\cite{Jucys1974} and Murphy~\cite{Murphy1981} and later used by Okounkov and Vershik~\cite{Okounkov1996}. Under any $S_n$ representation, it may be more comprehensible to treat $X_i = (1,i) + (2,i) + \cdots + (i-1,i)$ as the sum of matrix representations of these transpositions, or we can say that the representation is extended to the \emph{group algebra} 
\begin{align}
	\mathbb{C}[S_n] = \left\{ \sum_i c_i \sigma_i; \sigma_i \in S_n \right\}
\end{align}
consisting of formal finite linear combinations of $S_n$ group elements. As we will introduce in Section \ref{sec:SchurWeyl}, $\mathbb{C}[S_n]$ has a one-to-one correspondence between the collection of all $\SU(d)$-symmetric matrices acting on the entire Hilbert space of $n$ qudits. 

Let us consider coordinate differences $x-y$ of boxes from a Young diagram $\lambda$. Given any standard tableau $T$ of $\lambda$, its \emph{content vector} is defined by rearranging them with respect to the order of boxes determined by the tableau. For instance, the content vectors of the above five standard tableaux are listed as follows:
\begin{align}
	(0,1,2,3,4,-1), \quad (0,1,2,3,-1,4), \quad (0,1,2,-1,3,4), \quad (0,1,-1,2,3,4), \quad (0,-1,1,2,3,4).
\end{align} 
An important feature of YJM elements is their special actions under the Young basis as revealed by content vectors: 
\begin{enumerate}
	\item They are diagonal matrices under the Young basis (even each single transposition $(i,j)$ from Eq.~\eqref{eq:A-YJM} may not be diagonal).
	
	\item The diagonal entry of $X_i$ under the Young basis vector $\ket{\alpha_T}$ corresponding to standard tableau $T$ is just the $i$-th component of the content vector.
\end{enumerate}
On the irrep $S^{(5,1)}$, 
\begin{align}\label{eq:YJMExample}
	\begin{aligned}
		& X_1 = \begin{pmatrix} 0 & 0 & 0 & 0 & 0 \\ 0 & 0 & 0 & 0 & 0 \\ 0 & 0 & 0 & 0 & 0 \\  0 & 0 & 0 & 0 & 0 \\	0 & 0 & 0 & 0 & 0 \end{pmatrix}, \quad\ \ \
		X_2 = \begin{pmatrix} 1 & 0 & 0 & 0 & 0 \\ 0 & 1 & 0 & 0 & 0 \\ 0 & 0 & 1 & 0 & 0 \\  0 & 0 & 0 & 1 & 0 \\	0 & 0 & 0 & 0 & -1 \end{pmatrix}, \quad 
		X_3 = \begin{pmatrix} 2 & 0 & 0 & 0 & 0 \\ 0 & 2 & 0 & 0 & 0 \\ 0 & 0 & 2 & 0 & 0 \\  0 & 0 & 0 & -1 & 0 \\	0 & 0 & 0 & 0 & 1 \end{pmatrix}, \\
		& X_4 = \begin{pmatrix} 3 & 0 & 0 & 0 & 0 \\ 0 & 3 & 0 & 0 & 0 \\ 0 & 0 & -1 & 0 & 0 \\  0 & 0 & 0 & 2 & 0 \\	0 & 0 & 0 & 0 & 2 \end{pmatrix}, \quad
		X_5 = \begin{pmatrix} 4 & 0 & 0 & 0 & 0 \\ 0 & -1 & 0 & 0 & 0 \\ 0 & 0 & 3 & 0 & 0 \\  0 & 0 & 0 & 3 & 0 \\	0 & 0 & 0 & 0 & 3 \end{pmatrix}, \quad
		X_6 = \begin{pmatrix} -1 & 0 & 0 & 0 & 0 \\ 0 & 4 & 0 & 0 & 0 \\ 0 & 0 & 4 & 0 & 0 \\  0 & 0 & 0 & 4 & 0 \\	0 & 0 & 0 & 0 & 4 \end{pmatrix}.
	\end{aligned}
\end{align}
In summary, Young basis vectors $\ket{\alpha_T}$, standard tableaux $T$ and content vectors $\alpha_T$ are in one-to-one correspondence and uniquely determine the matrix representations of YJM elements. We will introduce and apply other remarkable properties of YJM elements in Sections \ref{sec:sketches} and \ref{sec:details} when deriving our new results.

The matrix representation of each \emph{adjacent transposition} $(i,i+1)$ can be explicitly read off in the Young basis by the \emph{Young orthogonal form}: Let $r = \alpha_T(i+1) - \alpha_T(i)$ be the \emph{axial distance} and let $(i,i+1) \cdot T$ denote the tableau defined by exchanging integers $i, i +1$ from $T$. It is easy to check that as long as $r \neq \pm 1$, $(i,i+1) \cdot T$ is still a standard Young tableau. Then
\begin{align}\label{eq:YoungOrthogonal}
	(i,i+1) \ket{\alpha_T} = \frac{1}{r} \ket{\alpha_T} + \sqrt{1 - \frac{1}{r^2}} \ket{\alpha_{(i,i+1) \cdot T}}, \quad  (i,i+1) \ket{\alpha_{(i,i+1) \cdot T}} = \sqrt{1 - \frac{1}{r^2}} \ket{\alpha_T} - \frac{1}{r} \ket{\alpha_{(i,i+1) \cdot T}}.  
\end{align} 
On the irrep $S^{(5,1)}$, 
\begin{align}
	\renewcommand\arraystretch{1.25}
	\begin{aligned}
		& (1,2) = \begin{pmatrix} 1 & 0 & 0 & 0 & 0 \\ 0 & 1 & 0 & 0 & 0 \\ 0 & 0 & 1 & 0 & 0 \\  0 & 0 & 0 & 1 & 0 \\	0 & 0 & 0 & 0 & -1 \end{pmatrix}, \quad\quad 
		(2,3) = \begin{pmatrix} 1 & 0 & 0 & 0 & 0 \\ 0 & 1 & 0 & 0 & 0 \\ 0 & 0 & 1 & 0 & 0 \\  0 & 0 & 0 & -\frac{1}{2} & \frac{\sqrt{3}}{2} \\	0 & 0 & 0 & \frac{\sqrt{3}}{2} & \frac{1}{2} \end{pmatrix}, \quad 
		(3,4) = \begin{pmatrix} 1 & 0 & 0 & 0 & 0 \\ 0 & 1 & 0 & 0 & 0 \\ 0 & 0 & -\frac{1}{3} & \frac{2\sqrt{2}}{3} & 0 \\  0 & 0 & \frac{2\sqrt{2}}{3} & \frac{1}{3} & 0 \\	0 & 0 & 0 & 0 & 1 \end{pmatrix}, \\
		& (4,5) = \begin{pmatrix} 1 & 0 & 0 & 0 & 0 \\ 0 & -\frac{1}{4} & \frac{\sqrt{15}}{4} & 0 & 0 \\ 0 & \frac{\sqrt{15}}{4} & \frac{1}{4}  & 0 & 0 \\  0 & 0 & 0 & 1 & 0 \\ 0 & 0 & 0 & 0 & 1 \end{pmatrix}, \
		(5,6) = \begin{pmatrix}  -\frac{1}{5} & \frac{2\sqrt{6}}{5} & 0 & 0 & 0 \\  \frac{2\sqrt{6}}{5} & \frac{1}{5} & 0 & 0 & 0 \\ 0 & 0 & 1 & 0 & 0 \\  0 & 0 & 0 & 1 & 0 \\	0 & 0 & 0 & 0 & 1 \end{pmatrix}.
	\end{aligned}
\end{align}
Eq.~\eqref{eq:YoungOrthogonal} is of great significance to facilitate our mathematical proofs in Section \ref{sec:detailsSU(d)}. Besides, Young orthogonal forms are also useful in numerical computations on the matrix representations of general permutations $\sigma \in S_n$. There are classical or quantum \emph{$S_n$-fast Fourier transform} methods designed for such tasks~\cite{Clausen_1993,Maslen_1998}.


\begin{definition}\label{def:CycleType}
	We say that a permutation $\sigma \in S_n$ is of \emph{cycle type} $\lambda = (\lambda_1,\ldots,\lambda_r)$ where $\lambda \vdash n$ corresponds to a partition/Young diagram, if it is decomposed into cycles of lengths $\lambda_1,\ldots,\lambda_r$. 
\end{definition}

For instance, $\sigma = (134)(56) \in S_6$ is of cycle type $\lambda = (3,2,1)$. The trivial permutation $\operatorname{id} \in S_n$ is of type $(1,\ldots,1)$. Transpositions or SWAPs $(i,j)$ are just 2-cycles, so products of transpositions like $(i,j)(k,l) \cdots (s,t)$ are of type $(2,2,\ldots,2,1,\ldots,1)$. 

\begin{definition}\label{def:PartitionFunction}
	Let $p(n)$ denote the {number of partitions} of $n$. It equals the number of all inequivalent $S_n$ irreps, as well as the number of possible cycle types in $S_n$. Analogously, we define $p(n,d)$ as the {number of partitions of $n$ with at most $d$ parts}, i.e., the number of all Young diagrams of $n$ boxes with at most $d$ rows. 
\end{definition}

In the study of $\SU(d)$-symmetric random circuits, the quantity $p(n,d)$ equals the number of inequivalent $S_n$ irreps that arise from the decomposition of the $n$-qudit space (see Ref.~\cite{SUd-k-Design2023} and Section \ref{sec:SchurWeyl} for more details). Due to the celebrated work of Ramanujan and Hardy~\cite{Ramanujan1918} and Uspensky~\cite{Uspensky1920}, we have
\begin{align}\label{eq:Ramanujan}
	p(n) \sim \frac{e^{\sqrt{\pi^2 2n /3}}}{4n\sqrt{3}}, n \to \infty.
\end{align}
There has been further study on this~\cite{Rademacher1938,Erdos1942}, and various useful bounds on $p(n)$ have been found later, such as~\cite{Maroti2003, Wladimir2009}
\begin{align}
	\frac{e^{2\sqrt{n}}}{an} < p(n) < e^{b\sqrt{n}}.
\end{align}
If $d = 2$, $p(n,2) = \lfloor \frac{n}{2} \rfloor + 1$. However, there are no closed-form formulas for these partition functions in general. 


\begin{proposition}\label{prop:CenterBasis}
	Let $c_\mu \in \mathbb{C}[S]$ be the sum of all $\sigma \in S_n$ with cycle type $\mu$. Considering all possible Young diagrams of size $n$, the collection $\{ c_\mu \}_{\mu \vdash n}$ forms a basis for the \emph{center} $Z(\mathbb{C}[S_n])$ consisting of all elements that commute with $\mathbb{C}[S_n]$. 
\end{proposition}  

By definition, $c_\mu$ commutes with any $\sigma \in S_n$. By the Wedderburn theorem~\cite{Sagan01,Goodman2009}, its matrix representation, still denoted by $c_\mu$ for simplicity, under any $S_n$ irrep $S^\lambda$ is just a scalar. As a result, the representation of $Z(\mathbb{C}[S_n])$ consists of scalar matrices within any $S_n$ irrep. Being a basis of $Z(\mathbb{C}[S_n])$ means being a basis capable of spanning all scalar matrices, called \emph{relative phase factors} when we study $k$-design with the presence of $\SU(d)$ symmetry, respecting the direct sum $\bigoplus S^\lambda$ of all inequivalent $S_n$ irreps. A more comprehensive discussion can be found in Ref.~\cite{SUd-k-Design2023}. In Section \ref{sec:U(1)}, we would also discuss relative phase factors under $\U(1)$ symmetry.

Besides the basis $\{ c_\mu \}$ defined above, we still have the following two kinds of bases to span the center $Z(\mathbb{C}[S_n])$:

\begin{theorem}\label{thm:CenterBases}
	The following two collections also constitute bases for $Z(\mathbb{C}[S_n])$:
	\begin{enumerate}
		\item Consider the $S_n$ group character
		\begin{align}\label{eq:Character}
			\chi_\lambda (\sigma) = \tr_\lambda \sigma
		\end{align}
		defined by taking the trace of $\sigma \in S_n$ restricted to the irrep $S^\lambda$. Then
		\begin{align}\label{eq:OrthogonalProjection}
			\Pi_\mu \vcentcolon = \frac{\dim S^\mu}{n!} \sum_{\sigma \in S_n} \bar{\chi}_\mu(\sigma) \sigma
		\end{align}
		is a projection exclusively into the irrep $S^\mu$. The collection $\{\Pi_\mu \}$ is an orthonormal basis. 
		
		\item Consider the YJM elements $X_i$. For any $\mu = (\mu_1,\ldots,\mu_r) \vdash n$, we set
		\begin{align}
			X_\mu = \sum_{2 \leq i_1 \neq i_2 \neq \cdots \neq i_r \leq n} X_{i_1}^{\mu_1 - 1} X_{i_2}^{\mu_2 - 1} \cdots X_{i_r}^{\mu_r - 1}
		\end{align} 
		The collection $\{X_\mu\}$ is also a basis for $Z(\mathbb{C}[S_n])$~\cite{Jucys1974,Murphy1981,Tolli2009}.
	\end{enumerate}
\end{theorem}

Different center bases have different theoretical implications. In Ref.~\cite{SUd-k-Design2023}, the bases built through permutations of the same cycle types and YJM elements are used to study the necessary locality of unitary ensembles in forming $\SU(d)$-symmetric unitary $k$-designs for large $k$. In this work, we also employ the orthogonal projections $\Pi_\mu$ to define invariant subspaces for the moment operators (see discussions and examples after Definition \ref{def:CQAEnsemble}). It facilitates the use of Markov chain theory in proving our results in Section \ref{sec:sketches} \& \ref{sec:details}.


We introduce in the following some results related to $S_n$ characters, dominance ordering and dimension evaluation of $S_n$ irreps. 

\begin{proposition}\label{prop:Integer}
	For any $\sigma \in S_n$, $\chi_\lambda(\sigma) \in \mathbb{Z}$. That is, $S_n$ characters are integer-valued.
\end{proposition}

One can prove a variety of similar facts using Galois theory for general finite groups. For our purpose, we simply note that there is the so-called \emph{Young's natural representation} which is a \emph{non-unitary} representation of $S^\lambda$ under which each $\sigma$ is expressed as matrices with integer entries~\cite{Sagan01}. As trace is invariant under matrix similarity, $\chi_\lambda (\sigma) = \tr_\lambda \sigma \in \mathbb{Z}$ in general. 

It is also well known that permutations $\sigma$ and $\sigma'$ with the same cycle type are conjugate to each other, so $\chi_\lambda (\sigma) = \chi_\lambda (\sigma')$ and hence we only care about the value of $\chi_\lambda$ for a given cycle type $\mu$. The so-called \emph{Frobenius character formula}~\cite{Goodman2009} expresses $S_n$ character values as coefficients of a power series. The coefficients can be formally computed by, e.g., contour integrals using the residue theorem. However, closed-form formulas only exist for very few simple cases~\cite{Ingram1950,Roichman1996}. For instance, the characters for 2-cycles (transpositions/SWAPs) are
\begin{align}\label{eq:CharacterValue}
	& \frac{\chi_\lambda(i,j)}{\dim S^\lambda} = \frac{2}{n(n-1)} \sum_i \Big( \binom{\lambda_i}{2} - \binom{\lambda_i'}{2} \Big) = \frac{1}{n(n-1)} \sum_i \big[ (\lambda_i - 1)(\lambda_i - 1 + i) - i(i-1) \big], 
\end{align}
where $\lambda'$ denotes the conjugate of $\lambda$, e.g.,
\begin{align*}
	\lambda = \ytableausetup{boxsize=1.25em,aligntableaux=center} \ydiagram{4,2,1} \qquad \lambda' = \ydiagram{3,2,1,1}
\end{align*}
If $\lambda_i < 2$, the corresponding binomial coefficient is set to zero.


Let us relate the above character formula to some techniques involving YJM elements. Restricted to any irrep $S^\lambda$, we can associate the following invariants: 
\begin{align}
	P_l = \left(\sum_i X_i\right)^l, l = 1,2,\ldots
\end{align}
where $\sum_i X_i$ is the summation of all YJM elements. Let us check its matrix form under the Young basis $\{\ket{\alpha_T}\}$:
\begin{align}
	\left(\sum_i X_i\right) \ket{\alpha_T} = \sum_i \alpha_T(i) \ket{\alpha_T}
\end{align}
where the $\alpha_T$ are the content vectors. Obviously, for any fixed Young diagram $\lambda$, the sum of all components of any of its content vector $\alpha_T$ is simply equal to the sum of all coordinate differences, and we denote it as $\alpha_\lambda$. Then
\begin{align}\label{eq:P_l}
	P_l \ket{\alpha_T} = (\alpha_\lambda)^l \ket{\alpha_T} 
\end{align}
for all standard tableaux/Young basis vectors of the Young diagram $\lambda$.  

Let $\tr_\lambda$ denote the trace within $S^\lambda$ (it is just the $S_n$ character in Eq.~\eqref{eq:Character}). When $l = 1$, we note that
\begin{align}\label{eq:P_l2}
	\alpha_\lambda = \frac{\tr_\lambda (P_l) }{\dim S^\lambda} = \frac{\tr_\lambda \sum_i X_i}{\dim S^\lambda} = \frac{n(n-1)}{2} \frac{\tr_\lambda (i,j)}{\dim S^\lambda} = \frac{n(n-1)}{2} \frac{\chi_\lambda(i,j)}{\dim S^\lambda},
\end{align}
which gives another way to compute the character value of 2-cycles by summing all components from the content vector. The method using YJM elements and content vectors to express general $S_n$ characters can be found in Ref.~\cite{Lassalle2008}.


\begin{definition}\label{def:Dominance}
	Given two partitions $\lambda = (\lambda_i), \mu = (\mu_i) \vdash n$. We say that $\lambda$ \emph{dominates} $\mu$, denoted by $\lambda \unrhd \mu$, if for all $j > 0$, $\sum^j_i \lambda_i \geq \sum^j_i \mu_i$. 
\end{definition}

For instance, we have
\begin{align}
	(6) \unrhd (5,1) \unrhd (4,2) \unrhd (4,1^2) 
\end{align}
where $(4,1^2)$ is the abbreviation of $(4,1,1)$. The dominance relation is not \emph{totally ordered}, e.g., we cannot compare $(4,1^2)$ and $(3,3)$. However, in the case of qubits ($d = 2$), where we focus on two-row Young diagrams by Schur--Weyl duality introduced in the following, partitions $\lambda = (\lambda_1, \lambda_2)$ with $\lambda_1 \geq \lambda_2$ clearly give rise to a total ordering. 

\begin{lemma}\label{lemma: total-ordering2}
	For any two unequal partitions  $\lambda, \mu \vdash n$, if $\lambda \unrhd \mu$, then $\alpha_\lambda > \alpha_\mu$.
\end{lemma}
\begin{proof}
	We prove this lemma by induction. Suppose that the statement holds for $n-1$. Given unequal $\lambda, \mu \vdash n$ with $\lambda \unrhd \mu$, there should be some $i$ such that $\lambda_i > \mu_i$, where $\lambda_i, \mu_i$ are the lengths of $i$-th rows of $\lambda$ and $\mu$ respectively. If $\lambda_i > \lambda_{i+1}$ and $\mu_i > \mu_{i+1}$, then we discard the RHS boxes on the $i$-th rows of $\lambda$ and $\mu$. The resultant Young diagrams, denoted by $\lambda'$ and $\mu'$, still satisfy the relation $\lambda' \unrhd \mu'$. Then by induction hypothesis, $\alpha_{\lambda'} > \alpha_{\mu'}$. On the other hand, the content of the discarded box from $\lambda$ is larger than that from $\mu$ by definition, hence we conclude that $\alpha_\lambda > \alpha_\mu$.
	
	Suppose that $\lambda_i = \lambda_{i+1} = \cdots = \lambda_r$ or $\mu_i = \mu_{i+1} = \cdots = \mu_s$. Then we are only allowed to discard the RHS boxes of $\lambda_r,\mu_s$ to ensure that $\lambda'$ and $\mu'$ are well-defined Young diagrams. Even when $r \neq s$, the dominance relation still holds because $\lambda_r > \mu_s$. By the same argument as above, we complete the proof. 
\end{proof}

By Eq.~\eqref{eq:P_l2}, the above lemma says that $\frac{\chi_{\lambda}(1,2)}{\dim S^\lambda}$ is strictly increasing with respect to the dominance order of $\lambda$. Lots of counterexamples occur when this order fails to hold: e.g., $(3,3)$ and $(4,1^2)$.


Let us end this subsection with some explicit analysis on the dimension of $S_n$ irreps of two-row Young diagrams. In this circumstance, the hook length formula in.\ref{def:YoungDiagram} can be further simplified as~\cite{Ingram1950,Sagan01}
\begin{align}\label{eq:two-row-dim}
	\dim S^\lambda = d_\lambda = \binom{n}{r} - \binom{n}{r-1} = \frac{n - 2r +1}{n - r + 1}\binom{n}{r}.
\end{align}
For binomial coefficients, we have another two useful bounds (assume $n = 2m$)~\cite{Gallier2011}:
\begin{align}
	\frac{2^n}{\sqrt{\pi(\frac{n}{2} + \frac{1}{3})}} \leq \binom{n}{n/2} \leq \frac{2^n}{\sqrt{\pi(\frac{n}{2} + \frac{1}{4})}}, \quad e^{-(\frac{n}{2} - r)^2/(n-r+1)} \leq \binom{n}{r} \Big/ \binom{n}{\frac{n}{2}} \leq e^{-(\frac{n}{2} - r)^2/(n-r)}.
\end{align}
Therefore, the ratio of $\dim S^\lambda$ to the dimension of the entire Hilbert space is
\begin{align}\label{eq:dim-Comparision}
	\frac{1}{\sqrt{\pi(\frac{n}{2} + \frac{1}{3})}} \frac{n - 2r +1}{n - r + 1}\binom{n}{r} \Big/ \binom{n}{\frac{n}{2}} \leq \frac{\dim S^\lambda}{2^n} = \frac{1}{2^n} \frac{n - 2r +1}{n - r + 1}\binom{n}{r} \leq \frac{1}{\sqrt{\pi(\frac{n}{2} + \frac{1}{4})}} \frac{n - 2r +1}{n - r + 1}\binom{n}{r} \Big/ \binom{n}{\frac{n}{2}}.
\end{align}

A lower bound for general $d$-row Young diagrams is also useful~\cite{Mishchenko1996,Giambruno2015}: suppose the first row $\lambda_1, \lambda_1'$ of $\lambda$ and its conjugate is upper bounded by $\frac{n}{\alpha}$ for some $\alpha > 1$. Then 
\begin{align}\label{eq:dim-Comparision2}
	\dim S^\lambda \geq \frac{\alpha^n}{n^{d(d+2)/2}}.
\end{align}


\subsection{Spectral gap of $S_n$ Cayley graphs and Aldous' conjecture}\label{sec:Aldous}

\begin{definition}
	Let $G$ be a finite group with a subset $\mathcal{T}$ satisfying the following properties:
	\begin{enumerate}
		\item The identity element $e \notin \mathcal{T}$.
		
		\item Any element $s \in \mathcal{T}$ implies $s^{-1} \in \mathcal{T}$.
		
		\item The subset $\mathcal{T}$ generates $G$.
	\end{enumerate}
	Then we consider the graph $\mathcal{G}(\mathcal{T})$ defined using elements of $G$ as vertices. Pair of vertices $x,y$ are connected if $xy^{-1} \in \mathcal{T}$. With the these properties, $\mathcal{G}(\mathcal{T})$ is a simple $\vert \mathcal{T} \vert$-regular graph, called the \emph{Cayley graph} over $G$ with connection set $\mathcal{T}$. 
\end{definition}

It is straightforward to see the following  by the definition of Cayley graphs:

\begin{proposition}
	The adjacency matrix of $\mathcal{G}(\mathcal{T})$ is similar to the sum of right regular representation $\text{R}$ of group elements from $\mathcal{T}$, which can be further decomposed into block matrices acting on $G$ irreps $\rho_i$ with multiplicities identical to their dimension $d_i$:
	\begin{align}
		A_{\mathcal{G}(\mathcal{T})}
		\sim \sum_{s \in \mathcal{T}} \text{R}(s) 
		\sim  \bigoplus_{i} \Big( \sum_{s \in \mathcal{T}} \rho_i(s)  \otimes I_{d_i} \Big).
	\end{align}
\end{proposition}

Since Cayley graphs are regular, their adjacency matrices $A$ differ merely by scalar matrices $D = \vert \mathcal{T} \vert I$ from their graph Laplacians $L = \vert \mathcal{T} \vert I - A$. In the following, we refer to the \emph{spectral gap} of a Cayley graph by either the difference between the smallest and second smallest eigenvalues of $A$ or the difference between the largest and second largest eigenvalues of $L$.

\begin{theorem}\label{thm:CayleyGap}~\cite{Bacher1994,Flatto1985,Friedman2000}
	The following facts about the spectral gap of standard $S_n$ Cayley graphs hold:
	\begin{enumerate}
		\item Let $\mathcal{T} = \{(1,2),(2,3),...,(n-1,n)\}$. Then 
		\begin{align}
			\lambda_1( A_{\mathcal{G}(\mathcal{T})} ) = n-1, \quad \lambda_2( A_{\mathcal{G}(\mathcal{T})} ) = n-3 + 2\cos\frac{\pi}{n}, \quad \Delta(A_{\mathcal{G}(\mathcal{T})}) = 2 \Big(1 - \cos\frac{\pi}{n} \Big).
		\end{align}
		
		\item Let $\mathcal{T} = \{(1,n),(2,n),...,(n-1,n)\}$. Then 
		\begin{align}
			\lambda_1(A_{\mathcal{G}(\mathcal{T})}) = n-1, \quad \lambda_2( A_{\mathcal{G}(\mathcal{T})} ) = n-2, \quad \Delta(A_{\mathcal{G}(\mathcal{T})}) = 1.
		\end{align}
		
		\item Let $\mathcal{T} = \{(i,j); i < j\}$. Then 
		\begin{align}
			\lambda_1( A_{\mathcal{G}(\mathcal{T})}) = \frac{n(n-1)}{2}, \quad \lambda_2( A_{\mathcal{G}(\mathcal{T})} ) = \frac{(n-1)(n-2)}{2} - 1, \quad \Delta(A_{\mathcal{G}(\mathcal{T})}) = n.
		\end{align}
	\end{enumerate}
\end{theorem}

Given a generating set $\mathcal{T}$ of transpositions in $S_n$, the famous Aldous’ spectral gap conjecture, proposed in Ref.~\cite{Aldous1992} and proved in Ref.~\cite{CaputoProof2010}, says that the second largest eigenvalue of $A_{\mathcal{G}(\mathcal{T})}$ is achieved by the standard representation of $S_n$. Formally,

\begin{theorem}
	Let $\mathcal{T}$ be a generating set of transpositions in $S_n$. The following equivalent facts hold
	\begin{enumerate}
		\item $\lambda_2(A_{\mathcal{G}(\mathcal{T})})$ is achieved in $S^{(n-1,1)}$.
		
		\item Let $\mathcal{G}_\mathcal{T}$ be the graph with vertex set $[n]$ such that two vertices $i,j$ are joint if $(i,j) \in \mathcal{T}$. Then the spectral gap of $L_{\mathcal{G}_\mathcal{T}}$ equals that of $L_{\mathcal{G}(\mathcal{T})}$.
	\end{enumerate}
\end{theorem}

The equivalence can be seen by the fact that the direct sum of standard representation and trivial representation is isomorphic with the defining representation of $S_n$, as introduced in the very beginning of this Appendix. Previous examples surely obey this statement. Generalizations of Aldous’ conjecture to other families of Cayley graphs, like those not merely containing transpositions, are intensively studied in recent researches like Refs.~\cite{Cesi2016,Huang2019,Parzanchevski2020,Li2023} 

The point related to our work when study $\U(1)$ and $\SU(d)$-symmetric random circuit models is that in Section \ref{sec:sketches} \& \ref{sec:details}, we need to construct operators like
\begin{align}
	A = \sum_{\tau \in \mathcal{T} \subset S_n} \tau \otimes \tau \ \text{ and } \ B = \sum_{\tau \in \mathcal{T} \subset S_n} \frac{1}{2} ( \tau \otimes \tau \otimes I \otimes I +  I \otimes I \otimes \tau \otimes \tau),
\end{align}
where $\tau$ denotes a transposition/SWAP acting on the $n$-qudit system $\mathcal{H} = (\mathbb{C}^d)^{\otimes n}$. After all, $\tau \otimes \tau$ is still a representation of $S_n$ on $\mathcal{H}^{\otimes 2}$. Therefore,
\begin{align}\label{eq:CayleyTensor}
	\sum_{\tau \in \mathcal{T} \subset S_n} \tau \otimes \tau \sim \bigoplus_\lambda \Big( \sum_{\tau \in \mathcal{T}} \rho_\lambda(\tau) \otimes I_{m_\lambda} \Big).
\end{align}
Although it is generally unclear which kinds of $S_n$ irreps $S^\lambda$ as well as multiplicities $m_\lambda$ we shall have from the above decomposition, Aldous’ spectral gap theorem guarantees that second largest possible ones are bounded by those in Theorem \ref{thm:CayleyGap} when $\mathcal{T}$ is taken to be the corresponding sets, i.e.,
\begin{align}
	\lambda_2(A) \leq \lambda_2( A_{\mathcal{G}(\mathcal{T})} ). 
\end{align}
By definition, the second largest eigenvalue of $B$ is just $\lambda_2(B) = \frac{1}{2} ( 1 + \lambda_2(A) )$, which can also be bounded by Theorem \ref{thm:CayleyGap}.


\subsection{Schur--Weyl duality, space decomposition under $\SU(d)$ symmetry and CQA architecture}\label{sec:SchurWeyl}

We now provide a brief review on Schur--Weyl duality, which would provide the right perspective to study unitary gates and states subject to $\SU(d)$ symmetry. Then we introduce the $S_n$-Convolutional Quantum Alternating (CQA) group proposed in Refs.~\cite{Zheng2021SpeedingUL,SUd-k-Design2023}, which form exact unitary $k$-designs with $\SU(d)$ symmetry. With these preparations, we prove in Section \ref{sec:detailsSU(d)} a polynomial convergence time with respect to the number $n$ of qudits for the induced CQA ensemble in the generation of approximate $\SU(d)$-symmetric unitary $2$-design. 

For quantum systems, there is a discrete set of translations corresponding to permuting the qudits as well as a continuous notion of translation corresponding to spatial rotations by elements of $\SU(d)$. To be precise, let $\mathbb{C}^d$ be a $d$-dimensional complex Hilbert space with orthonormal basis $\{e_1,\ldots,e_d\}$. The $n$-qudit Hilbert space $\mathcal{H} = (\mathbb{C}^d)^{\otimes n}$ admits two natural representations: the \textit{tensor product representation} $\rho_{\SU(d)}$ of $\SU(d)$ acting as
\begin{align}\label{eq:SUd}
	\rho_{\SU(d)}(g) (e_{i_1} \otimes \cdots \otimes e_{i_n}) \vcentcolon = g \cdot e_{i_1} \otimes \cdots \otimes g \cdot e_{i_n}, 
\end{align}
where $g \cdot e_{i_k}$ is given by the fundamental representation of $\SU(d)$, and the \textit{permutation representation} $\rho_{S_n}$ of $S_n$ acting as
\begin{align}\label{eq:Sn}
	\rho_{S_n}(\sigma) (e_{i_1} \otimes \cdots \otimes e_{i_n}) \vcentcolon =
	e_{i_{\sigma^{-1}(1)}} \otimes \cdots \otimes e_{i_{\sigma^{-1}(n)}}.
\end{align}
Schur--Weyl duality states that the action of $\SU(d)$ and $S_n$ on $V^{\otimes n}$ jointly decompose the space into irreducible representations of both groups in the form 
\begin{align}\label{eq:A-SchurWeyl}
	\mathcal{H} = \bigoplus_\lambda W_\lambda \otimes S^\lambda.
\end{align}
Again, $\lambda$ denotes a Young diagram. In this setting, it corresponds not only to a unique $S_n$ irrep $S^\lambda$, but also an $\SU(d)$ irrep $W_\lambda$ \cite{Goodman2009,Tolli2009}. It should be noted that within an $n$-qudit system, only irreps corresponding to $\lambda$ ranging over {Young diagrams of size $n$ with at most $d$ rows} can be found in the decomposition. 

We denote by $\operatorname{1}_{m_{\SU(d),\mu}} \cong S^\mu, \operatorname{1}_{ m_{S_n,\lambda}} \cong W_\lambda$ the multiplicity spaces of $\SU(d)$ and $S_n$ irreps, respectively. Then,
\begin{align}
	& \rho_{\SU(d)} \cong \bigoplus_\mu W_\mu \otimes \operatorname{1}_{m_{\SU(d),\mu}}, \quad  \rho_{S_n} \cong \bigoplus_\lambda \operatorname{1}_{ m_{S_n,\lambda}} \otimes S^\lambda,
\end{align}
where $m_{\SU(d),\mu} = \dim S^\mu$ and $m_{S_n,\lambda} = \dim W_\lambda$.

An operator $A$ acting on the system being $\SU(d)$-symmetric or invariant means that
\begin{align}
	\rho_{\SU(d)}(g) A = A \rho_{\SU(d)}(g) \text{ or } g^{\otimes n} A = A g^{\otimes n}.
\end{align}
One can check by Eqs.~\eqref{eq:SUd} and \eqref{eq:Sn} that these permutation actions clearly commute with $g^{\otimes n}$. Furthermore, Schur--Weyl duality as well as the double commutant theorem~\cite{Goodman2009,Tolli2009} confirms that $\SU(d)$-symmetric operators are exactly built from permutations in the symmetric group $S_n$. That is, they can be expressed as linear combinations, like $\sum c_i \sigma_i$, of permutations.

It is a conventional practice in physics to decompose the entire space into $\SU(d)$ irreps. Quantum states living in these subspaces are actually permutation-invariant or $S_n$-symmetric. Since our focus is on quantum circuits with $\SU(d)$ symmetry, we should decompose the entire Hilbert space with respect to $S_n$ irreps (see Refs.~\cite{Tolli2009,krovi,Zheng2021SpeedingUL,SUd-k-Design2023} for more details). As a reminder, although the entire Hilbert space is decomposed into smaller subspaces, one should not expect related problems, such as finding the ground state energy of $\SU(d)$-symmetric Hamiltonians and constructing a $\SU(d)$-symmetric random quantum circuits, to become easier. There are two general reasons: 
\begin{enumerate}
	\item There are various irreducible $S_n$ irreducible representations of the decomposition to deal with, and the total number is $p(n,d)$, which scales at most super-polynomially (see \eqref{eq:Ramanujan}) in $n$ and has no closed-form formula for evaluation. 
	
	\item Even for qubits with $d = 2$, using the hook length formula from Definition \ref{def:YoungDiagram}, we know that (cf.~Eq.~\eqref{eq:two-row-dim})
	\begin{align}
		\dim S^{(m,m)} = \frac{(2m)!}{(m+1)!\, m!} = \frac{2^m}{m+1} \prod_{k=1}^m \frac{2k-1}{k} > \frac{2^m}{m+1} 
	\end{align}
	for the $S_n$ irrep of Young diagram $\lambda = (m,m)$ on a $2m$-qubit system. One can find other examples with exponentially large subspaces respecting the $\SU(d)$ symmetry~\cite{Giambruno2015}, which still cause difficulties when approaching the problem.
\end{enumerate}


We now introduce the mathematical definition of $S_n$-CQA group and the induced ensemble:

\begin{definition}\label{def:CQA}
	Let us consider the following Hamiltonians
	\begin{align}\label{eq:CQAGenerators}
		H_S = \sum_{j = 1}^{n-1} (j,j+1), \quad H_{\YJM} = \sum_{k,l} \beta_{kl} X_k X_l,
	\end{align}
	where $H_S$ is just the summation of adjacent transpositions and $H_{\YJM}$ is a linear combination of YJM elements up to second-order products with real parameters $\beta_{kl}$. Then
	\begin{align}
		\CQA_{\SU(d)} \vcentcolon = \Big\langle \exp(-i \sum_{k,l} \beta_{kl} X_k X_l), \quad \exp(-i\gamma H_S) \Big\rangle.
	\end{align}
	The subscript $\SU(d)$ is necessary because later we have similar definitions for the case under $\U(1)$ symmetry.
\end{definition}

As a reminder, supposing we set $X_1$ to be the identity operator,  the collection $\{X_k X_l \}$ automatically includes both first and second-order products of YJM elements, which is what we require in the above definition. However, to be consistent with the fact that components of content vectors (see Section \ref{sec:SnTheory}) form the spectra of YJM elements, $X_1$ should be zero. Regardless of these details, for brevity, we slightly abuse the notation $X_k X_l$ for YJM elements up to second-order products. One can also set $k \leq l$ in the above definition because YJM elements are commutative with each other. 

Importantly, CQA is contained in the \emph{group of $\SU(d)$-symmetric unitaries}. To define this group, let $\U(S^\lambda)$ denote the unitary group acting on the representation space $S^\lambda$, i.e., $\U(S^\lambda) \cong \U(\dim S^\lambda)$. A typical element $g$ from the group of $\SU(d)$-symmetric unitaries is then a collection of unitaries:
\begin{align}\label{eq:GroupElements}
	g = \bigoplus_\lambda U_\lambda^{\oplus m_{S_n, \lambda}},
\end{align}
where $U_\lambda \in \U(S^\lambda)$ and $\lambda$ range over all Young diagrams of size $n$ with at most $d$ rows. For simplicity, we omit the multiplicities and denote this group by $\mathcal{U}_{\SU(d)}$. On the other hand, by restricting the phase factors to be 1 on each $S^\lambda$, we obtain the special unitary group $\SU(S^\lambda)$ as well as $S\mathcal{U}_{\SU(d)} = \SU(\bigoplus_\lambda S^\lambda)$ consisting of $\SU(d)$-symmetric unitaries with unit determinant on each $S_n$ irrep block, i.e., unitaries with trivial relative phase factors with respect to each irrep sector. 

Let $\mathcal{V}_{r,\SU(d)}$ be the group generated by $\SU(d)$-symmetric $r$-local unitaries. There are special algebraic properties of $S_n$ and $\SU(d)$ irreps that arise from the subspace decomposition of the $n$-qudit system with $d \geq 3$~\cite{Biedenharn1,Biedenharn2,Marin1,Marin2} which lead to the fact that $S\mathcal{U}_{\SU(d)} \nsubseteqq \mathcal{V}_{2,\SU(d)}$~\cite{MarvianNature,MarvianSU2}, namely, $2$-local $\SU(d)$-symmetric unitaries are unable to even generate $2$-designs on general qudits~\cite{MarvianSUd}. Then, it is demonstrated in Ref.~\cite{Zheng2021SpeedingUL} that
\begin{align}\label{eq:universalitySU(d)}
	S\mathcal{U}_{\SU(d)} \subsetneqq \CQA_{\SU(d)} \subsetneqq \mathcal{V}_{4,\SU(d)} \subsetneqq \mathcal{U}_{\SU(d)}.
\end{align}
This is followed by the work of Ref.~\cite{SUd-k-Design2023}, which proves that
\begin{theorem}
	For an $n$-qudit system with $n \geq 9$ and $d < n$, the group $\CQA_{\SU(d)}$, as well as $\mathcal{V}_{4,\SU(d)}$, forms exact $\SU(d)$-symmetric unitary $k$-designs for all $k < n(n-3)/2$. 
\end{theorem}
A more recent work Ref.~\cite{Marvian3local} verifies that $S\mathcal{U}_{\SU(d)} \subsetneqq \mathcal{V}_{3,\SU(d)}$, which optimizes the smallest necessary locality to cover $S\mathcal{U}_{\SU(d)}$ to be $3$. 


\subsection{Permutation modules and space decomposition under $\U(1)$ symmetry}\label{sec:Tabloid}

We now discuss $\U(1)$ symmetry governed by the particle number conservation law on $n$ qubits $(d = 2)$. Mathematically, the symmetry requires
\begin{align}
	\Big[ U, \sum_i Z_i \Big] 0 \Leftrightarrow \Big[U, \sum_i \frac{Z_i + I}{2} \Big] = 0
\end{align}
for $U \in \U(2^n)$. We denote by $\mathcal{U}_{\U(1)}$ the group of all such $\U(1)$-symmetric unitaries.

It is well known that the decomposition of $n$-qubit space into the direct sum of $\U(1)$ charge sectors is given by 
\begin{align}\label{eq:SpaceDecompositionU(1)}
	(\mathbb{C}^2)^{\otimes n} \cong \bigoplus_{0 \leq r \leq n} \text{span}\{ \ket{i_1,...,i_n}; i = 0 \text{ or } 1, r \text{ many indices are selected to be } 0 \}  = \bigoplus_{\mu = (n-r,r)} S^\mu.
\end{align}
It should be noted that the superscript $\mu = (n-r, r)$ here may not refer to a well-defined Young diagram because $r$ can be larger than $\lfloor n/2 \rfloor$. This  causes no essential trouble, and we would explain in detail when any possible ambiguity occurs. As a reminder, we always denote by $S^\lambda$ an $S_n$ irrep when studying the $\SU(d)$ symmetry. In some cases, we may slightly abuse the notations $\mu$ and $\lambda$ for convenience. For example in Section \ref{sec:sketches}, many formulas and statements on random circuits under $\SU(d)$ and $\U(1)$ symmetries share similar mathematical forms. We may even omit super/subscripts containing $S^\lambda$ or $S^\mu$ for brevity when no confusion can arise.

Let $\U(S^\mu)$ denote the unitary group acting on the charge sector $S^\mu$. A typical element $g$ from the group of $\U(1)$-symmetric unitaries is then a collection of (multiplicity-free) unitaries:
\begin{align}
	g = \bigoplus_\mu U_\mu = \bigoplus_{0 \leq r \leq n} U_{(n-r, r)} \in \mathcal{U}_{\U(1)}.
\end{align}

The space decomposition under $\U(1)$ symmetry is more apparent and straightforward than the $\SU(d)$ case which demands a nontrivial transformation through Schur--Weyl duality. Even so, these $\U(1)$ charge sectors are still tightly related to the $S_n$ representation theory which we need to take into account in Section \ref{sec:U(1)}. We now provide a brief introduction. To begin with, let us recall the definition of standard Young tableaux in Section \ref{sec:SnTheory}. Suppose we discard the requirement that integers should be filled in the increasing order from left to right and from top to bottom in each row and column of the diagram. An arbitrary filling defines a \emph{Young tableau} $Y$.

\begin{definition}
	Two Young tableaux $Y_1,Y_2$ are \emph{row equivalent} if they are of the same shape $\lambda$ and the corresponding rows contain the same collection of integers. The row equivalence class $[Y]$ of a Young tableau $Y$ is called a \emph{Young tabloid}. 
\end{definition}

For any permutation $\sigma \in S_n$, let $\sigma \cdot Y$ be the Young tableau obtained by permuting integers filled in $Y$ through $\sigma$. Then, it is easy to check that $\sigma \cdot [Y] = [\sigma \cdot Y]$ is a well-defined group action. Consequently, given any Young diagram $\mu$, let $M^\mu$ be the vector space spanned by all tabloids of this shape. The action induces a representation of $S_n$ on $M^\mu$ called \emph{permutation module}.
 
\begin{theorem}
	The following facts hold for permutation modules:
	\begin{enumerate}
		\item Let $\mu = (n-r,r)$ with $r \leq \lfloor n/2 \rfloor$, the $\U(1)$ charge sector $S^\mu$ is isomorphic with the permutation module $M^\mu$~\cite{Goodman2009}.
		
		\item Any permutation module $M^\mu$ can be further decomposed into $S_n$ irreps $S^\lambda$ as long as $\lambda$ dominates $\mu$:
		\begin{align}
			M^\mu = \bigoplus_{\lambda \unrhd \mu} \mathrm{1}_{\kappa_{\mu\lambda}} \otimes S^\lambda,
		\end{align}
		where $\kappa_{\mu\lambda}$ is the \emph{Kostka number}~\cite{Sagan01}.
	\end{enumerate}
\end{theorem}

When $\mu = (n-r,r)$, Definition \ref{def:Dominance} indicates that only when $\lambda = (n-s,s)$ with $s \leq r$, $\lambda \unrhd \mu$. In these cases, $\kappa_{\mu\lambda}$ trivially equals one. For our analysis in Section \ref{sec:detailsU(1)}, it is necessary to know the above decomposition. Let us consider \eqref{eq:CayleyTensor} by restricting the action of $\tau$ to a $\U(1)$ charge sector $M^\mu$ from the entire Hilbert space. Then the action of $\tau \otimes \tau$ is decomposed into
\begin{align}
	M^\mu \otimes M^\mu \cong \bigoplus_{\lambda, \lambda' \unrhd \mu}  S^\lambda \otimes S^{\lambda'}. 
\end{align}
Obviously, each of these tensor product $S^\lambda \otimes S^{\lambda'}$ can be further decomposed into $S_n$ irreps. According to basic character theory~\cite{Fulton1997,Sagan01}, only when $\lambda = \lambda'$ the decomposition admits trivial irrep. Together with our earlier discussion, we find exactly $r+1$ multiples of trivial representation and hence the multiplicity of unit eigenvalue of Eq.~\eqref{eq:CayleyTensor} is $r+1$. This is subtle point that need taking into account when presenting our result. More details is discussed in Section \ref{sec:detailsU(1)}. 


\subsection{Unitary $k$-designs under continuous symmetry and invariant subspaces of moment operators}\label{sec:kDesigns}

\ytableausetup{boxsize = 2pt}
We first provide basic definitions for unitary $k$-designs under continuous symmetry~\cite{SUd-k-Design2023} for self-containment. Then we introduce the unitary ensembles studied in this work and elucidate the concept of invariant subspaces, which facilitates our mathematical proofs later. 

\begin{definition}
	Given a compact group $G$ with Haar measure $\mu$ and a unitary representation $\rho$ on the concerned Hilbert space $\mathcal{H}$, for any operator $M \in \operatorname{End}(\mathcal{H}^{\otimes k})$, the \emph{$k$-fold channel} is given by twirling over the Haar measure $\mu$ over $G$ acting on $M$: 
	\begin{align}\label{eq: tpe}
		T_k^G(M) = \int_G d\mu(g) \rho^{\otimes k}(g) M (\rho^{\otimes k}(g))^\dagger = \int_G dU U^{\otimes k} M U^{\dagger \otimes k},
	\end{align}
	where we denote the matrix representations of group elements simply by $U$ on the RHS of the above equation. Despite its integral form, $T_k^G(\cdot)$ is merely a linear map acting on $\operatorname{End}(\mathcal{H}^{\otimes k})$ and can be reformulated as the \emph{$k$-th moment (super)-operator}:
	\begin{align}
		T_k^G = \int_G U^{\otimes k } \otimes \bar{U}^{ \otimes k} dU.
	\end{align}
	We can analogously define $T_k^{\mathcal{E}}$ for an arbitrary ensemble $\mathcal{E}$ by replacing  $G$ by  $\mathcal{E}$, which provides a basis for the study of (approximate) $k$-designs.
\end{definition}

The representation space of $G$ can always be written as a direct sum 
\begin{align}\label{eq:SpaceDecomposition}
	\mathcal{H} \cong \bigoplus_i V_i \otimes \mathrm{1}_{m_i}
\end{align}
of subspaces with multiplicities of the Hilbert space, which can characterize quantum systems obeying certain continuous symmetries or conservation laws. For $\SU(d)$ symmetry, $G = \mathcal{U}_{\SU(d)}$ is defined in Section \ref{sec:SchurWeyl} and the above decomposition is given by Schur-Weyl duality. For $\U(1)$ symmetry, $G = \mathcal{U}_{\U(1)}$ is defined in Section \ref{sec:Tabloid} and the above decomposition is obtained by counting Hamming weights of the computational basis.  

In later contexts, when we write $T_k^G$ for certain compact group $G$, the integral is automatically taken over the Haar measure. The bi-invariance of Haar measure implies that $T_k^G$ is a \emph{projector} from $\operatorname{End}(\mathcal{H}^{\otimes k})$ onto its the \emph{commutant algebra} 
\begin{align}
	\Comm_k(G) \vcentcolon = \{M \in \operatorname{End}(\mathcal{H}^{\otimes k}); U^{\otimes k} M = M U^{\otimes k} \}, 
\end{align}
i.e., the subspace of all operators that commute with the tensor product representation $\rho^{\otimes k}$ of $G$. Eigenvalues of $T_k^G$ are either 0 or 1. See also Refs.~\cite{Dankert2026PRA,Gross2006,HarrowTEP08,HarrowTEP09} for more details.

To formally discuss convergence properties, we also need notions of approximate designs. Here we use the standard strong notion based on $\epsilon$-approximation in terms of complete positivity:
\begin{definition}\label{def:AppDesign}
	Given a compact group $G$, an ensemble $\mathcal{E}$ of unitaries is called an \emph{$\epsilon$-approximate unitary $k$-design with respect to $G$} if the following matrix inequality holds in the sense of complete positivity (that is, $A \leq_{\mathrm{cp}} B$ means $B - A$ is completely positive):
	\begin{align}
		(1 - \epsilon) T_k^G \leq_{\mathrm{cp}} T_k^{\mathcal{E}} \leq_{\mathrm{cp}} (1 + \epsilon) T_k^G. 
	\end{align} 
	We denote by $c_{\mathrm{cp}}(\mathcal{E}, k)$ the smallest constant $\epsilon$ that satisfies the above inequalities.
\end{definition}

Let $N = d^n$ be the total dimension of the Hilbert space $\mathcal{H}$ of $n$ qudits. It turns out that
\begin{align}
	c_{\mathrm{cp}}(\mathcal{E}, k) \leq N^{2k} \Vert T_k^{\mathcal{E}} - T_k^G \Vert_{2 \to 2},
\end{align}
where $\Vert T_k^{\mathcal{E}} - T_k^G \Vert_{2 \to 2}$ is the induced 2-norm on the super-operator~\cite{harrow2016local,SUd-k-Design2023,Haferkamp2024linear}. Viewing super-operators $T_k^{\mathcal{E}}, T_k^G$ as ordinary operators, the induced $2$-norm is exactly the infinity norm. When the operator is Hermitian and positive semidefinite, it simply equals the largest eigenvalue of $T_k^{\mathcal{E}} - T_k^G$. From this perspective, one necessary condition for converging to  unitary $k$-designs with respect to $G$ is that $\Vert T_k^{\mathcal{E}} - T_k^G \Vert_{2 \to 2} < 1$ because this ensures that (note that $T_k^{\mathcal{E}}, T_k^G$ are commutative and simultaneously diagonalizable due to the bi-invariance of Haar measure) for arbitrarily small $\epsilon$, there exists a sufficiently large $p$, which can be interpreted as the depth of the random circuit or evolution time, such that
\begin{align}
	\Vert (T_k^{\mathcal{E}})^p - T_k^G \Vert_{2 \to 2} \leq \epsilon.
\end{align}

With the assumption that $T_k^{\mathcal{E}}$ is Hermitian and positive semidefinite, $\Vert T_k^{\mathcal{E}} - T_k^G \Vert_{2 \to 2} < 1$ further indicates that
\begin{enumerate}
	\item The unit eigenspace of $T_k^{\mathcal{E}}$ is identical to that of $T_k^G$, which is just $\Comm_k(G)$.
	
	\item The second largest eigenvalue of $T_k^{\mathcal{E}}$, $\lambda_2(T_k^{\mathcal{E}})$,  is strictly smaller than one. When the circuit depth satisfies
	\begin{align}
		 p\geq\frac{1}{1 -\lambda_2} \log \frac{N^{2k}}{\epsilon} = \frac{1}{1 -\lambda_2} \Big( 2kn\log d + \log \frac{1}{\epsilon} \Big),
	\end{align}
	the ensemble forms approximate $k$-designs within precision $\epsilon$~\cite{HarrowTEP08,HarrowTEP09,harrow2023approximate,Schuster2024lowDepth}.
\end{enumerate}
It is proved in Ref.~\cite{SUd-k-Design2023} that the CQA ensemble defined in the following satisfies the first condition. Our present work mainly focus on the second one to bound the second largest eigenvalue or, equivalently, the \emph{spectral gap} $1-\lambda_2$ and verify the polynomial efficiency of the ensemble in approximating $k$-designs under $\SU(d)$ and $\U(1)$ symmetries. 


Motivated from Definition \ref{def:CQA} of the CQA group, we define  CQA ensembles as follows:

\begin{definition}[CQA random walk ensemble]\label{def:CQAEnsemble}
	At each step of the random circuit, we uniformly sample an index $j \in [n]$ and parameters $t, \beta_{kl}, \beta_{kl}' \in [0, 2\pi]$.
	\begin{enumerate}
		\item Under $\SU(d)$ symmetry, we implement 
		\begin{align}
			& \exp(-i \sum_{k, l} \beta_{kl} X_k X_l) \exp(-it (j,j+1) ) \exp(-i \sum_{k, l} \beta_{kl}' X_k X_l) 
		\end{align}
		on the qudits. The induced ensemble is denoted by $\mathcal{E}_{\CQA,\SU(1)}$.
		
		\item Under $\U(1)$ symmetry, we implement
		\begin{align}
			\exp(-i \sum_{k, l} \beta_{kl} Z_k Z_l) \exp(-it (j,j+1) ) \exp(-i \sum_{k, l} \beta_{kl}' Z_k Z_l)
		\end{align}
		on the qubits with $Z_k Z_l$ being products of Pauli-$Z$ matrices at site $k$ and $l$. The induced ensemble is denoted by $\mathcal{E}_{\CQA,\U(1)}$.
	\end{enumerate}
\end{definition}

It should be noted that the YJM elements as well as Pauli-$Z$ matrices mutually commute, thus during any implementation
\begin{align}
	\exp(-i \sum_{k, l} \beta_{kl} X_k X_l) = \prod_{k,l} \exp(-i \beta_{kl} X_k X_l), \quad \exp(-i \sum_{k, l} \beta_{kl} Z_k Z_l) = \prod_{k,l} \exp(-i \beta_{kl} Z_k Z_l).
\end{align}
An analogous implementation in the situation without symmetry
\begin{align}
	\exp(-i \sum_{k, l} \beta_{kl} Z_k Z_l) \exp(-it X_i ) \exp(-i \sum_{k, l} \beta_{kl}' Z_k Z_l),
\end{align}
where $X_i$ here exclusively refers to the Pauli-$X$ matrix at site $i$ recovers the scheme of Quantum Approximate Optimization Algorithm (QAOA)~\cite{Farhi_2014,Lloyd_2018,Morales_2020,farhi2019Supremacy}.


In contrast to the case without symmetry, one notable challenge that arises when evaluating the second largest eigenvalue of the $k$-th moment operator under continuous symmetry is that we have to handle $2k$-fold tensor products of distinct subspaces from Eq.~\eqref{eq:SpaceDecomposition}. We employ helpful techniques from representation theory for analysis, which will be introduced in more detail in Section \ref{sec:detailsT2CQA}. As a warm-up, let $\mathcal{S}$ be any fixed subspace from the decomposition. We now take a close look at the restriction of $T_k^G, T_k^{\mathcal{E}}$, with $\mathcal{E}$ being any ensemble respecting the symmetry associated with $G$, to the homogeneous tensor product $\operatorname{End}(\mathcal{S}^{\otimes k}) \cong \mathcal{S}^{\otimes 2k}$. Note that the case without symmetry is trivially included by assuming $\mathcal{H} = \mathcal{S}$.
By Schur--Weyl duality, $\Comm_k(G)$ is spanned by permutations $S_k$ acting on $\mathcal{S}^{\otimes k}$. Formally, we have 
\begin{align}
	\mathcal{S}^{\otimes k} \cong S^\nu \otimes \mathrm{1}_{m_\nu}, \quad \Pi_\nu = I_{S^\nu} \otimes I_{m_{\nu}}, 
\end{align}
where $\Pi_\nu$ is the projection onto one $S_k$ irrep with multiplicity labeled by $\nu$. It is immediate to see by Theorem \ref{thm:CenterBases} that
\begin{align}
	\Pi_\nu = \frac{\dim S^\nu}{k!} \sum_{\sigma \in S_k} \bar{\chi}_\nu(\sigma) \sigma \in \Comm_k(G).
\end{align}
Let $\ket{\nu},\ket{\zeta}$ be arbitrary unit eigenstates of $\Pi_\nu,\Pi_\zeta$, the following identity holds when integrating over the ensemble $\mathcal{E}$:
\begin{align}\label{eq:Invariant}
	\int_{\mathcal{E}} U^{\otimes k} \ket{\nu} \bra{\zeta} U^{\dag \otimes k} dU = \int_{\mathcal{E}} U^{\otimes k} \Pi_{\nu} \ket{\nu} \bra{\zeta} \Pi_{\zeta} U^{\dag \otimes k} dU = \Pi_{\nu} \Big( \int_{\mathcal{E}} U^{\otimes k} \ket{\nu} \bra{\zeta} U^{\dag \otimes k} dU \Big) \Pi_{\zeta}.
\end{align}
This identity implies that $\ket{\nu} \bra{\zeta}$ are invariant under the action of $T_k^{\mathcal{E}} \vert_{\mathcal{S}^{\otimes 2k}}$. Equivalently, the image of the (super)-operator
\begin{align}
	\mathcal{P}_{\nu\zeta}(\rho) \vcentcolon = \Pi_\nu \rho \Pi_\zeta
\end{align}
defines an invariant subspace for $T_k^{\mathcal{E}} \vert_{\mathcal{S}^{\otimes 2k}}$. Suppose $k \leq \dim \mathcal{S}$, we have $p(k)^2$ (see the definition of partition function $p(k)$ in Section \ref{sec:SnTheory}) invariant subspaces for $T_k^{\mathcal{E}} \vert_{\mathcal{S}^{\otimes 2k}}$. 

As long as $\mathcal{E}$ converges to unitary $k$-design with respect to $G$, the unit eigenvectors of $\mathcal{E}$ are exclusively contained in $p(k)$ invariant subspaces given by $\mathcal{P}_{\nu\zeta}$ with $\nu = \zeta$. To be more concrete, we illustrate by examples and diagrams of $k = 2$ and $3$:

\begin{example}\label{example:InvariantSubK=2}
	When $k = 2$, let $\hat{e}, \hat{\sigma}$ be the matrix representations of the identity and SWAP from $S_2$ acting $\mathcal{S}^{\otimes 2}$. It is fairly straightforward to write down the projections onto the trivial and sign representation of $S_2$:
	\begin{align}\label{eq:S2projectors}
		\Pi_{+} = \frac{\hat{e} + \hat{\sigma}}{2}, \quad \Pi_{-} = \frac{\hat{e} - \hat{\sigma}}{2}.
	\end{align}
	Labeling basis elements of $\mathcal{S}$ by $a,b,c,d...$ under a certain order, unit eigenstates of $\Pi_{+}, \Pi_{-}$ can be written as
	\begin{align}\label{eq:S2eigenstates}
		\Big\{ \frac{1}{\sqrt{2}} (\ket{a,b} + \ket{b,a}), \ket{a,a}; a < b \Big\}, \quad \Big\{ \frac{1}{\sqrt{2}} (\ket{a,b} - \ket{b,a}); a < b \Big\},
	\end{align}
	respectively. Then we define
	\begin{align}\label{eq: projector-invariant-subspace}
	\begin{aligned}
		\mathcal{P}_{\sbA}(\rho) \vcentcolon = \Pi_+ \rho \Pi_+, \quad 
		\mathcal{P}_{\sbB}(\rho) \vcentcolon = \Pi_+ \rho \Pi_-, \quad 
		\mathcal{P}_{\sbC}(\rho) \vcentcolon = \Pi_- \rho \Pi_+, \quad
		\mathcal{P}_{\sbD}(\rho) \vcentcolon = \Pi_- \rho \Pi_-, 
	\end{aligned}
	\end{align}
	which yield four invariant subspaces $\Ima \mathcal{P}_{\sbA},\Ima \mathcal{P}_{\sbB},\Ima \mathcal{P}_{\sbC},\Ima \mathcal{P}_{\sbD}$ obtained by taking tensor products from \eqref{eq:S2eigenstates} like
	\begin{align}\label{eq:ExampleBasis}
	\begin{aligned}
		& \frac{1}{2} (\ket{a,b} + \ket{b,a})  (\bra{c,d} + \bra{c,d}), \quad 
		\frac{1}{2} (\ket{a,b} + \ket{b,a})  (\bra{c,d} - \bra{c,d}), \\
		& \frac{1}{2} (\ket{a,b} - \ket{b,a})  (\bra{c,d} + \bra{c,d}), \quad
		\frac{1}{2} (\ket{a,b} - \ket{b,a})  (\bra{c,d} - \bra{c,d}).
	\end{aligned}
	\end{align}
	Viewing $\operatorname{End}(\mathcal{S}^{\otimes 2})$ as a large block of matrices, the decomposition of four invariant subspaces can be sketched in the following diagram:
	
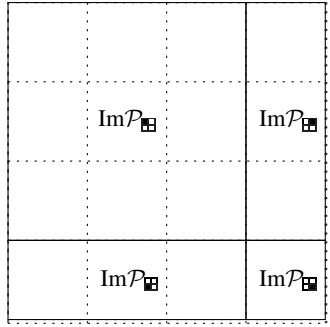
\begin{center}
	\begin{tikzpicture}[x=0.75pt,y=0.75pt,yscale=-2,xscale=2]
		\draw  [draw opacity=0][dash pattern={on 0.84pt off 2.51pt}] (86.87,117.21) -- (167.23,117.21) -- (167.23,198.14) -- (86.87,198.14) -- cycle ; \draw  [dash pattern={on 0.84pt off 2.51pt}] (86.87,117.21) -- (86.87,198.14)(106.87,117.21) -- (106.87,198.14)(126.87,117.21) -- (126.87,198.14)(146.87,117.21) -- (146.87,198.14)(166.87,117.21) -- (166.87,198.14) ; \draw  [dash pattern={on 0.84pt off 2.51pt}] (86.87,117.21) -- (167.23,117.21)(86.87,137.21) -- (167.23,137.21)(86.87,157.21) -- (167.23,157.21)(86.87,177.21) -- (167.23,177.21)(86.87,197.21) -- (167.23,197.21) ; \draw  [dash pattern={on 0.84pt off 2.51pt}]  ;
		\draw   (86.87,117.21) -- (146.87,117.21) -- (146.87,177.21) -- (86.87,177.21) -- cycle ;
		\draw   (146.87,177.21) -- (166.87,177.21) -- (166.87,197.21) -- (146.87,197.21) -- cycle ;
		\draw   (86.87,177.21) -- (146.87,177.21) -- (146.87,197.21) -- (86.87,197.21) -- cycle ;
		\draw   (146.87,117.21) -- (166.87,117.21) -- (166.87,177.21) -- (146.87,177.21) -- cycle ;
		
		\draw (120.72,147) node  [font=\small] [align=left] {\begin{minipage}[lt]{33.34pt}\setlength\topsep{0pt}
				Im$\mathcal{P}_{\sbA}$
		\end{minipage}};
		\draw (160.72,147) node  [font=\small] [align=left] {\begin{minipage}[lt]{31.52pt}\setlength\topsep{0pt}
				Im$\mathcal{P}_{\sbB}$
		\end{minipage}};
		\draw (120.72,187.21) node  [font=\small] [align=left] {\begin{minipage}[lt]{31.52pt}\setlength\topsep{0pt}
				Im$\mathcal{P}_{\sbC}$
		\end{minipage}};
		\draw (160.72,187.21) node  [font=\small] [align=left] {\begin{minipage}[lt]{31.52pt}\setlength\topsep{0pt}
				Im$\mathcal{P}_{\sbD}$
		\end{minipage}};
	\end{tikzpicture}
	\end{center}
	It is easy to check that the two unit eigenvectors of both $T_2^G$ are maximally mixed state constructed from \eqref{eq:S2eigenstates} living in $\Ima\mathcal{P}_{\sbA}$ and $\Ima\mathcal{P}_{\sbD}$ respectively. Together with the action of YJM elements that we will study comprehensively in Section \ref{sec:sketches} and Section \ref{sec:detailsT2CQA}, studying $T_2^{\mathcal{E}_{\CQA}}$ within these four invariant subspaces is a crucial step to bound its second largest eigenvalue.
\end{example}

\begin{example}\label{example:InvariantSubK=3}
	When $k = 3$ on $\mathcal{S}^{\otimes 3}$, we have projections
	\begin{align}
		\Pi_{(3)}, \quad \Pi_{(2,1)}, \quad \Pi_{(1^3)}
	\end{align}
	corresponding to the trivial, standard and sign representations of $S_3$ respectively. Collections of projections yields $9$ invariant subspaces for $T_3^{\mathcal{E}}$. As a comparison to the previous example, by Schur-Weyl duality and the Wedderburn theorem~\cite{Fulton1997,Sagan01,Goodman2009}, $T_3^{\mathcal{E}}$ admits
	\begin{enumerate}
		\item[(a)] one unit eigenvector from $\mathcal{P}_{(3)(3)}$.
		
		\item[(b)] four unit eigenvectors from $\mathcal{P}_{(2,1)(2,1)}$.
		
		\item[(c)] one unit eigenvector from $\mathcal{P}_{(1^3)(1^3)}$.
	\end{enumerate}
	A diagrammatic representation of the space decomposition can be: 
	\begin{center}
	\resizebox{0.35\textwidth}{0.35\textwidth}
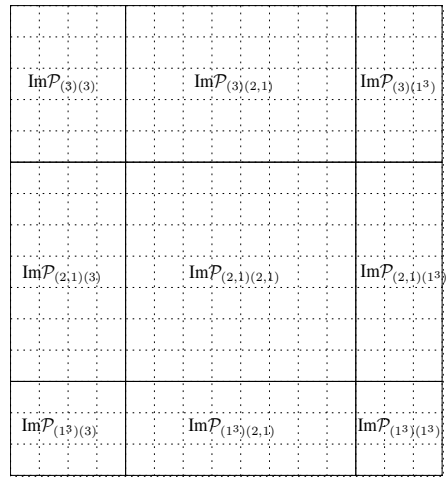
{
	\begin{tikzpicture}[x=0.75pt,y=0.75pt,yscale=-1,xscale=1]
		\draw  [draw opacity=0][dash pattern={on 0.84pt off 2.51pt}] (402.95,0.32) -- (703.78,0.32) -- (703.78,300.91) -- (402.95,300.91) -- cycle ; \draw  [dash pattern={on 0.84pt off 2.51pt}] (402.95,0.32) -- (402.95,300.91)(422.95,0.32) -- (422.95,300.91)(442.95,0.32) -- (442.95,300.91)(462.95,0.32) -- (462.95,300.91)(482.95,0.32) -- (482.95,300.91)(502.95,0.32) -- (502.95,300.91)(522.95,0.32) -- (522.95,300.91)(542.95,0.32) -- (542.95,300.91)(562.95,0.32) -- (562.95,300.91)(582.95,0.32) -- (582.95,300.91)(602.95,0.32) -- (602.95,300.91)(622.95,0.32) -- (622.95,300.91)(642.95,0.32) -- (642.95,300.91)(662.95,0.32) -- (662.95,300.91)(682.95,0.32) -- (682.95,300.91)(702.95,0.32) -- (702.95,300.91) ; \draw  [dash pattern={on 0.84pt off 2.51pt}] (402.95,0.32) -- (703.78,0.32)(402.95,20.32) -- (703.78,20.32)(402.95,40.32) -- (703.78,40.32)(402.95,60.32) -- (703.78,60.32)(402.95,80.32) -- (703.78,80.32)(402.95,100.32) -- (703.78,100.32)(402.95,120.32) -- (703.78,120.32)(402.95,140.32) -- (703.78,140.32)(402.95,160.32) -- (703.78,160.32)(402.95,180.32) -- (703.78,180.32)(402.95,200.32) -- (703.78,200.32)(402.95,220.32) -- (703.78,220.32)(402.95,240.32) -- (703.78,240.32)(402.95,260.32) -- (703.78,260.32)(402.95,280.32) -- (703.78,280.32)(402.95,300.32) -- (703.78,300.32) ; \draw  [dash pattern={on 0.84pt off 2.51pt}]  ;
		\draw   (482.95,100.32) -- (642.95,100.32) -- (642.95,240.32) -- (482.95,240.32) -- cycle ;
		\draw   (642.95,240.32) -- (702.95,240.32) -- (702.95,300.32) -- (642.95,300.32) -- cycle ;
		\draw   (482.95,240.32) -- (642.95,240.32) -- (642.95,300.32) -- (482.95,300.32) -- cycle ;
		\draw   (402.95,240.32) -- (482.95,240.32) -- (482.95,300.32) -- (402.95,300.32) -- cycle ;
		\draw   (642.95,100.32) -- (702.95,100.32) -- (702.95,240.32) -- (642.95,240.32) -- cycle ;
		\draw   (402.95,100.32) -- (482.95,100.32) -- (482.95,240.32) -- (402.95,240.32) -- cycle ;
		\draw   (402.95,0.32) -- (482.95,0.32) -- (482.95,100.32) -- (402.95,100.32) -- cycle ;
		\draw   (482.95,0.32) -- (642.95,0.32) -- (642.95,100.32) -- (482.95,100.32) -- cycle ;
		\draw   (642.95,0.32) -- (702.95,0.32) -- (702.95,100.32) -- (642.95,100.32) -- cycle ;
		
		\draw (562.95,172) node  [font=\small] [align=left] {\begin{minipage}[lt]{53.66pt}\setlength\topsep{0pt}
				{\small Im$\mathcal{P}_{( 2,1)( 2,1)}$}
		\end{minipage}};
		\draw (442.95,50.32) node  [font=\small] [align=left] {\begin{minipage}[lt]{41.73pt}\setlength\topsep{0pt}
				Im$\mathcal{P}_{( 3)( 3)}$
		\end{minipage}};
		\draw (672.95,270.32) node  [font=\small] [align=left] {\begin{minipage}[lt]{43pt}\setlength\topsep{0pt}
				Im$\mathcal{P}_{( 1^{3}) ( 1^{3})}$
		\end{minipage}};
		\draw (562.95,50.32) node  [font=\small] [align=left] {\begin{minipage}[lt]{48.21pt}\setlength\topsep{0pt}
				Im$\mathcal{P}_{( 3)( 2,1)}$
		\end{minipage}};
		\draw (672.95,50.32) node  [font=\small] [align=left] {\begin{minipage}[lt]{40pt}\setlength\topsep{0pt}
				{\small Im$\mathcal{P}_{( 3) ( 1^{3})}$}
		\end{minipage}};
		\draw (682.95,172) node  [font=\small] [align=left] {\begin{minipage}[lt]{55pt}\setlength\topsep{0pt}
				Im$\mathcal{P}_{( 2,1) ( 1^{3})}$
		\end{minipage}};
		\draw (442.95,172) node  [font=\small] [align=left] {\begin{minipage}[lt]{47.67pt}\setlength\topsep{0pt}
				Im$\mathcal{P}_{( 2,1)( 3)}$
		\end{minipage}};
		\draw (442.95,270.32) node  [font=\small] [align=left] {\begin{minipage}[lt]{48.4pt}\setlength\topsep{0pt}
				Im$\mathcal{P}_{( 1^{3})( 3)}$
		\end{minipage}};
		\draw (562.95,270.32) node  [font=\small] [align=left] {\begin{minipage}[lt]{54.4pt}\setlength\topsep{0pt}
				Im$\mathcal{P}_{( 1^{3})( 2,1)}$
		\end{minipage}};	
	\end{tikzpicture}
	}
	\end{center}
\end{example}

Unlike the case of $k = 2$, despite the insight gained from using invariant subspaces, there are still difficulties when analyzing the spectral gap of $T_k^{\mathcal{E}}$ for $k \geq 3$:
\begin{enumerate}
	\item It is not easy to write down the projections explicitly like \eqref{eq:S2projectors} due to the complexity to compute Eq.~\eqref{eq:OrthogonalProjection} for general symmetric group $S_k$, especially in computing $S_k$ characters~\cite{Ingram1950,Roichman1996,Lassalle2008,Giambruno2015}. Essentially, bases like those in~\eqref{eq:ExampleBasis} for the invariant subspaces should be built by Young basis of $S_k$ irreps and using Schur-Weyl duality on $\mathcal{S}^{\otimes k}$ (see Section \ref{sec:SnTheory} and Ref.~\cite{Zheng2021SpeedingUL} for more details). It is computationally feasible with explicit $n$ and $k$~\cite{bacon2005quantum,krovi}, but there is no closed-formformula to obtain the generic Young basis.
	
	\item  It is also infeasible to orthogonalize permutations to construct an orthonormal basis for $\Comm_k(G)$ when $k \geq 3$. Moreover, the subspaces $\mathcal{S}$ of interest are taken from decomposition of the $n$-qudit system under symmetries. There are always subspaces whose dimensions scale polynomially in $n$ rather than exponentially like the entire Hilbert space. This hinders the use of \emph{approximate orthogonality} of permutations for exponentially large spaces~\cite{harrow2016local,Haferkamp2021,Harrow2023orthogonality,Metger2024,Haferkamp2024linear}, which also makes the Hamiltonian spectral gap methods~\cite{Knabe1988,Nachtergaele1996} inapplicable.
\end{enumerate}
Lower-bounding the spectral gap for general $k$ is thus challenging, and we leave it as an open problem. We mainly focus on the case of $k=2$ in this work.


\subsection{Induced Markov chain and comparison theorems}\label{sec:Markov}

We borrow techniques from the Markov chain theory, especially methods for bounding the spectral gap of a reversible transition matrix by comparing various properties with another transition matrix with known spectral gap. More systematic descriptions of these methods can be found in Refs.~\cite{Diaconis1988,DiaconisCheeger1991,DiaconisComparison1993,chung1997spectral,Levin2009}.  

\begin{definition}\label{def:induced-markov-chain}
	Let $(\mathcal{X},P)$ be a Markov chain with state space $\mathcal{X}$ and a row-stochastic transition matrix $P$. Given a subspace $\mathcal{S} \subset \mathcal{X}$ and a path $\gamma$ starting from some state $x_0 \in \mathcal{S}$, the \emph{first returning time} $\tau_x^+(\gamma)$ is the minimum time/number of steps $t$ such that
	\begin{align}
		\gamma(1) \notin \mathcal{S},...,\gamma(t - 1) \notin \mathcal{S}, \gamma(t) \in \mathcal{S}.
	\end{align}
	An \emph{induced Markov chain} $(\mathcal{S}, \tilde{P})$ on the subspace $\mathcal{S}$ is defined by the following transition probability between any $x, y \in \mathcal{S}$:
	\begin{align}
		\tilde{P}(x,y) = P(x,y) + \sum_{n=1}^\infty \sum_{x_1,...,x_n \notin \mathcal{S}} P(x,x_1) P(x_1,x_2) \cdots P(x_{n-1},x_n) P(x_n,y).
	\end{align}
	The summation is thus taken over all paths starting from $x$ whose first returning point is exactly $y$.
\end{definition}

The notion of induced Markov chains appears when studying random walks on subgraphs with boundary conditions~\cite{chung1997spectral} as well as general Markov processes~\cite{Levin2009}. In order to be self-contained, we prove that $\tilde{P}(x,y)$ is well-defined and $\tilde{P}$ is still row-stochastic. The transition matrix $P$ is always assumed to be irreducible and reversible.

Let us truncate the above infinite series into partial sums and consider 
\begin{align}
	P_N = & \sum_{y \in \mathcal{S}} \sum_{n=1}^N  \sum_{x_1,...,x_n \notin \mathcal{S}} P(x,x_1) P(x_1,x_2) \cdots P(x_{n-1},x_n) P(x_n,y), \\
	S_N = & \sum_{x_1,...,x_N,x_{N+1} \notin \mathcal{S}} P(x,x_1) P(x_1,x_2) \cdots P(x_{N-1},x_N) P(x_N,x_{N+1}).
\end{align}
Let $Q_N = S_N + P_N$. It is rather straightforward to check by definition that
\begin{align}
	P_N \leq Q_N \equiv \sum_{x_1 \notin \mathcal{S}} P(x,x_1) = P(x, \mathcal{X} \setminus  \mathcal{S}) = P(x, \mathcal{S}^c),
\end{align}
which guarantees the convergence of the series. 

Since the transition matrix $P$ is irreducible, there is an integer $m$ for which $P^m$ has nonzero entries everywhere. Therefore, for any state $z \in \mathcal{S}^c$, $P^m(z,\mathcal{S}^c) < 1$ because there must be transitions to states in $\mathcal{S}$. Actually, as the state space is finite, there is a uniform upper bound $b$ such that 
\begin{align}
	P^m(z,\mathcal{S}^c) < b < 1
\end{align}
holds for any $z \in \mathcal{S}^c$. Consequently,
\begin{align}
	\begin{aligned}
		S_N = & \sum_{x_1,...,x_N,x_{N+1} \notin \mathcal{S}} P(x,x_1) \cdots P(x_{N-1},x_N) P(x_N,x_{N+1}) \\
		= & \sum_{x_1,...,x_N,x_{N+1} \in \mathcal{S}^c} P(x,x_1)  \cdots P(x_{N-m}, x_{N+1-m}) P(x_{N+1-m}, x_{N+2-m}) \cdots P(x_{N-1},x_N) P(x_N,x_{N+1}) \\
		< & \sum_{x_1,...,x_{N-m},x_{N+1-m} \in \mathcal{S}^c} P(x,x_1) \cdots P(x_{N-m}, x_{N+1-m}) P^m(x_{N+1-m},\mathcal{S}^c) \\
		< & b \sum_{x_1,...,x_{N-m},x_{N+1-m} \in \mathcal{S}^c} P(x,x_1) \cdots P(x_{N-m}, x_{N+1-m}) = bS_{N-m}.
	\end{aligned}
\end{align}
This is sufficient to imply that the sequence $\{S_N\}$ has a subsequence converging to zero. Since $S_{N-1} > S_N > 0$ by definition, $\{S_N\}$ also converges to zero. Recall that $S_N = Q_N - P_N$, we conclude that 
\begin{align}
	\lim_{N \to \infty} P_N = P(x,\mathcal{S}^c) \implies \sum_{y \in \mathcal{S}} \tilde{P}(x,y) = \sum_{y \in \mathcal{S}} P(x,y) + \lim_{N \to \infty} P_N = 1.
\end{align}

\begin{example}\label{example:Markov}
	Let $\mathcal{X} = \{s_1, s_2, s_3\}$ consist of three states with transition matrix
	\begin{align}
		P = \begin{pmatrix}
			P(1,1) & P(1,2) & P(1,3) \\
			P(2,1) & P(2,2) & P(2,3) \\
			P(3,1) & P(3,2) & P(3,3) 
		\end{pmatrix}.
	\end{align}
	We provide a simple example for constructing the induced transition matrix on the first 2 states from the original space, which is helpful for understanding the proofs in Section \ref{sec:detailsU(1)} \& \ref{sec:detailsSU(d)}. By definition,
	\begin{align}
		\tilde{P}(1,1) & = P(1,1) + P(1,3) \big( 1 + P(3,3) + P(3,3)^2 + \cdots \big) P(3,1) 
		= P(1,1) + P(1,3) \frac{P(3,1)}{1 +  P(3,3)},  \\
		\tilde{P}(1,2) & = P(1,2) + P(1,3) \big( 1 + P(3,3) + P(3,3)^2 + \cdots \big) P(3,2) 
		= P(1,2) + P(1,3) \frac{P(3,2)}{1 - P(3,3)}. 
	\end{align}
	Since $P(3,1) + P(3,2) +  P(3,3) = 1$, $\tilde{P}(1,1) + \tilde{P}(1,2) = 1$. The other two transition probabilities can be obtained accordingly and $\tilde{P}$ is well-defined.
\end{example}

\begin{theorem}[\cite{Levin2009}]\label{thm:induced-chain-spectral-gap} 
	Let $(\mathcal{X},P)$ be a irreducible and reversible Markov chain. Let $\mathcal{S} \subset \mathcal{X}$ be any subspace. Then the induced Markov chain $(\mathcal{S},\tilde{P})$ has a larger spectral gap. That is $\Delta(\tilde{P}) \geq \Delta(P)$.
\end{theorem} 

\begin{definition} \label{def:dirichlet-form}
	Let $(\mathcal{X},P)$ be a reversible Markov chain. The \emph{Dirichlet form}, as a binary form, is defined for functions $f,h$ on the state space $\mathcal{X}$ by
	\begin{align}
		\langle (I-P) f, g \rangle_{\pi_P} = \sum_{x \in \mathcal{X}} ((I-P)f )(x) g(x) \pi_P(x),
	\end{align}
	where $\pi_P$ is the stationary distribution of $P$. Especially,
	\begin{align}
		\mathcal{E}(P,f) \vcentcolon = \langle (I-P) f, f \rangle_{\pi_P} = \frac{1}{2} \sum_{x,y \in \mathcal{X}} (f(x) - f(y))^2 \pi_P(x) P(x,y).
	\end{align}
\end{definition}

\begin{theorem}[\cite{Levin2009}]\label{thm:comparison-dirichlet}
	Let $P_1,P_2$ be irreducible and reversible Markov chains on the same state space. If there is some $\alpha > 0$ such that $\mathcal{E}(P_1,f) \leq \alpha \mathcal{E}(P_2,f)$ for any $f$, then
	\begin{align}
		\Delta(P_1) \leq \alpha \max_{x \in \mathcal{X}} \frac{\pi_{P_1}(x)}{\pi_{P_2}(x)} \Delta(P_2).
	\end{align}
\end{theorem}

\begin{definition}\label{def:Congestion}
	Let $P_1,P_2$ be irreducible and reversible Markov chains on the same state space. For any pair of distinct states $x,y \in \mathcal{X}$, as long as $P_1(x,y) > 0$, we select a path $\gamma_{xy} = (x = x_0,x_1,...,x_{N-1},y = x_N)$ satisfying $P_2(x_i,x_{i+1}) > 0$. Note that $\gamma_{xy}$ can be different from $\gamma_{yx}$. The \emph{congestion ratio} $A$ defined through this collection of paths is 
	\begin{align}
		A = \max_{(p,q), P_2(p,q) > 0} \frac{1}{\pi_{P_2}(p) P_2(p,q)} \sum_{\stackrel{x,y\in \mathcal{X}}{(p,q) \in \gamma_{xy}}} \pi_{P_1}(x) P_1(x,y) \vert \gamma_{xy} \vert
	\end{align} 
	with $\vert \gamma_{xy} \vert$ the length of the path.
\end{definition}

\begin{theorem}[\cite{DiaconisComparison1993}]\label{thm:path-comparison-theorem} 
	Let $P_1,P_2$ be irreducible and reversible Markov chains on the same state space. Given any collection of paths with the corresponding congestion ratio $A$,
	\begin{align}
		\Delta(P_1) \leq A \max_{x \in \mathcal{X}} \frac{\pi_{P_1}(x)}{\pi_{P_2}(x)} \Delta(P_2).
	\end{align}
\end{theorem}


\newpage 
\section{$\U(1)$-symmetric universality and unitary $k$-designs from CQA}\label{sec:U(1)}

In this section, we specifically discuss the problem of generating $\U(1)$-symmetric unitaries and $k$-designs with respect to $\mathcal{U}_{\times,\U(1)}$. Many results can also be found from systematic studies in Refs.~\cite{MarvianNature,marvian2023nonuniversality,U(1)Design2023,PRXQuantum.4.040331,MarvianDesign,mitsuhashi2024Designs,mitsuhashi2024Designs2}. Here we first revisit some key results by techniques developed for $\SU(d)$- and permutation-symmetric circuits in Refs.~\cite{Zheng2021SpeedingUL,SUd-k-Design2023} as a warm-up, before presenting the main results on convergence speed.

\subsection{Universality under $\U(1)$ symmetry}\label{sec:U(1)Univeristy}

In analogy to Definition \ref{def:CQA}, we now define $\CQA_{\U(1)}$ using unitary time evolutions of SWAPs and Pauli-$Z$ matrices. We are going to prove that this group turns out to be generated by $2$-local $\U(1)$-symmetric unitaries and contains all $\U(1)$-symmetric unitaries with trivial relative phases (cf.~\eqref{eq:universalitySU(d)}).

\begin{definition}\label{def:CQAU(1)}
	Consider the  Hamiltonians
	\begin{align}
		H_S = \sum_{j = 1}^{n-1} (j,j+1), \quad H_{Z} = \sum_{k,l} \beta_{kl} Z_k Z_l,
	\end{align}
	where $H_S$ is still the summation of adjacent transpositions and $H_{Z}$ is a linear combination of Pauli-$Z$ matrices up to second-order tensor products with real parameters $\beta_{kl}$. Then
	\begin{align}
		\CQA_{\U(1)} \vcentcolon = \Big\langle \exp(-i \sum_{k,l} \beta_{kl} Z_k Z_l), \quad \exp(-i\gamma H_S) \Big\rangle.
	\end{align}
\end{definition}

Let $\mathfrak{gl}_{\times,\U(1)}(\mathbb{C})$ denote the collection of $\U(1)$-symmetric complex matrices acting on the Hilbert space respecting the following direct sum decomposition
\begin{align}
	(\mathbb{C}^2)^{\otimes n} \cong \bigoplus_{0 \leq r \leq n} \text{span}\{ \ket{i_1,...,i_n}; i = 0 \text{ or } 1, r \text{ many indices are selected to be } 0 \}  = \bigoplus_{\mu = (n-r,r)} S^\mu, \tag{B42$^\ast$}
\end{align}
where we simply work with the computational basis as mentioned in Section \ref{sec:Tabloid}. Suppose $\mathfrak{d}_{\times,\U(1)}(\mathbb{C})$ encompasses all diagonal matrices, which is formally the Cartan subalgebra $\mathfrak{gl}_{\times,\U(1)}(\mathbb{C})$. Note that since U$(1)$ charge sectors are multiplicity-free, arbitrary diagonal matrix under computational basis is $\U(1)$-symmetric. As a comparison, only diagonal matrices, under Young basis, having repeated diagonal entries in equivalent $S_n$ irrep sectors can be SU$(d)$-symmetric (see Section \ref{sec:SnTheory} \& \ref{sec:SchurWeyl}).

Any diagonal matrix under the computational basis $\ket{e_i}$ can be spanned by tensor products $Z_{r_1} \otimes \cdots \otimes Z_{r_j}$ of Pauli-$Z$ matrices. Here we would like to retain the $2$-locality. As a result, our argument uses at most second-order tensor product of Pauli-$Z$'s as Definition \ref{def:CQAU(1)}. To begin with, we denote by $\alpha_i(r)$ the eigenvalue of $\ket{e_i}$ measured by $Z_r$. Let $E_{kl} = \ket{e_k}\bra{e_l}$ be a matrix unit. The so-called \emph{root} $a_{kl}$ of $E_{kl}$ is defined by Lie brackets (adjoint representations):
\begin{align}\label{eq:RootEigenvalue}
	a_{kl}(r) E_{kl} \vcentcolon = [Z_r, E_{kl}] = (\alpha_k(r) - \alpha_l(r)) E_{kl}, \quad a_{kl}(r,s) E_{kl} = [Z_r Z_s, E_{kl}] \vcentcolon = (\alpha_k(r) \alpha_k(s) - \alpha_l(r) \alpha_k(s)) E_{kl}.
\end{align}
A complete root vector takes $a_{kl}(r), a_{kl}(r,s)$ and these obtained from all higher order tensor products of Pauli-$Z$'s as components. As a basic result from the theory of Lie algebra, different off-diagonal matrix units cannot share the same root vector. While for practice we only consider the \emph{partial} root vector, denoted by $a_{kl}^{(2)}$, consisting of $a_{kl}(r), a_{kl}(r,s)$ determined by at most $2$-local operators. Writing $a_{kl}^{(2)} = \alpha_k^{(2)} - \alpha_l^{(2)}$ according to Eq.~\eqref{eq:RootEigenvalue}, we have the following lemma (a similar phenomenon in the case of $\SU(d)$ symmetry is verified using YJM elements and content vectors in Ref.~\cite{Zheng2021SpeedingUL}):

\begin{lemma}\label{lemma:PartialRoot}
	Different off-diagonal matrix units cannot have the same partial root vector. That is, for computational basis elements $u_1,u_2,v_1,v_2$ with (a) $u_1 \neq u_2, v_1 \neq v_2$ or (b) $u_1 = u_2$, $v_1 \neq v_2$ or (c) $u_1 \neq u_2$, $v_1 = v_2$, 
	\begin{align}
		a_{u_1 v_1}^{(2)} = \alpha_{u_1}^{(2)} - \alpha_{v_1}^{(2)} \neq \alpha_{u_2}^{(2)} - \alpha_{v_2}^{(2)} = a_{u_2 v_2}^{(2)}. 
	\end{align}
\end{lemma}
\begin{proof}
	Suppose the statement is true for $n-1$ qubits. We denote by $0,1$ the spin basis of each qubit for brevity. Then possible choices of the spins of basis elements $u_1,v_1, u_2,v_2$ (not eigenvalues) at the last qubit are listed as:
	\begin{align}
		& (0,0, \quad 0,0), \quad (0,0, \quad 1,1), \quad (1,1, \quad 0,0), \quad (1,1, \quad 1,1), \\
		& (0,1, \quad 0,1), \quad (1,0, \quad 1,0), \quad (0,1, \quad 1,0), \quad (1,0, \quad 0,1), \\
		& (0,0, \quad 0,1), \quad (0,0, \quad 1,0), \quad (0,1, \quad 0,0), \quad (1,0, \quad 0,0), \\
		& (0,1, \quad 1,1), \quad (1,0, \quad 1,1), \quad (1,1, \quad 0,1), \quad (1,1, \quad 1,0).
	\end{align} 
	The statement can be confirmed immediately for the last ten cases because the difference between eigenvalues measured by $Z_n$ on the last qubit are different. That is, 
	\begin{align}
		\alpha_{u_1}^{(2)}(n) - \alpha_{v_1}^{(2)}(n) \neq \alpha_{u_2}^{(2)}(n) - \alpha_{v_2}^{(2)}(n).
	\end{align}
	
	We still need to analyze the first six cases. Erasing the last qubits from $u_1,v_1,u_2,v_2$, if the resultant vectors $\hat{u}_1,\hat{v}_1,\hat{u}_2,\hat{v}_2$ still satisfy one of the three assumptions of the statement, we get the proof by induction. Otherwise, we should have $\hat{u}_1 = \hat{u}_2$ and $\hat{v}_1 = \hat{v}_2$ as binary strings of $0, 1$ without the last qubit. However, the original vectors $u_1,v_1,u_2,v_2$ are defined under one assumption of the statement. The spins on their last qubits must be selected obeying the second or third cases from the above list. Together with the fact that we are considering off-diagonal matrix units, $u_1 \neq v_1, u_2 \neq v_2$ and there exists some $r < n$ such that $\alpha_{u_1}(r) - \alpha_{v_1}(r) = \alpha_{u_2}(r) - \alpha_{v_2}(r) \neq 0$. Using the second-order product with $Z_n$, we have
	 \begin{align}
	 	\alpha_{u_1}^{(2)}(r,n) - \alpha_{v_1}^{(2)}(r,n) = \alpha_{u_1}(r) - \alpha_{v_1}(r) \neq - \alpha_{u_2}(r) + \alpha_{v_2}(r) = - \big( \alpha_{u_2}^{(2)}(r,n) - \alpha_{v_2}^{(2)}(r,n) \big)
	 \end{align}
	 which completes the proof.
\end{proof}

The following lemma is given in our previous paper Ref.~\cite{Zheng2021SpeedingUL}, and we reproduce here with proof details for further use. For $\U(1)$ symmetry, one can simply replace $\mathfrak{gl}(D,\mathbb{C}), \mathfrak{d}(D)$ by $\mathfrak{gl}_{\times,\U(1)}(\mathbb{C}), \mathfrak{d}_{\times,\U(1)}(\mathbb{C})$ respectively. The matrix $M$ defined below also needs to respect the $\U(1)$ symmetry, i.e., having nonzero off-diagonal entries within $\U(1)$ charge sectors in Eq.~\eqref{eq:SpaceDecompositionU(1)}.

\begin{lemma}
	Given any $D \times D$ matrix $M$, let $\mathcal{I} \subset \{1,...,D\} \times \{1,...,D\}$ be the index set corresponding to nonzero off-diagonal entries $c_{ij}$ of $M$. Let $\mathfrak{d}(D)$ be the collection of all $D \times D$ complex diagonal matrix. Then 
	\begin{align}
		\langle \mathfrak{d}(D), M \rangle = \mathfrak{d}(D) \oplus \Big( \bigoplus_{(i,j) \in \mathcal{I}} \text{span}(E_{ij}) \Big).
	\end{align}
	That is, the generated Lie subalgebra contains all $D \times D$ complex diagonal matrix and those spanned by off-diagonal matrix units $E_{ij}$ as long as $c_{ij} \neq 0$.
\end{lemma}
\begin{proof}
	For any diagonal matrix $h = \operatorname{diag}(h_1,...,h_D) \in \mathfrak{d}(D)$, the Lie bracket $[h, E_{kl}] = (h_k - h_l) E_{kl}$ gives the root $a_{kl}(h) \vcentcolon = h_k - h_l$ of $E_{kl}$ under $h$. Let $h \vcentcolon = \sum_i \lambda_i E_{ii}$ with $\lambda_i$ being determined later. We define 
	\begin{align}
		M_1 \vcentcolon = [h, M] = [\sum_i \lambda_i E_{ii}, \sum_{(k,l) \in \mathcal{I}} c_{kl} E_{kl} ] = \sum_i \sum_{(k,l) \in \mathcal{I}} a_{kl} (E_{ii}) \lambda_i E_{kl} = \sum_{(k,l) \in \mathcal{I}} (a_{kl} \cdot \lambda) c_{kl} E_{kl},
	\end{align}
	where in the second step of the above computation, we omit all possible diagonal elements of $A$. This is legitimate because $h$ commutes with any diagonal matrix. Besides, $a_{kl} \cdot \lambda$ is thought of as the inner product of $a_{kl} = ( a_{kl}(E_{ii}) )$ and $\lambda = (\lambda_i)$. Although most components of $a_{kl}$ equal zero by the definition of roots, it turns out that the notation $a_{kl} \cdot \lambda$ is convenient for the following proof. With $M_1$ defined, we continue to set 
\begin{align}\label{eq:VandermondeBasisChange}
		\begin{aligned}
			M_2 \vcentcolon & = [h,[h,A]] =  \sum_{(k,l) \in \mathcal{I}} (a_{kl} \cdot \lambda)^2 c_{kl} E_{kl}, \\
			& \vdots \\
			M_{\vert \mathcal{I} -1\vert} \vcentcolon & = [h,...,[h,A] \cdots ] = \sum_{(k,l) \in \mathcal{I}} (a_{kl} \cdot \lambda)^{\vert \mathcal{I} -1\vert} c_{kl} E_{kl}.
		\end{aligned}
	\end{align}
	Recall that $\vert \mathcal{I} \vert$ is the number of nonzero off-diagonal elements of $M$. Then let us consider the following \emph{Vandermonde matrix}:
	\begin{align}\label{eq:Vandermonde}
		V (a_{kl} \cdot \lambda) = \begin{pmatrix}
			1 & 1 & \cdots & 1 \\ a_{k_1 l_1} \cdot \lambda & a_{k_2 l_2} \cdot \lambda & \cdots & a_{k_{\vert \mathcal{I} \vert} l_{\vert \mathcal{I} \vert}} \cdot \lambda \\ \vdots & \vdots & \ddots & \vdots \\ (a_{k_1 l_1} \cdot \lambda)^{\vert \mathcal{I} \vert - 1} & (a_{k_2 l_2} \cdot \lambda)^{\vert \mathcal{I} \vert - 1} & \cdots & (a_{k_{\vert \mathcal{I} \vert} l_{\vert \mathcal{I} \vert}}\cdot \lambda)^{\vert \mathcal{I} \vert - 1} 
		\end{pmatrix}.
	\end{align}
	Viewing $\mathfrak{gl}(D,\mathbb{C})$ as a $D \times D$-dimensional vector space with $E_{ij}$ as the standard basis, we note that $V(a_{kl} \cdot \lambda)$ transforms \emph{nonzero} vectors 
	\begin{align}
		\{c_{k_1 l_1} E_{k_1 l_1},...,c_{k_{\vert \mathcal{I} \vert} l_{\vert \mathcal{I} \vert}} E_{k_{\vert \mathcal{I} \vert} l_{\vert \mathcal{I} \vert}}\} \text{ to }  \{M,M_1,...,M_{\vert \mathcal{I} \vert - 1}\}.
	\end{align}
	If $\det V (a_{kl} \cdot \lambda) = (-1)^{\vert \mathcal{I} \vert (\vert \mathcal{I} \vert -1)/2} \prod_{s < t} (a_{k_s l_s} \cdot \lambda - a_{k_t l_t} \cdot \lambda) \neq 0$, then the transformation has an inverse and thus the linear span of $\{M_r\}$ equals that of $\{c_{kl} E_{kl} \}$ which turns out to be $\bigoplus_{(i,j) \in \mathcal{I}} \text{span} (E_{ij})$ by definition. 
	
	To show that $\det V (a_{kl} \cdot \lambda) \neq 0$, we simply note that $a_{k_s l_s} - a_{k_t l_t} \neq 0$ as  different roots. A basic statement from linear algebra tells us that the union of hyperplanes of finitely-many nonzero vectors, here are $a_{k_s l_s} - a_{k_t l_t}$, cannot cover the whole vector space. Therefore, we can always find some nonzero $\lambda$ which does not belong to any of these hyperplanes: i.e., $(a_{k_s l_s} \cdot \lambda - a_{k_t l_t} \cdot \lambda) \neq 0$ to fulfill the requirement. On the other hand, the Lie subalgebra by $\mathfrak{d}(D)$ and $M$ contains all possible linear combinations and Lie brackets of these elements and the proof follows. 
\end{proof}

\begin{remark}
	As mentioned before, different matrix units correspond to different root vectors determined by a complete basis of diagonal matrices, which is the essential ingredient to find inverse of the Vandermonde matrix used in the above lemma. On the other hand, Lemma \ref{lemma:PartialRoot} tells us that partial root vectors of second-order Pauli-$Z$ matrices are sufficient to distinguish off-diagonal matrix units and can be used instead, like:
	\begin{align}\label{eq:VandermondeBasisChange2}
	\begin{aligned}
		[\sum_{r,s} \lambda_{rs} Z_r Z_s, M] = [\sum_{r,s} \lambda_{rs} Z_r Z_s, \sum_{(k,l) \in \mathcal{I}} c_{kl} E_{kl}] & = \sum_{r,s} \sum_{(k,l) \in \mathcal{I}} (\alpha^{(2)}_k(r,s) - \alpha^{(2)}_l(r,s) ) \lambda_{rs} c_{kl} E_{kl} \\
		& = \sum_{(k,l) \in \mathcal{I}} \big( (\alpha^{(2)}_k - \alpha^{(2)}_l) \cdot \lambda \big) c_{kl} E_{kl}.
	\end{aligned}
	\end{align} 
	Accordingly, we are still able to find a nonzero solution to $\lambda$ such that iterated Lie brackets of $h \vcentcolon = \sum_{r,s} \lambda_{rs} Z_r Z_s$ and $M$ generate 1-dimensional subspaces corresponding to nonzero entries of $M$. 
\end{remark}

A proper choice of $M$ finally enables us to generate all off-diagonal matrices over $\mathbb{C}$. In the $\U(1)$ case, there are several candidates like $H_S = \sum_{j=1}^{n-1} (j,j+1)$. One common feature of them is that their matrix representations under computational basis within any $\U(1)$ charge sector are \emph{irreducible} or \emph{path-connected}. That is, the graph represented by indices (vertices) and nonzero entries (edges) of these matrices restricted to arbitrary $\U(1)$ charge sector is connected. More generally, let $\mathcal{T}$ be a collection of SWAPs that generates $S_n$. Given two computational basis states $\ket{\psi_1}, \ket{\psi_2}$ from the same charge sector with the same number of spin-up/down qubits, there must be a general permutation $\sigma$ on qubits mapping $\ket{\psi_1}$ to $\ket{\psi_2}$. The permutation can be written as products of SWAPs from $\mathcal{T}$, thus $H_{\mathcal{T}} = \sum_{\tau \mathcal{T}} \tau$ is path-connected. Especially, we can set $\mathcal{T} = \{(1,2),...,(n-1,n)\}$. Then $H_{\mathcal{T}} = H_S$. As mentioned, a similar fact in the case of $\SU(d)$ symmetry is verified using the Young orthogonal form and $S_n$ irrep in Ref.~\cite{Zheng2021SpeedingUL}. 

\begin{theorem}\label{thm:universalityU(1)-1}
	Let $H_P$ be an irreducible/path-connected Hamiltonian described above. Let $\mathfrak{d}_{Z_k Z_l}(\mathbb{R})$ be the collection of real diagonal matrices spanned by at most second-order Pauli matrices. Then
	\begin{align}
		\bigoplus_{\mu = (n-r,r)} \mathfrak{su}(S^\mu) \subset \langle i\mathfrak{d}_{Z_kZ_l}(\mathbb{R}), iH_P \rangle_\mathbb{R}. 
	\end{align}
	As a result, the Lie group generated the unitary evolution of Hamiltonians $H_S = \sum_{j=1}^{n-1} (j,j+1), H_Z = \sum \beta_{kl} Z_k Z_l$ is universal when restricted to the semisimple Lie group $S\mathcal{U}_{\U(1)}$ of $\U(1)$-symmetric unitaries with trivial relative phase.
\end{theorem}
\begin{proof}
	We first check the Lie brackets from Eq.~\eqref{eq:VandermondeBasisChange} and \eqref{eq:VandermondeBasisChange2} with $M = iH_P$. Note that the components of (partial) root vectors with respect to $iZ_r Z_s$ are now pure imaginary. On the other hand, generating over $\mathbb{R}$, the vector $\lambda$ selected to define $h = \sum_{r,s} \lambda_{rs} iZ_r Z_s$ can only take real components. As a result, odd rows of the Vandermonde matrix from Eq.~\eqref{eq:Vandermonde} are all pure imaginary, but even rows are all real. This further implies that it is cannot be invertible over $\mathbb{R}$. However, if we separate terms from Eq.~\eqref{eq:VandermondeBasisChange} into odd and even degrees, using similar methods we retrieve basis elements like
	\begin{align}
		(E_{ij} - E_{ji}), \quad i(E_{ij} + E_{ji})
	\end{align}
	as long as $H_P$ has nonzero $(i,j)$ and $(j,i)$ entries. Since $H_P$ is path-connected and since, for instance
	\begin{align}
		\begin{aligned}
			& [i(E_{ij} + E_{ji}), i(E_{jk} + E_{kj})] = - (E_{ik} - E_{ki}), \\
			& [i(E_{ij} + E_{ji}), (E_{ij} - E_{ji})] = 2i(E_{jj} - E_{ii}),
		\end{aligned}
	\end{align}
	we are able to find all skew-Hermitian off-diagonal as well as traceless diagonal matrices with respect to U$(1)$ charge sectors, which spans $\bigoplus_\mu \mathfrak{su}(S^\mu)$.
\end{proof}

Evidently, if we incorporate all higher order tensor products of Pauli-$Z$ matrices in proving the above theorem, we will reach the complete universality. To explore more details with concern about relative phases and locality. We define 
\begin{align}
	C_{Z_r} = \bigoplus_\mu \frac{\tr_\mu(Z_r)}{\dim S^\mu} \Pi_\mu,
\end{align}
the scalar matrix obeying U$(1)$ symmetry given by scalars $\tr_\mu(Z_r)/\dim S^\mu$ with projections $\Pi_\mu$ into each $\U(1)$ charge sector. It is also important to note that
\begin{align}\label{eq:U(1)phase-1}
	C_{Z_r} = C_{Z_s} = \frac{1}{n} \sum_i Z_i \equiv C_1 \text{ for any } s \neq r,
\end{align}
because $Z_r$ is similar to other single Pauli-$Z$'s. Similarly we define
\begin{align}\label{eq:U(1)phase-2}
	C_{Z_k Z_l} \vcentcolon = \bigoplus_\mu \frac{\tr_\mu(Z_k Z_l)}{\dim S^\mu} \Pi_\mu = \frac{1}{n(n-1)} \Big[ \Big( \sum_i Z_i \Big)^2 - nI \Big] = \frac{2}{n(n-1)} \sum_{i < j} Z_i Z_j \equiv C_2.
\end{align}
as well as $C_3,...,C_n$ in general. Plus the identity matrix, $\{I = C_0,C_1,...,C_n\}$ is a collection of orthogonal operators because $\tr(C_i C_j) = 0$ for $i \neq j$. It spans all relative phase factors with respect to the $\U(1)$ symmetry.

On the other hand, for SWAPs $(k,l)$, we consider
\begin{align}
	C_{(k,l)} \vcentcolon = \bigoplus_\mu \frac{\text{Tr}_\mu (k,l)}{\dim S^\mu} \Pi_\mu.
\end{align}
The definition is also independent to the choice of $k, l$. Expanding SWAPs by Pauli basis, we note that since
\begin{align}
	(k,l) = \frac{1}{2} \sum_{P  = I,X,Y,Z} P_k \otimes P_l^\dagger = \frac{1}{2} \sum_{P = I,X,Y,Z} P_k \otimes P_l
\end{align}
and since $\bra{\psi} (X_i X_j + Y_i Y_j) \ket{\psi} = 0$ for any computational basis element $\ket{\psi}$,
\begin{align}
	C_{(k,l)} = \frac{1}{2}(I + C_{Z_k Z_l}) \equiv \frac{1}{2}(C_0 + C_2).
\end{align}

Let $\mathcal{V}_{2, \U(1)}$ denote the group generated by any $2$-local $\U(1)$-symmetric unitaries. This is a compact subgroup of $\mathcal{U}_{ \U(1)}$ whose Lie algebra $\mathfrak{v}_{2, \U(1)}$ is generated by $\U(1)$-symmetric Hamiltonians supported on arbitrary two locations~\cite{MarvianNature,marvian2023nonuniversality,U(1)Design2023}. Suppose $H_{k_0 l_0}$ is an arbitrary $2$-local Hamiltonian acting on the $k_0$-th and $l_0$-th sites. By Theorem \ref{thm:universalityU(1)-1} and the fact that a 2-qubit system has exactly three independent $\U(1)$ relative phases, $H_{k_0 l_0}$ can thus be written as a linear combination of $(k_0,l_0), Z_{k_0}, Z_{l_0}, Z_{k_0} Z_{l_0}$ as well as their Lie brackets. Let $L$ denote these brackets, then
\begin{align}\label{eq:2localHamiltonian}
	H_{k_0 l_0} = c_0 I + c_1 (k_0,l_0) + c_2 Z_{k_0} + c_3 Z_{l_0} + c_4 Z_{k_0} Z_{l_0} + L.
\end{align} 
Since Lie brackets are traceless, $(\tr_\mu (H_{k_0 l_0})/\dim S^\mu)$ is totally determined by $I, C_1, C_2$, which helps to justify:

\begin{theorem}\label{thm:universalityU(1)-2}
	The Lie algebra 
	\begin{align}
		\langle i\mathfrak{d}_{Z_kZ_l}(\mathbb{R}), iH_S \rangle_\mathbb{R} = \bigoplus_{\mu = (n-r,r)} \mathfrak{su}(S^\mu) \oplus \text{span}\{iC_0, iC_1, iC_2\}_{\mathbb{R}} = \mathfrak{v}_{2, \U(1)}.
	\end{align}
	Moreover, since $\{C_0,C_1,C_2,...,C_n\}$ spans all scalar matrices/relative phase with respect to the $\U(1)$ symmetry, 
	\begin{align}\label{eq:eq:universalityU(1)}
		S\mathcal{U}_{\U(1)} \subsetneqq \CQA_{\U(1)} = \mathcal{V}_{2,\U(1)} \subsetneqq \mathcal{U}_{\U(1)}.
	\end{align}
\end{theorem}

\begin{remark}
	As a comparison to $\SU(d)$ case, Eq.~\eqref{eq:U(1)phase-1}, \eqref{eq:U(1)phase-2} and all other orthogonal scalar matrices used to create $\U(1)$ relative phases are replaced by those in Proposition \ref{prop:CenterBasis} and Theorem \ref{thm:CenterBases}, which use (a) permutations of the same cycle type, (b) summing over $S_n$ weighted by characters and (c) products of YJM elements. The technique involving partial root vectors does not hold on the first choice and the second choice loses locality immediately. Therefore, we consider products of the YJM element. As a reminder, YJM elements do not produce orthogonal diagonal matrices in general and more delicate analyses using $S_n$ representation theory are required in proving the counterpart of Theorem \ref{thm:universalityU(1)-2} for $\SU(d)$ symmetry~\cite{Zheng2021SpeedingUL,SUd-k-Design2023}. Moreover, $\CQA_{\U(1)} = \mathcal{V}_{2,\U(1)}$ here, but $\CQA_{\SU(d)}$ is a proper subgroup of $\mathcal{V}_{4,\SU(d)}$ in \eqref{eq:universalitySU(d)} because there are relative phases corresponding to 4-local permutations which cannot be generated by $\CQA_{\SU(d)}$ (see~\cite{Zheng2021SpeedingUL} for more details). 
\end{remark}

Let $\mathcal{V}_{k,\U(1)}$ be the group generated by $k$-local $\U(1)$-symmetric unitaries. Using a similar argument like expanding Eq.~\eqref{eq:2localHamiltonian}, we obtain the following corollary, which recovers the $\U(1)$ semi-universality proved in Refs.~\cite{MarvianNature,marvian2023nonuniversality}.

\begin{corollary}\label{corollary:universalityU(1)}
	Let $\mathfrak{d}_k(\mathbb{R})$ denote the collection of all real diagonal matrices spanned by up to $k$-th Pauli-$Z$ matrices. The Lie algebra of $\mathcal{V}_{k,\U(1)}$ satisfies
	\begin{align}
		\mathfrak{v}_{k, \U(1)} = \langle i\mathfrak{d}_{k}(\mathbb{R}), iH_S \rangle_\mathbb{R} = \bigoplus_\mu \mathfrak{su}(S^\mu) \oplus \text{span}\{iC_0, iC_1,...,iC_k\}_{\mathbb{R}}.
	\end{align} 
	Therefore, an arbitrary ensemble of $\U(1)$-symmetric unitaries can never be completely universal if it has a bounded locality. 
\end{corollary}


\subsection{$\U(1)$-symmetric unitary $k$-designs}\label{sec:U(1)designs}

We have proved in our earlier work Ref.~\cite{SUd-k-Design2023} that $\CQA_{\SU(d)}$ forms $\SU(d)$-symmetric exact $k$-designs for all $k < n(n-3)/2$. Here, we will show an analogous result that $\CQA_{\U(1)} = \mathcal{V}_{2,\U(1)}$ forms $\U(1)$-symmetric exact $k$-designs for all $k < 2n-2$. This result is known and generalized to circuits with higher locality in Refs.~\cite{MarvianDesign,mitsuhashi2024Designs,mitsuhashi2024Designs2}. Later, we demonstrate that it is impossible for ensembles of $\U(1)$-symmetric unitaries with any bounded locality to form $\U(1)$-symmetric $k$-designs for arbitrarily large $k$, let alone converge to $\U(1)$-symmetric Haar measure.

Most of the proof details here are similar and even simpler than the $\SU(d)$ case presented in~\cite{SUd-k-Design2023} because we are working with the computational basis rather than the Young basis defined by $S_n$ representation theory and $\U(1)$ charge sectors are multiplicity-free on qubits. With these considerations, we only elaborate on the crucial steps.

By Theorem \ref{thm:universalityU(1)-2}, we have
\begin{align}\label{eq:A-TkCQAU(1)}
	T_{k,\U(1)}^{\CQA} & = \int_{\gamma_j} (e^{-i (\gamma_0 C_0 + \gamma_2 C_2 + \gamma_2 C_2) })^{\otimes k} \otimes (e^{i (\gamma_0 C_0 + \gamma_2 C_2 + \gamma_2 C_2)} )^{\otimes k} d\gamma \int_{S\mathcal{U}_{\U(1)}} V^{\otimes k} \otimes \bar{V}^{\otimes k} dV.
\end{align}
because integrals on the RHS are commutative and the intersection of their unit eigenspaces is identical to that of $\CQA_{\U(1)}$ (see Section \ref{sec:kDesigns}). To prove that $\CQA_{\U(1)}$ forms $\U(1)$-symmetric unitary $k$-designs for some $k$, we need to check whether $T_k^{\CQA_{\U(1)}} = T_k^{\mathcal{U}_{\U(1)}}$ or not. To this end, we first study the integral of phases $C_i$ and that of $S\mathcal{U}_{\U(1)}$ separately. 

We expand the integrand from
\begin{align}
	T_{k,\U(1)}^{S\mathcal{U}_\times} = \int_{S\mathcal{U}_{\U(1)}} V^{\otimes k} \otimes \bar{V}^{\otimes k} dV = \int_{S\mathcal{U}_{\U(1)}} \Big( \bigoplus_\mu V_\mu \Big)^{\otimes k} \otimes \Big( \bigoplus_\nu \bar{V}_\nu \Big)^{\otimes k} dV
\end{align}
into the direct sum of tensor products on several charge sectors. We denote by $S^{(-n)}$ the 1-dimensional sector corresponding to the polarized down state. As the simplest example, when $k = 1$ we have
\begin{align}\label{eq:TrivialExample}
	\int_{\SU(S^{(n)})} V_{(n)} dV_{(n)} \otimes \int_{\SU(S^{(-n)})} \bar{V}_{(-n)} dV_{(-n)} 
	& \neq 0 = \int_{\U(S^{(n)})} U_{(n)} dU_{(n)} \otimes \int_{\U(S^{(-n)})} \bar{U}_{(-n)} dU_{(-n)} \\
	\implies T_{1,\U(1)}^{S\mathcal{U}_\times} & \neq T_{1,\U(1)}^{\mathcal{U}_\times}.
\end{align}
That is, $S\mathcal{U}_{\U(1)}$ cannot even be unitary $1$-design with respect to $\mathcal{U}_{\U(1)}$. The reasons is: both $S^{(n)}$ and $S^{(-n)}$ are $1$-dimensional charge sectors, but they are inequivalent. Consequently, we have two independent nonzero integrals over $\SU(S^{(n)})$ and $\SU(S^{(-n)})$ $\cong \SU(1) = \{1\}$ on the LHS from above. A similar situation happens on the RHS, but the integrals over $\U(S^{(n)})$ and $\U(S^{(-n)})$ $\cong \U(1)$ is now zero. Multiplying with the phases integrals of $C_i$ from Eq.~\eqref{eq:A-TkCQAU(1)} would alleviate such problems, which we summarizes in Theorem \ref{thm:CQAU(1)-design} after the following lemma:

\begin{lemma}\label{lemma:LR}\cite{SUd-k-Design2023}
	Let $0 \leq r \neq s < d$. The tensor product $V^{\otimes r} \otimes \bar{V}^{\otimes s}$ for $V \in \SU(d)$ cannot be decomposed into the trivial representation. In this case, $\int_{\SU(d)} V^{\otimes r} \otimes \bar{V}^{\otimes s} dV = 0$. It is nonzero when $r + (d-1)s$ can be divided by $d$. 
\end{lemma}

In contrast, $\int_{\U(d)} U^{\otimes r} \otimes \bar{U}^{\otimes s} dU = 0$ unless $r = s$ and there is no further conditions.

\begin{theorem}\label{thm:CQAU(1)-design}
	For an $n$-qubit system, $\CQA_{\U(1)} = \mathcal{V}_{2,\U(1)}$ forms an exact $\U(1)$-symmetric $k$-design with respect to $\mathcal{U}_{\mathrm{U}(1)}$ for $k$ being at most $2n-3$.
\end{theorem}
\begin{proof}	
	It is well-known that, except $S^{(n)},S^{(-n)},S^{(n-1,1)},S^{(1,n-1)}$, the lowest dimension of $\U(1)$ charge sector is $\binom{n}{2} = n(n-1)/2$ and by Lemma \ref{lemma:LR} 
	\begin{align}
		\int_{\SU(S^\mu)} V_\mu^{\otimes r} \otimes \bar{V}_\mu^{\otimes s} dV_\mu = \int_{\U(S^\mu)} U_\mu^{\otimes r} \otimes \bar{U}_\mu^{\otimes s} dU_\mu
	\end{align}
	if $r,s < n(n-1)/2$ and $\dim S^\mu \geq n(n-1)/2$. Inconsistency like \eqref{eq:TrivialExample} arises when we integrate over sectors with even lower dimension, only including $S^{(n)},S^{(-n)},S^{(n-1,1)},S^{(1,n-1)}$. Then we need to verify that the phase integral as in Eq.~\ref{eq:A-TkCQAU(1)} would remedy the problem when $k < 2n - 2$.
	
	With respect to these four sectors, let us denote by $C_i^{\mu}, i = 0,1,2$ the phase factors of given by $C_i$ in CQA respectively. Direct computations show that
	\begin{align}\label{eq:U(1)phase-examples}
		& C_0^{\mu} \equiv 1, \\ 
		& C_1^{(n)} = 1, \quad C_1^{(-n)} = -1, \quad C_1^{(n-1,1)} = \frac{n-2}{n}, \quad C_1^{(1,n-1)} = -\frac{n-2}{n}, \\
		& C_2^{(n)} = 1, \quad C_2^{(-n)} = \mathmakebox[6mm]{1,} \quad C_2^{(n-1,1)} = \frac{n-4}{n}, \quad C_2^{(1,n-1)} = \frac{n-4}{n},
	\end{align}
	Let 
	\begin{align}
		(r_{(n)}, r_{(-n)}, r_{(n-1,1)}, r_{(1,n-1)}, \quad s_{(n)}, s_{(-n)}, s_{(n-1,1)}, s_{(1,n-1)} )
	\end{align}
	denote the orders of tensor product of unitaries $V_\mu$ and $\bar{V}_\mu$ on the aforementioned four sectors expanded from $V^{\otimes k} \otimes \bar{V}^{\otimes k}$, so they are non-negative and $\sum_\mu r_\mu = \sum_\mu s_\mu = k$.
	
	Let $\Delta k_\mu = r_\mu - s_\mu$. We now move to verify that when $k < 2n - 2$, there is \emph{no} $r_\mu, s_\mu$ defined above such that the following conditions hold simultaneously:
	\begin{align}
		& \text{at least one } \Delta k_\mu \neq 0, \\
		& \Delta k_\mu \equiv 0, \mod \dim S^\mu \text{ for all } \mu, \\
		& \Delta k_{(n)} - \Delta k_{(-n)} + \frac{n-2}{n} \Delta k_{(n-1,1)} - \frac{n-2}{n} \Delta k_{(1,n-1)} = 0, \\
		& \Delta k_{(n)} + \Delta k_{(-n)} + \frac{n-4}{n} \Delta k_{(n-1,1)} + \frac{n-4}{n} \Delta k_{(1,n-1)} = 0, \\
		& \Delta k_{(n)} + \Delta k_{(-n)} + \Delta k_{(n-1,1)} + \Delta k_{(1,n-1)} = 0.
	\end{align} 
	To be precise, the first two conditions indicate that $T_{1,\U(1)}^{S\mathcal{U}_\times} \neq T_{1,\U(1)}^{\mathcal{U}_\times}$ by Lemma \ref{lemma:LR}. The last three says that the phase integral is nonzero. All of them indicate that the moment operator of CQA, when restricted to these sectors, cannot be identical to that of $\mathcal{U}_{ \U(1)}$, forbidding the formation of unitary $k$-designs.  
	
	Assume the first two conditions hold, otherwise the proof is complete. Then according to Lemma \ref{lemma:LR}, we have 
	\begin{align}
		\bigotimes_{\mu} \int_{\SU(S^\mu)} V_\mu^{\otimes r_\mu} \otimes \bar{V}_\mu^{\otimes s_\mu} dV_\mu \neq 0 = \bigotimes_{\mu} \int_{\U(S^\mu)} U_\mu^{\otimes r_\mu} \otimes \bar{U}_\mu^{\otimes s_\mu} dU_\mu,
	\end{align}
    where $S^\mu$ is taken to be  $S^{(n)},S^{(-n)},S^{(n-1,1)},S^{(1,n-1)}$. Assume the last thee conditions also hold. Solving them, we obtain
    \begin{align}
    	\Delta k_{(n)} = - \Delta k_{(-n)} = \frac{\Delta k_{(1,n-1)} (n-2) }{n}, \quad \Delta k_{(n-1,1)} = - \Delta k_{(1,n-1)}.
    \end{align}
    Suppose $\Delta k_{(1,n-1)} > 0$. Note the second condition implies that $\Delta k_{(1,n-1)} = cn$ for a certain integer $c > 0$. Then
    \begin{align}
    	k = r_{(n)} + r_{(-n)} + r_{(n-1,1)} + r_{(1,n-1)} = (s_{(n)} + c(n-2)) + r_{(-n)} + r_{(n-1,1)} + (s_{(1,n-1)} + cn)
    \end{align}
	achieves its minimum when $s_{(n)} = r_{(-n)} = r_{(n-1,1)} = 0$ and $c = 1$, which is $2n - 2$. With respect to which,
	\begin{align}
		\hspace{-1em} \int_{\CQA_{\U(1)}} W_{(n)}^{\otimes n-2} \otimes W_{(1,n-1)}^{\otimes n} \otimes \bar{W}_{(-n)}^{\otimes n-2} \otimes \bar{W}_{(n-1,1)}^{\otimes n} dW \neq 0 = 
		\int_{\mathcal{U}_{\U(1)}} U_{(n)}^{\otimes n-2} \otimes U_{(1,n-1)}^{\otimes n} \otimes \bar{U}_{(-n)}^{\otimes n-2} \otimes \bar{U}_{(n-1,1)}^{\otimes n} dU.
	\end{align}
	This verifies $k$ can be at most $2n-3$.
\end{proof}

The following corollary is proved by comparing the commutant of the moment operator $T^{\CQA}_{k,\U(1)}$ and the unit eigenspace of $T_{k,\U(1)}^{\mathcal{E}_{\CQA}}$ of the ensemble $\mathcal{E}_{\CQA,\U(1)}$ defined in Definition \ref{def:CQAEnsemble}. The proof details are similar to those for the $\SU(d)$ case which can be found in Ref.~\cite{SUd-k-Design2023}, so we simply state the result:

\begin{corollary}
	For an $n$-qubit system, the ensemble $\mathcal{E}_{\CQA,\U(1)}$ can converge to $k$-design to with respect to $\mathcal{U}_{ \U(1)}$ for $k$ being at most $2n - 3$.
\end{corollary}


\begin{theorem}\label{thm:U(1)DesignLocality} 
	Given a finite or infinite ensemble $\mathcal{E}$ of $\U(1)$-symmetric unitary gates, in order to generate $k$-designs under $\U(1)$ symmetry for $k \to \infty$ in either exact or approximate senses, $\mathcal{E}$ must contain unitaries acting on the whole $n$-qubit system. 
\end{theorem}
\begin{proof}
	We first assume that $\mathcal{E} = \mathcal{V}_{\gamma,\U(1)}$ with locality $\gamma \geq 2$ and decompose its Lie algebra into $\bigoplus_\mu \mathfrak{su}(S^\mu) \oplus \text{span}\{iC_0, iC_1,...,iC_\gamma\}_{\mathbb{R}}$ by Corollary \ref{corollary:universalityU(1)}. Then we study the general case. Similar to Eq.~\eqref{eq:A-TkCQAU(1)},
	\begin{align}
		\begin{aligned}
			T_{k,\U(1)}^{\mathcal{V}_\gamma} = & \int_{\gamma_j} (e^{-i \sum_j \gamma_j C_j})^{\otimes k} \otimes (e^{i \sum_j \gamma_j C_j} )^{\otimes k} d\gamma \int_{S\mathcal{U}_{\U(1)}} V^{\otimes k} \otimes \bar{V}^{\otimes k} dV,
		\end{aligned}
	\end{align}
	where $C_i$ are defined by trace of tensor products of Pauli-$Z$'s like Eq.~\eqref{eq:U(1)phase-1} \& \eqref{eq:U(1)phase-2}.
	
	We use a similar strategy as in Theorem \ref{thm:CQAU(1)-design}. Let $S^{\mu_i}$ with $\mu_i = (n-i,i)$ and $i = 0,1,...,n$ denote all $\U(1)$ charge sectors and let
	\begin{align}
		(r_0,\ldots,r_n, s_0,\ldots,s_n)
	\end{align}
	denote the orders of tensor product of unitaries $V_i$ and $\bar{V}_i$ on the block $S^{\mu_i}$ expanded from $V^{\otimes k} \otimes \bar{V}^{\otimes k}$. Then $r_i, s_i$ are non-negative and $\sum_i r_i = \sum_i s_i = k$. By Lemma \ref{lemma:LR}, if
	\begin{align}\label{eq:LRCondition}
		r_i \equiv s_i \mod d_i = \dim S^{\mu_i} = \binom{n}{i} \text{ for all } i, \text{ and } r_{i_0} \neq s_{i_0} 
	\end{align}
	for at least one $i_0$, then the integral is nonzero over $S\mathcal{U}_{\U(1)}$ but still vanishes over $\mathcal{U}_{\U(1)}$. 
	
	Then we show that the phase integral involving $e^{-i \sum_j \gamma_j C_j}$ is also nonzero for some large $k$, which implies that $T_{k,\U(1)}^{\mathcal{V}_\gamma} \neq T_{k,\U(1)}^{\mathcal{U}_{\times}}$. By definition, multiplying with the orders $r_i, s_i$ of tensor products, these phases are combined as the following matrix product
	\begin{align}\label{eq:LinearSystem}
		\begin{pmatrix} r_0 - s_0, & r_1 - s_1, & \cdots, & r_n - s_n \end{pmatrix} \begin{pmatrix}
			C^{\mu_0}_0 & \cdots & C^{\mu_0}_\gamma \\ C^{\mu_1}_0 & \cdots & C^{\mu_1}_\gamma \\ \vdots & & \vdots \\ C^{\mu_n}_0 & \cdots & C^{\mu_n}_\gamma
		\end{pmatrix},
	\end{align} 
	where $C^{\mu_i}_j$ is the trace (relative phase) of $C_j$ when restricted to $S^{\mu_i}$ as we defined in Eq.~\eqref{eq:U(1)phase-examples}. If each combination coefficient is zero, taking the exponential only yields the identity matrix and the phase integral ends up being nonzero. This situation happens if we can find a {nontrivial integral solution} $\Delta k_i = r_i - s_i$ such that \eqref{eq:LinearSystem} equals $0$. As a reminder, the condition $\sum_i r_i - \sum_i s_i = 0$ above Eq.~\eqref{eq:LRCondition} in defining $(r_i,s_i)$ is already included as the first equation since $C_0 = I$. 
	
	As we discussed in the vicinity of Eq.~\eqref{eq:U(1)phase-2}, $C^{\mu_i}_j$ are rational numbers and if $\gamma = n$, the above matrix is orthonormal. Therefore, nontrivial solutions $(\Delta k_i)$ always exist when $\gamma < n$. Let
	\begin{align}
		r_i = \begin{cases} d_0 \cdots d_n \cdot \Delta k_i, & \Delta k_i > 0 \\ 0, & \Delta k_i \leq 0 \end{cases},   \quad  
		s_i = \begin{cases} 0, & \Delta k_i > 0 \\ -d_0 \cdots d_n \cdot \Delta k_i, & \Delta k_i \leq 0		\end{cases}.
	\end{align}
	This satisfies all the above conditions including \eqref{eq:LRCondition} with 
	\begin{align}
		d_0 \cdots d_n \leq k \leq \big( \sum_i \vert \Delta k_i \vert \big) d_0 \cdots d_n \text{ and } T_{k,\U(1)}^{\mathcal{V}_\gamma} \neq T_{k,\U(1)}^{\mathcal{U}_{\times}}.
	\end{align}
	
	Given a generic $\U(1)$-symmetric ensemble $\mathcal{E}$ of $\gamma$-local unitaries, it must be contained in $\mathcal{V}_{\gamma,\U(1)}$ by definition. According to our earlier proof, the locality of $\mathcal{V}_{\gamma,\U(1)}$, and hence of $\mathcal{E}$ must be $n$ in order to form a $\U(1)$-symmetric unitary $k$-design for arbitrarily large $k$.
\end{proof}


\newpage
\section{Spectral gap of unitary $2$-designs under $\SU(d)$ and $\U(1)$ symmetries}\label{sec:sketches}

We now introduce the proof strategy for evaluating the spectral gap of the second moment operator of CQA ensembles under both $\U(1)$ and $\SU(d)$ symmetry, which is used to bound their convergence time to a symmetric unitary $2$-design. Detailed proofs are postponed to Section \ref{sec:details}.

By Definition \ref{def:CQAEnsemble}, the second moment operators of interest are defined as
\begin{align}
	T^{\mathcal{E}_{\CQA}}_{2,\SU(d)} & = \frac{1}{\vert \mathcal{T} \vert} \sum_{\tau \in \mathcal{T}} T^{\YJM}_2 T^\tau_2 T^{\YJM}_2, \label{eq:TkCQASU(d)} \\
	T^{\mathcal{E}_{\CQA}}_{2,\mathrm{U(1)}} & = \frac{1}{\vert \mathcal{T} \vert} \sum_{\tau \in \mathcal{T}} T^{Z}_2 T^\tau_2 T^{Z}_2, \label{eq:TkCQAU(1)}
\end{align}
where $\tau$ are selected from a generating set $\mathcal{T} \subset S_n$ of transpositions/SWAPs. To study the spectral gap, more detailed expansions of $T^{\YJM}_2, T^{Z}_2 T^\tau_2$ are necessary and can be found in Section \ref{sec:detailsT2CQA}.

As discussed in Section \ref{sec:kDesigns}, decompositions of the Hilbert space under certain symmetries is always fraught with inequivalent irrep/charge sectors. The most straightforward way to evaluate the spectral gap of  $T^{\mathcal{E}_{\CQA}}_{2,\SU(d)},T^{\mathcal{E}_{\CQA}}_{2,\U(1)}$  would be to exhaust all possible tensor products of the sectors (with multiplicities):
\begin{align}
	(S^{\lambda_1} \otimes \mathrm{1}_{m_{\lambda_1}}) \otimes (S^{\lambda_2} \otimes \mathrm{1}_{m_{\lambda_2}}) \otimes (S^{\lambda_3} \otimes \mathrm{1}_{m_{\lambda_3}}) \otimes (S^{\lambda_4} \otimes \mathrm{1}_{m_{\lambda_4}})
\end{align}
for $2$-designs, expanding integrands of Eq.~\eqref{eq: tpe} for $T^{\YJM}_2, T^{Z}_2 T^\tau_2$ within these subspaces and evaluating the spectral gaps one by one. The procedure looks formidable, but fortunately we prove in Lemma \ref{lemma:CQAk} in Section \ref{sec:detailsT2CQA} that, under $\SU(d)$ symmetry, 
\begin{align}\label{eq:YJMPorjection}
	T^{\YJM}_2 E\indices{^a_b}E\indices{^c_d} = T^{\YJM}_2 \ket{a,c}\bra{b,d} = 0 
\end{align}
unless $a = c, b = d$ denoting basis vectors from $S_n$ irrep sectors up to multiplicities. More explicitly, $T^{\YJM}_2$ rules out all fuzzy cases, allowing us to focus on the expansion of $T^{\mathcal{E}_{\CQA}}_{2,\SU(d)}$ within  
\begin{align}
	& (S^{\lambda} \otimes \mathrm{1}_{m_{\lambda}}) \otimes (S^{\mu} \otimes \mathrm{1}_{m_{\mu}}) \otimes (S^{\lambda} \otimes \mathrm{1}_{m_{\lambda}}) \otimes (S^{\mu} \otimes \mathrm{1}_{m_{\mu}}), \tag{Type 1} \label{eq:Case1} \\
	& (S^{\lambda} \otimes \mathrm{1}_{m_{\lambda}}) \otimes (S^{\mu} \otimes \mathrm{1}_{m_{\mu}}) \otimes (S^{\mu} \otimes \mathrm{1}_{m_{\mu}}) \otimes (S^{\lambda} \otimes \mathrm{1}_{m_{\lambda}}), \tag{Type 2} \label{eq:Case2} \\
	& (S^{\lambda} \otimes \mathrm{1}_{m_{\lambda}}) \otimes (S^{\lambda} \otimes \mathrm{1}_{m_{\lambda}}) \otimes (S^{\lambda} \otimes \mathrm{1}_{m_{\lambda}}) \otimes (S^{\lambda} \otimes \mathrm{1}_{m_{\lambda}}). \tag{Type 3} \label{eq:Case3}
\end{align}
It is not difficult to deal with the first two cases, which we present in Section \ref{sec:detailsC1C2}. We now outline the method to tackle the last case. A similar method also works for the $\U(1)$ case, thus we will not explicitly differentiate these two symmetries when no confusion can arise. Notations like $T^{\mathcal{E}_{\CQA}}_{2,\SU(d)},T^{\mathcal{E}_{\CQA}}_{2,\mathrm{U(1)}}$ are all abbreviated as $T^{\mathcal{E}_{\CQA}}_2$ for simplicity. Sectors are all denoted by $S^\lambda, S^\mu$ and etc. A clarification of the differences in the proofs can be found in Section \ref{sec:detailsSU(d)}.

Taking one sector $S^\lambda$ from its equivalent copies, we just consider $T_2^{\mathcal{E}_{\CQA}} \vert_{(S^{\lambda})^{\otimes 4}}$ for the operator should have the same matrix representation on other copies of $S^\lambda$. As introduced in Example \ref{example:InvariantSubK=2}, we can decompose $(S^{\lambda})^{\otimes 4}$ into four subspaces invariant under the action of $T_2^{\mathcal{E}_{\CQA}}$ by projectors defined in Eq.~\eqref{eq: projector-invariant-subspace}. It is proved in Section \ref{sec:detailsT2CQA} that the moment operator $T_2^{\YJM}$ (or $T_2^Z$) further projects and filters basis elements with homogeneous indices in each subspace as listed below (cf. Eq.~\eqref{eq:ExampleBasis}):
\begin{align}
	& S_{\sbA} = \operatorname{span} \left\{ \frac{1}{2} \Big( E\indices{^a_a} E\indices{^b_b} + E\indices{^b_b} E\indices{^a_a}) +  (E\indices{^a_b} E\indices{^b_a} + E\indices{^b_a} E\indices{^a_b} \Big), E\indices{^a_a} E\indices{^a_a}; a < b  \right\} \subset \Ima\mathcal{P}_{\sbA} \subset (S^{\lambda})^{\otimes 4}, \label{eq:InvariantSubapce-states1} \\
	& S_{\sbD} = \operatorname{span} \left\{ \frac{1}{2} \Big( E\indices{^a_a} E\indices{^b_b} + E\indices{^b_b} E\indices{^a_a}) -  (E\indices{^a_b} E\indices{^b_a} + E\indices{^b_a} E\indices{^a_b} \Big); a < b \right\} \subset \Ima\mathcal{P}_{\sbD} \subset (S^{\lambda})^{\otimes 4}, \label{eq:InvariantSubapce-states2}
\end{align}
and
\begin{align}
	& S_{\sbB} = \operatorname{span} \left\{ \frac{1}{2} \Big( E\indices{^a_a} E\indices{^b_b} - E\indices{^b_b} E\indices{^a_a}) - (E\indices{^a_b} E\indices{^b_a} - E\indices{^b_a} E\indices{^a_b} \Big); a < b \right\} \subset \Ima\mathcal{P}_{\sbB} \subset (S^{\lambda})^{\otimes 4}, \label{eq:InvariantSubapce-states3} \\
	& S_{\sbC} = \operatorname{span} \left\{ \frac{1}{2} \Big( E\indices{^a_a} E\indices{^b_b} - E\indices{^b_b} E\indices{^a_a}) + (E\indices{^a_b} E\indices{^b_a} - E\indices{^b_a} E\indices{^a_b} \Big); a < b \right\} \subset \Ima\mathcal{P}_{\sbC} \subset (S^{\lambda})^{\otimes 4}. \label{eq:InvariantSubapce-states4}
\end{align}
Here $a,b,c,d...$ denote the Young basis elements of $S^\lambda$ with a prescribed ordering. 
In other words, $T_2^{\mathcal{E}_{\CQA}} \vert_{(S^{\lambda})^{\otimes 4}}$ is nontrivially supported over these four smaller subspaces. We denote the restrictions by  
\begin{align}
	T_{2,\sbA}^{\mathcal{E}_{\CQA}} = T_{2}^{\mathcal{E}_{\CQA}} \vert_{S_{\sbA}}, \quad
	T_{2,\sbB}^{\mathcal{E}_{\CQA}} = T_{2}^{\mathcal{E}_{\CQA}} \vert_{S_{\sbB}}, \quad
	T_{2,\sbC}^{\mathcal{E}_{\CQA}} = T_{2}^{\mathcal{E}_{\CQA}} \vert_{S_{\sbC}}, \quad
	T_{2,\sbD}^{\mathcal{E}_{\CQA}} = T_{2}^{\mathcal{E}_{\CQA}} \vert_{S_{\sbD}}.
\end{align}
This setup facilitates the application of most techniques performed in Section \ref{sec:details}.


\begin{definition}\label{def:Modified}
	In what follows, we mainly work with a \emph{modified} version of the CQA moment operator
	\begin{align}
		& M^{\mathcal{E}_{\CQA}}_{2,\SU(d)} \vcentcolon = \frac{1}{\vert \mathcal{T} \vert} \sum_{\tau \in \mathcal{T}} T^{\YJM}_2 M^{\tau}_2 T^{\YJM}_2, \\
		& M^{\mathcal{E}_{\CQA}}_{2,\U(1)} \vcentcolon = \frac{1}{\vert \mathcal{T} \vert} \sum_{\tau \in \mathcal{T}} T^{Z}_2 M^{\tau}_2 T^{Z}_2,
	\end{align}
	where the \emph{modified moment operators} $M^{\tau}_2$ of SWAPs are defined as
	\begin{align}
		M^{\tau}_2 = \frac{1}{8}(6IIII + I \tau I \tau + \tau I \tau I + \tau I I \tau + I \tau \tau I - II \tau \tau - \tau \tau II)
	\end{align}
	Here the tensor product notation, like $I \otimes I \otimes I \otimes I$, is omitted for brevity (cf. the expansion of $T_2^\tau$ in Eq.~\eqref{eq: Ttau-k2-integral}). 
\end{definition}

We replace the subscripts of $\SU(d)$ and $\U(1)$ by ``$\times$'' this section for conciseness.

\begin{lemma}\label{lemma:Modified}
	The modified second moment operator $M_2^\tau$ is a square root of $T_2^\tau$. That is $(M_2^\tau)^2 = T_2^\tau$. Moreover, the following properties hold:
	\begin{enumerate}
		\item Both $M^\tau_2$ and $ T^\tau_2$ share the same unit eigenspace. Let $\mathcal{T}$ be any generating set of SWAPs defined in Section \ref{sec:Aldous}, both $M_{2,\times}^{\mathcal{E}_{\CQA}}$ and $T_{2,\times}^{\mathcal{E}_{\CQA}}$ share the same unit eigenspace.  
		
		\item $M_{2,\times}^{\mathcal{E}_{\CQA}} \geq T_{2,\times}^{\mathcal{E}_{\CQA}}$ and $2 (I - M_{2,\times}^{\mathcal{E}_{\CQA}}) \geq I - T_{2,\times}^{\mathcal{E}_{\CQA}}$, hence
		\begin{align}
			2\Delta(M_{2,\times}^{\mathcal{E}_{\CQA}} ) \geq \Delta( T_{2,\times}^{\mathcal{E}_{\CQA}} ) \geq \Delta(M_{2,\times}^{\mathcal{E}_{\CQA}} ).
		\end{align}
		
		\item Within any fixed $(S^{\lambda})^{\otimes 4}$, the operator is still invariant under $S_{\sbA},S_{\sbB},S_{\sbC}$ and $S_{\sbD}$.
	\end{enumerate}
\end{lemma}

Based on this lemma, we turn to evaluate the spectral gap of $M^{\mathcal{E}_{\CQA}}_2$ and we can analyze situations separately in different invariant subspaces defined in Eq.~\eqref{eq:InvariantSubapce-states1} to \eqref{eq:InvariantSubapce-states4}. The conveniences will be revealed in Section \ref{sec:details}. A key observation is that within the first and last subspaces, the action of $M_2^{\mathcal{E}_{\CQA}}$ furnishes two stochastic matrices:

\begin{lemma}\label{lemma:CQA-stochastic-matrix}
	Within any fixed $(S^{\lambda})^{\otimes 4}$, the matrices representing the modified CQA second-moment operator $M_{2,\times}^{\mathcal{E}_{\CQA}}$, under bases given in Eq.~\eqref{eq:InvariantSubapce-states1} \& \eqref{eq:InvariantSubapce-states2}, are irreducible, doubly stochastic and symmetric within $S_{\sbA}$ and $S_{\sbD}$ respectively.
\end{lemma}

Based on this fact, we call basis elements from Eq.~\eqref{eq:InvariantSubapce-states1} \& \eqref{eq:InvariantSubapce-states2} \emph{basis states}. A casual change of basis may not guarantee that the resultant matrix representations are still stochastic, so we shall focus on these two bases. There are other basis states defined in Section \ref{sec:details} for technical purpose.

It is obvious to see by Schur--Weyl duality that $T_{2,\times}^{\mathcal{E}_{\CQA}}$, as well as the modified operator $M_{2,\times}^{\mathcal{E}_{\CQA}}$, admits 2 unit eigenvalues within $(S^{\lambda})^{\otimes 4}$. Each of which actually corresponds to the stationary distribution when we further restrict to $S_{\sbA}$ and $S_{\sbD}$. As a result, our proof procedure includes bounding the second largest eigenvalues from $S_{\sbA}$ and $S_{\sbD}$, and then the largest eigenvalues from $S_{\sbB}$ and $S_{\sbC}$. One crucial step to begin with is considering the Cayley moment operators defined as follows:

\begin{definition}
	The \emph{Cayley moment operators} under $\SU(d)$ and $\U(1)$ symmetries are: 
	\begin{align}\label{eq:Cayley-moment-operator}
		& \Cay_{2,\SU(d)} = \frac{1}{\vert \mathcal{T} \vert} \sum_{\tau \in \mathcal{T}} T_2^{\YJM} \frac{1}{8}( 6IIII + \tau I \tau I + I \tau I \tau) T_2^{\YJM}, \\
		& \Cay_{2,\U(1)} = \frac{1}{\vert \mathcal{T} \vert} \sum_{\tau \in \mathcal{T}} T_2^Z \frac{1}{8}( 6IIII + \tau I \tau I + I \tau I \tau) T_2^Z.
	\end{align}
	We also omit the tensor product notation (cf. the expansion of $T_2^\tau$ in Eq.~\eqref{eq: Ttau-k2-integral}). 
\end{definition}

In either the $\SU(d)$ or $\U(1)$ case, within any fixed $(S^{\lambda})^{\otimes 4}$, the Cayley moment operator admits a \emph{unique} unit eigenvector and its spectral gap can be inferred from that of the Cayley graph whose generators are given by $\mathcal{T}$ (see Theorem \ref{thm:CayleyGap} and discussions in Section \ref{sec:Aldous}). We note that $\Cay_{2, \times}$ can be transformed into a block-diagonal form with only one nontrivial invariant block, denoted by $\Cay_{2, \times, \Sym}$, containing its unit eigenvector. Significantly, $\Cay_{2, \times, \Sym}$ furnishes a reversible irreducible Markov transition operator \cite{Diaconis1988} whose state space is isomorphic to $\Ima \mathcal{P}_{\sbA}$. With all these preparation, we now apply the techniques from Section \ref{sec:Markov} to bound the eigenvalues. 

\begin{lemma}\label{lemma:CayletBound1}
	With respect to the system size $n$ and for some fixed $(S^{\lambda})^{\otimes 4}$, the spectral gap of the modified moment operator restricted to $S_{\sbA}$ has the same scaling with that of the Cayley moment operator:
	\begin{align}
		\Delta(M^{\mathcal{E}_{\CQA}}_{2, \times, \sbA}) = \Theta(\Delta(\Cay_{2,\times})).
	\end{align}
    \begin{proofsketch}
        By direct computations, we can show that $M^{\mathcal{E}_{\CQA}}_{2, \times, \sbA}$ is a doubly stochastic irreducible Markov transition operator. We denote Dirichlet forms of $M^{\mathcal{E}_{\CQA}}_{2, \times, \sbA}$ and $\Cay_{2, \times, \Sym}$ by $\mathcal{E}_{\sbA}(\cdot)$ and $\mathcal{E}_{\Cay}(\cdot)$, respectively. Then we prove in either $\U(1)$ or $\SU(d)$ cases, for $f \in \Ima \mathcal{P}_{\sbA}$:
        \begin{align}
        	\frac{1}{4}\mathcal{E}_{\sbA}(f) \leq \mathcal{E}_{\operatorname{Cay}}(f) \leq  \frac{d_\lambda + 1}{2d_\lambda} \mathcal{E}_{\sbA}(f) 
        \end{align}
        where $d_\lambda = \dim S^\lambda$. Comparisons between Dirichlet forms yields comparisons between spectral gaps \cite{chung1997spectral,Levin2009}, which allows us to conclude that $\Delta(M^{\mathcal{E}_{\CQA}}_{2, \times} {\sbA}) = \Theta(\Delta(\Cay_{2, \times}))$. 
    \end{proofsketch}
\end{lemma}

\begin{lemma}\label{lemma:CayletBound2}
	 For some fixed $(S^{\lambda})^{\otimes 4}$, the largest eigenvalue of $M^{\mathcal{E}_{\CQA}}_2$ restricted to $S_{\sbB} \oplus S_{\sbC}$ is no larger than the second largest eigenvalue of the Cayley moment operator:
	\begin{align}\label{eq:bound-vanishing-subspaces}
		\lambda_1(M^{\mathcal{E}_{\CQA}}_{2,\times, \ytableaushort{{} {*(black)}, {*(black)} {}}}) \leq \lambda_2(\Cay_{2,\times}).
	\end{align}
    \begin{proofsketch}
        We first consider $M^{\mathcal{E}_{\CQA}}_{2, \times}$ restricted to the off-diagonal invariant subspaces $\Ima \mathcal{P}_{\ytableaushort{{} {*(black)}, {*(black)} {}}}$. In Section \ref{sec:detailsU(1)} and \ref{sec:detailsSU(d)} we show that unit eigenvectors do not appear here, so we check its largest eigenvalue. It is easy to show that the unique unit eigenvector of $\Cay_{2, \times}$ within $S^{\lambda \otimes 4}$ is orthogonal to $\Ima \mathcal{P} {\ytableaushort{{} {*(black)}, {*(black)} {}}}$. Furthermore by direct computation, one can see that $M^{\mathcal{E}_{\CQA}}_{2, \times}$ and $\Cay_{2, \times}$ are identical when restricted to $\Ima \mathcal{P} {\ytableaushort{{} {*(black)}, {*(black)} {}}}$. It turns out that $\lambda_1(M^{\mathcal{E}_{\CQA}}_{2, \times, {\ytableaushort{{} {*(black)}, {*(black)} {}}}}) \leq \lambda_2(\Cay_{2, \times})$. 
    \end{proofsketch}
\end{lemma}

\begin{lemma}\label{lemma:induced-markov-process}
	With respect to the system size $n$ and for some fixed $(S^{\lambda})^{\otimes 4}$, the spectral gap of the modified moment operator restricted to $S_{\sbD}$ has a scaling no smaller than that of  the Cayley moment operator:
	\begin{align}
		\Delta(M^{\mathcal{E}_{\CQA}}_{2,\times,\sbD}) = \Omega(\Delta(\Cay_{2,\times})).
	\end{align}
    \begin{proofsketch}
    We note $\Ima \mathcal{P}_{\sbD} \ncong \Sym$ and it is no longer possible to compare the Cayley moment operator and $M^{\mathcal{E}_{\CQA}}_{2, \times, \sbD}$ directly. Fortunately, it turns out that  $\Delta(M^{\mathcal{E}_{\CQA}}_{2, \times, \sbD})$ can be lower-bounded by comparing with $M^{\mathcal{E}_{\CQA}}_{2, \times, \sbA}$ via the induced Markov chain defined in Section \ref{sec:Markov}. Intuitively, we can shrink $\Ima \mathcal{P}_{\sbA}$ to a subspace $\mathcal{X} \cong \Ima \mathcal{P}_{\sbD}$ by deleting certain states to build the induced chain $\widetilde{M}^{\mathcal{E}_{\CQA}}_{2, \times, \sbA}$. By classical results in Markov chain theory \cite{Levin2009}, $\Delta(\widetilde{M}^{\mathcal{E}_{\CQA}}_{2, \times, \sbA}) \geq \Delta(M^{\mathcal{E}_{\CQA}}_{2, \times, \sbA})$
    
    Then, we denote these transition operators with stationary distributions by $(\widetilde{P}_{\sbA},\widetilde{\pi})$ and $(P_{\sbD},\pi_{\sbD})$ acting on the state space $\mathcal{X}  \cong \Ima \mathcal{P}_{\sbD}$. For any pair of states $x, y \in \mathcal{X}$, it can be directly observed that $\widetilde{P}_{\sbA}(x,y) = \mathcal{P}_{\sbD}(x,y)$ or $\widetilde{P}_{\sbA}(x,y) > \mathcal{P}_{\sbD}(x,y) = 0$, which implies that $\Delta(\widetilde{M}^{\mathcal{E}_{\CQA}}_{2, \times, \sbA}) \geq \Delta(M^{\mathcal{E}_{\CQA}}_{2, \times, \sbD})$ by using Dirichlet forms similarly as the before. However, for the converse direction, there exist states $x, y \in \mathcal{X}$ such that $\widetilde{P}_{\sbA}(x, y) > 0$ while $P_{\sbD}(x, y) = 0$ so that the direct comparison with Dirichlet form is inconclusive. This motivates us to utilize the path-comparison theorem in spectral graph theory \cite{DiaconisComparison1993,Levin2009,Oliveira2design2007a,Oliveira2design2007b,Harrow2design2009}. In particular, whenever $\widetilde{P}_{\sbA}(x, y) > 0$ while $P_{\sbD}(x, y) = 0$, there exists a unique length-2 path $\gamma_{xy} = \{ (x, z), (z, y)\}$ such that $P_{\sbD}(x, y)|_{\gamma_{xy}} > 0$. However, if $\widetilde{P}_{\sbA}(x, y) > 0$ and $P_{\sbD}(x, y) > 0$ then we simply take the length-$1$ path $\gamma_{xy} = \{(x, y)\}$. Finally, we consider the congestion ratio (Definition \ref{def:Congestion})
    \begin{align}
    	A = \max_{(p, q), P_{\sbD}(p, q) > 0} \frac{1}{\pi_{\sbD}(p)P_{\sbD}(p, q)} \sum_{\substack{x, y \in \mathcal{X} \\ (p, q) \in \gamma_{xy}}} \widetilde{\pi}(x)\widetilde{P}_{\sbA}(x, y) |\gamma_{xy}|.
    \end{align}
    Since the length $|\gamma_{x, y}| \leq 2$ for all $x, y \in \mathcal{X}$, we can prove that $A = \Theta(1)$ under both $\SU(d)$ and $\U(1)$ symmetries. The path-comparison theorem indicates that $\Delta(\widetilde{M}^{\mathcal{E}_{\CQA}}_{2, \times, \sbA}) \leq A \Delta(M^{\mathcal{E}_{\CQA}}_{2, \times, \sbD})$. Therefore, we conclude the argument by combining the bounds from Lemma~\ref{lemma:CayletBound1}.
    \end{proofsketch}
\end{lemma}

Evidently, the spectral gap of $M_{2,\U(1)}^{\mathcal{E}_{\CQA}}$ is
\begin{align}
	\min_{\lambda} \Big\{ \Delta(M_{2,\times,\sbA}^{\mathcal{E}_{\CQA}}), \ \Delta(M_{2,\times,\sbD}^{\mathcal{E}_{\CQA}}), \ 1 - \lambda_1( M_{2,\times,\ytableaushort{ {} {*(black)} , {*(black)} {} }}^{\mathcal{E}_{\CQA}})   \Big\},
\end{align}
and $\Delta( \Cay_2 )$ can be solved by Theorem \ref{thm:CayleyGap} and our discussions in Section \ref{sec:Aldous}. Based on which we conclude that:

\begin{theorem}\label{thm:CQAGap}
	The spectral gap of the (modified) second-moment operator of CQA under both $\U(1)$ and $\SU(d)$ symmetries is $\Theta (\Delta(\Cay_{2,\times}))$. In particular, with $\U(1)$ symmetry:
	\begin{enumerate}
		\item If $\mathcal{T}$ is given by the nearest-neighbour SWAPs (1D chain), then $\Delta = \Omega(1/n^3)$.
		
		\item If $\mathcal{T}$ is given by SWAPs centred at some qudit (star), then $\Delta = \Omega(1/n)$.
		
		\item If $\mathcal{T}$ is given by the all-to-all interactions (complete graph), then $\Delta = \Omega(1/n)$.
	\end{enumerate}
	With $\SU(d)$ symmetry and nearest-neighbour SWAPs,
	$\Delta = \Omega(1/n^3)$.
\end{theorem}

The proof for $\SU(d)$ symmetry utilizes the notion of Young orthogonal forms introduced in Section \ref{sec:SnTheory}, which are defined only for nearest-neighbour SWAPs. The YJM elements are defined by summation of SWAPs, e.g., $X_n = (1,n) + \cdots + (n-1,n)$. To bound the spectral gap of general local random circuits, a simple application of Trotter-Lie formula \cite{Suzuki1991Trotter,Lloyd1996Trotter,Berry2006Trotter} can be applied 
\begin{align}
	\Big\Vert \Big(\exp i \frac{t}{m} (1,n) \cdots \exp i \frac{t}{m} (n-1,n) \Big)^m - \exp it X_n \Big\Vert \leq \frac{(n-1)^2 t^2}{m},
\end{align}
where the largest value of $t$ is $2 \pi$ (see Section \ref{sec:detailsT2CQA}). Let $m = \text{poly}(n)$, we can use operator norms to bound $T^{\YJM}_2$ and hence  $T^{\mathcal{E}_{\CQA}}_{2,\SU(d)}$ in Eq.~\eqref{eq:TkCQASU(d)} by moment operators defined using 4-local unitaries. Since there are only $n(n-1)/2$ many second-order products of YJM elements, the spectral gap associated with, e.g., the ensemble $\mathcal{E}_{\mathcal{V}_4}$ \cite{SUd-k-Design2023} defined by sampling arbitrary 4-local $\SU(d)$ symmetric unitaries should still be $O(1/\text{poly}(n))$.

Let $\mathcal{E}_{\mathcal{V}_r}$ be the ensemble defined by drawing arbitrary $r$-local $\SU(d)$ symmetric gates with fixed $r \geq 3$. With respect to system size $n$, the scaling of spectral gaps associated with $\mathcal{E}_{\mathcal{V}_4}$ and $\mathcal{E}_{\mathcal{V}_r}$ can only differ by constants. For instance, let $r = 3$ (see also Ref.~\cite{Marvian3local}). We can explicitly write down the 4-local Haar projector under $\SU(d)$ symmetry for any 4 sites $(i_1,i_2,i_3,i_4)$ of the system:
\begin{align}
	T_{2,\SU(d)}^{(i_1,i_2,i_3,i_4)} = \int U_{i_1,i_2,i_3,i_4}^{\otimes 2} \otimes \bar{U}_{i_1,i_2,i_3,i_4}^{\otimes 2} dU_{i_1,i_2,i_3,i_4}.
\end{align}
Within this subsystem, we can define four different 3-local Haar projectors under $\SU(d)$ symmetry, e.g.,
\begin{align}
	T_{2,\SU(d)}^{(i_1,i_2,i_3)} = \int U_{i_1,i_2,i_3}^{\otimes 2} \otimes \bar{U}_{i_1,i_2,i_3}^{\otimes 2} dU_{i_1,i_2,i_3}.
\end{align}
Since the localities are fixed, in the sense of positive semidefiniteness, there must be constants $C_1,C_2$ (independent of $n$, but possibly depending on $k$ if we consider general $k$-designs) such that
\begin{align}
	C_1 (I - T_{2,\SU(d)}^{(i_1,i_2,i_3,i_4)}) \leq I - \frac{1}{4}( T_{2,\SU(d)}^{(i_1,i_2,i_3)} + T_{2,\SU(d)}^{(i_1,i_2,i_4)} + T_{2,\SU(d)}^{(i_1,i_3,i_4)} + T_{2,\SU(d)}^{(i_2,i_3,i_4)} ) \leq C_2 (I - T_{2,\SU(d)}^{(i_1,i_2,i_3,i_4)} ).
\end{align}
As a caveat, if the ensembles do not form $k$-designs, e.g., $2$-local Haar projectors on qudits with $d \geq 3$ (see Section \ref{sec:SchurWeyl}), the inequality on the right can never hold. Summing over indices, we have
\begin{align}
	& C_1 \binom{n}{4} (I - T_{2,\SU(d)}^{\mathcal{E}_{\mathcal{V}_4}}) \leq \binom{n}{3}\frac{n-3}{4} (I - T_{2,\SU(d)}^{\mathcal{E}_{\mathcal{V}_3}}) \leq C_2 \binom{n}{4} (I - T_{2,\SU(d)}^{\mathcal{E}_{\mathcal{V}_4}}) \\
	\implies & \Delta( T_{2,\SU(d)}^{\mathcal{E}_{\mathcal{V}_3}}) ) = \Theta( \Delta(T_{2,\SU(d)}^{\mathcal{E}_{\mathcal{V}_4}})).
\end{align}
We leave the problem of finding more tight bounds on the spectral gaps of $\U(1)$- and $\SU(d)$-symmetric designs based on general circuit graphs as a future research opportunity.

\newpage 


\section{Proof of lower bounds on spectral gaps}\label{sec:details}

We now present the full proofs of the lower bounds on the spectral gaps of the CQA moment operators for both $\U(1)$ and $\SU(d)$ symmetries. For the subsequent proofs, we treat $\U(1)$ and $\SU(d)$ symmetries on the same footing by utilizing the similarity in their subspace decompositions, which will make the proof more concise for two different symmetries with different basis sets. As mentioned in Section \ref{sec:Tabloid}, we now unify the notation for the decompositions of the Hilbert space as 
\begin{align}
    \mathcal{H} = \bigoplus_\lambda S^\lambda \otimes \mathrm{1}_{m_\lambda}, 
\end{align}
where $\lambda$ specifies a sector would either be given by an $S_n$ irrep in Eq.~\eqref{eq:A-SchurWeyl} or may be reducible but given by the $\U(1)$ conservation law in Eq.~\eqref{eq:SpaceDecompositionU(1)}). The number $m_\lambda$ denote dimension of the multiplicity and we set $d_\lambda$ the dimension of the irreps. For the Abelian $\U(1)$ symmetry, the computational basis still respects the space decomposition in Eq.~\eqref{eq:SpaceDecompositionU(1)}. While for  $\SU(d)$, we need to change to the Young basis introduced in Section \ref{sec:SnTheory}. For the $\U(1)$ symmetry, $m_\lambda \equiv 1$ for all $\lambda$, which  greatly simplifies calculations in many applications such as the study of charge spreading~\cite{SUd-k-Design2023Application}. However, in evaluating eigenvalues instead of frame potentials (see Ref.~\cite{SUd-k-Design2023} for more details), nontrivial multiplicities in the $\SU(d)$ case do not bother us either because the moment operator defined on tensor products of different copies of given irreps always yields similar matrix representations with identical spectra.   


\subsection{Expansion of $T_{2,\times}^{\mathcal{E}_{\CQA}}$ and poof of Eq.~\eqref{eq:YJMPorjection}}\label{sec:detailsT2CQA}

Since any transposition/SWAP $\tau = (i,j) \in S_n$ satisfies $\tau^2 = I$, its unitary time evolution can be expanded by the Euler identity:
\begin{align}
	e^{-i\theta \tau} = \cos\theta I - i\sin\theta \tau. 
\end{align} 
Then
\begin{align}\label{eq:EulerIntegral}
	T_{k = 1}^\tau & = \frac{1}{2\pi} \int_0^{2\pi} e^{-i\theta \tau} \otimes e^{i\theta \tau} d\theta = \frac{1}{2\pi}  \int_0^{2\pi} (\cos\theta I - i\sin\theta \tau) \otimes (\cos\theta I + i\sin\theta \tau) d\theta = \frac{1}{2} (I \otimes I + \tau \otimes \tau),
\end{align} 
and
\begin{align}\label{eq: Ttau-k2-integral}
	\begin{aligned}
		T_{k = 2}^\tau = & \frac{1}{2\pi} \int_0^{2\pi} (e^{-i\theta \tau})^{\otimes 2} \otimes (e^{i\theta \tau})^{\otimes 2} d\theta = \frac{1}{2\pi} \int_0^{2\pi} (\cos\theta I + i\sin\theta \tau)^{\otimes 2} \otimes (\cos\theta I - i\sin\theta \tau)^{\otimes 2} d\theta  \\
		= & \frac{1}{2\pi} \int_0^{2\pi} \Big(\cos^4\theta I \otimes I \otimes I \otimes I + \sin^4\theta \tau \otimes \tau \otimes \tau \otimes \tau + \cos^2\theta \sin^2\theta \big(I \otimes \tau \otimes I \otimes \tau + I \otimes \tau \otimes \tau \otimes I \\
		& + \tau \otimes I \otimes I \otimes \tau + \tau \otimes I \otimes \tau \otimes I - I \otimes I \otimes \tau \otimes \tau - \tau \otimes \tau \otimes I \otimes I \big) \Big) d\theta \\
		= & \frac{1}{8} (3 IIII + 3\tau\tau\tau\tau + I\tau I\tau + I\tau\tau I + \tau I I\tau + \tau I\tau I - I I \tau\tau - \tau\tau I I ), 
	\end{aligned}
\end{align} 
where we omitted the tensor product notations in the last line for conciseness. The expansions of $T_k^\tau$ for general $k$ can be derived using $\sin^2\theta + \cos^2\theta = 1$ and
\begin{align}
	& \frac{1}{2\pi} \int_0^{2\pi} \cos^{2k}\theta d\theta = \frac{1}{2\pi} \int_0^{2\pi} \sin^{2k}\theta d\theta = \frac{1}{2^{2k}} \binom{2k}{k}, \\
	& \frac{1}{2\pi} \int_0^{2\pi} \cos^{2k-1}\theta \sin\theta d\theta = \frac{1}{2\pi} \int_0^{2\pi} \cos\theta \sin^{2k-1}\theta d\theta = 0.
\end{align} 


The moment operators of the time evolution of second-order YJM elements $X_r X_s$ and Pauli $Z$-matrices $Z_r Z_s$ are formally expressed as: 
\begin{align}
	& T^{\YJM}_2 = \int \Big( \exp(-i \sum_{r \leq s} \beta_{rs} X_r X_s) \Big)^{\otimes 2} \otimes \Big( \exp(i \sum_{r' \leq s'} \beta_{r's'} X_{r'} X_{s'}) \Big)^{\otimes 2} d\beta, \label{eq:TkYJM} \\
	& T^{Z}_2 = \frac{1}{(2\pi)^{n(n+1)/2}} \int \Big( \exp(-i \sum_{r \leq s} \beta_{rs} Z_r Z_s) \Big)^{\otimes 2} \otimes \Big( \exp(i \sum_{r' \leq s'} \beta_{r's'} Z_{r'} Z_{s'}) \Big)^{\otimes 2}d\beta, \label{eq:TkZ}
\end{align}
where the parameters are integrated over the uniform distribution. Since the eigenvalues of these operators are all integers (see Section \ref{sec:SnTheory} for more details), we only need to integrate over $[0,2\pi]$ for parameters $\beta_{rs}$ and $t$. Since $[X_r, X_s] = [Z_r, Z_s] = 0$, we can set $r \leq s$. For the same reason, 
\begin{align}
	\exp(-i \sum_{r \leq s} \beta_{rs} X_r X_s) = \prod_{r \leq s} \exp(-i \beta_{rs} X_r X_s) , \quad 
	\exp(-i \sum_{r \leq s} \beta_{rs} Z_r Z_s) = \prod_{r \leq s} \exp(-i \beta_{rs} Z_r Z_s), \tag{B56$^\ast$}
\end{align}
and the products can be taken in arbitrary orders.

Nonetheless, it is still difficult to expand $T_2^{\YJM}$ (Eq.~\eqref{eq:TkYJM}) and $T_2^Z$ (Eq.~\eqref{eq:TkZ}) directly. Fortunately, YJM elements and Pauli matrices are diagonalizable with integer eigenvalues (see Section \ref{sec:SnTheory} \& \ref{sec:U(1)Univeristy}). Like integrating over a compact group with respect to its Haar measure, these moment operators are defined to be orthogonal projections with only unit and zero eigenvalues. Actually, we have a  clear understanding on their unit eigenspaces, which suffices to tell us about their properties:
\begin{align}
	& \Comm_2(\YJM) = \Big\{ M \in \operatorname{End}(\mathcal{H}^{\otimes 2}); \forall \beta_{rs}, \ \Big[M, \Big( \exp(-i \sum_{r \leq s} \beta_{rs} X_r X_s) \Big)^{\otimes 2} \Big] = 0 \Big \}, \\
	& \Comm_2(Z) = \Big\{ M \in \operatorname{End}(\mathcal{H}^{\otimes 2}); \forall \beta_{rs}, \ \Big[M, \Big( \exp(-i \sum_{r \leq s} \beta_{rs} Z_r Z_s) \Big)^{\otimes 2} \Big] = 0 \Big \}.
\end{align}
Let 
\begin{align}\label{eq:LieAlgebraRep}
	\rho(\mathmakebox[8mm]{X_r X_s}) &= (\mathmakebox[8mm]{X_r X_s}) \otimes I + I \otimes (\mathmakebox[8mm]{X_r X_s}), \\
	\rho(\mathmakebox[8mm]{Z_r Z_s}) &= (\mathmakebox[8mm]{Z_r Z_s}) \otimes I  + I \otimes (\mathmakebox[8mm]{Z_r Z_s})
\end{align} 
denote the Lie algebra representation of $iX_r X_s, iZ_r Z_s$ from $\mathfrak{u}(\mathcal{H})$ on the $k$-fold tensor product. It is straightforward to check that
\begin{align}
	& \forall \beta_{rs}, \ \Big[M, \Big( \exp(-i \sum_{r \leq s} \beta_{rs} \mathmakebox[8mm]{X_r X_s}) \Big)^{\otimes 2} \Big] = 0 \quad \Leftrightarrow \quad \forall r,s, \ \Big[M, \rho(\mathmakebox[8mm]{X_r X_s}) \Big] = 0, \\
	& \forall \beta_{rs}, \ \Big[M, \Big( \exp(-i \sum_{r \leq s} \beta_{rs} \mathmakebox[8mm]{Z_r Z_s}) \Big)^{\otimes 2} \Big] = 0 \quad \Leftrightarrow \quad \forall r,s, \ \Big[M, \rho(\mathmakebox[8mm]{Z_r Z_s}) \Big] = 0.
\end{align}
We can also deal with the case of larger $k$ at the price of employing $k$-th order products of YJM elements. Then we employ the following lemma in our earlier work~\cite{SUd-k-Design2023}:

\begin{lemma}\label{lemma:CQAk}
	Consider the tensor product $\mathcal{H}^{\otimes k}$ of an $n$-qudit system with $\mathcal{H} = (\mathbb{C}^d)^{\otimes n}$. Taking representations $\rho(X_{i_1} \cdots X_{i_s})$, defined like Eq.~\eqref{eq:LieAlgebraRep} on $k$-fold tensor product, for all $1 \leq i_1 \leq \cdots \leq i_s \leq n$ and $s \leq k$ is sufficient to generate 
	\begin{align}
		\rho(D) = D \otimes I \otimes \cdots \otimes I + I \otimes D \otimes \cdots \otimes I + I \otimes I \otimes \cdots \otimes D
	\end{align}
	for an arbitrary diagonal matrix 
	\begin{align}\label{eq:SU(d)Diagonal}
		D = \bigoplus_\lambda D_{S^\lambda} \otimes I_{m_\lambda}  
	\end{align}
	under the Young basis respecting the $\SU(d)$-symmetric space decomposition of $\mathcal{H}$.
\end{lemma}

The underlying proof uses tensor analysis, commutativity of YJM elements and, more importantly, the fact proved by Okounkov and Vershik in Ref.~\cite{Okounkov1996}, that all ordinary YJM elements are able to generate any diagonal matrix under the Young basis, which can be treated as the case when $k = 1$ from above. Based on the same spirit, the above lemma also works for Pauli-$Z$ matrices under computational basis because they are also commutative and capable of generating arbitrary diagonal matrices under that basis:
\begin{align}\label{eq:U(1)Diagonal}
	D = \bigoplus_\mu D \vert_{S^\mu}.
\end{align}
As mentioned in Section \ref{sec:U(1)Univeristy}, $\U(1)$ charge sectors are multiplicity-free on the $n$-qubit system. The matrix $D$ is thus an arbitrary diagonal matrix acting on the entire Hilbert space.


For the present case of $k = 2$, these results indicate that the unit eigenspace of $T_2^{\YJM}$ is spanned by these who commute with $D \otimes I + I \otimes D$, with $D$ from Eq.~\eqref{eq:SU(d)Diagonal}, on the $2$-fold tensor product $\mathcal{H}^{\otimes 2}$. Like in Example \ref{example:InvariantSubK=2}, we label Young basis elements by $a,b,c,d,...$ with $m_\lambda$ emphasizing the type of the irrep and corresponding multiplicities. One can easily check that the following operators form an orthonormal basis for the concerned unit eigenspace:
\begin{align}
	& E\indices{^{a, m_\lambda^1}_{a, m_\lambda^2} } \otimes  E\indices{^{c, m_\mu^3}_{c, m_\mu^4} } = \Lket{ ( a, m_\lambda^1 ), ( c, m_\mu^3 )} \Lbra{( a, m_\lambda^2 ), ( c, m_\mu^4 ) }, \label{eq:YJMBases1} \\
	& E\indices{^{a, m_\lambda^1}_{c, m_\mu^4} } \otimes  E\indices{^{c, m_\mu^3}_{a, m_\lambda^2} } = \Lket{ ( a, m_\lambda^1 ), ( c, m_\mu^3 )} \Lbra{( c, m_\mu^4 ), ( a, m_\lambda^2 ) }, \label{eq:YJMBases2} \\
	& E\indices{^{a, m_\lambda^1}_{a, m_\lambda^2} } \otimes  E\indices{^{b, m_\lambda^3}_{b, m_\lambda^4} } = \Lket{ ( a, m_\lambda^1 ), ( b, m_\lambda^3 )} \Lbra{( a, m_\lambda^2 ), ( b, m_\lambda^4 ) }, \label{eq:YJMBases3-1} \\
	& E\indices{^{a, m_\lambda^1}_{b, m_\lambda^4} } \otimes  E\indices{^{b, m_\lambda^3}_{a, m_\lambda^2} } = \Lket{ ( a, m_\lambda^1 ), ( b, m_\lambda^3 )} \Lbra{( b, m_\lambda^4 ), ( a, m_\lambda^2 ) }. \label{eq:YJMBases3-2}
\end{align}
Covariant and contravariant basis indices $a,b,c$ need to appear in pairs and there are no restrictions when selecting multiplicities. It should be clear now that the first two kinds of basis elements are contained in subspaces from \eqref{eq:Case1} and \eqref{eq:Case2} respectively. The last two kinds collectively stay in subspaces from \eqref{eq:Case3}. By Eq.~\eqref{eq:TkCQASU(d)}, $T_{2,\SU(d)}^{\mathcal{E}_{\CQA}}$ is defined by the conjugate action of $T_{2}^{\YJM}$. It can \emph{only} have nontrivial action over these basis elements, which correspond subspaces from above cases. 

As a basic application of Schur--Weyl duality, when the moment operator is restricted to one subspace of either \eqref{eq:Case1} or \eqref{eq:Case2}, it admits a unique unit eigenvector up to scaling and multiplicity:
\begin{align}
	& \sum_{a,c} E\indices{^{a, m_\lambda^1}_{a, m_\lambda^2} } \otimes  E\indices{^{c, m_\mu^3}_{b, m_\mu^4} } =  \sum_{a,c} \Lket{ ( a, m_\lambda^1 ), ( c, m_\mu^3 )} \Lbra{( a, m_\lambda^2 ), ( c, m_\mu^4 ) }, \label{eq:Case1-Eigenvector} \\
	& \sum_{a,c} E\indices{^{a, m_\lambda^1}_{c, m_\mu^4} } \otimes  E\indices{^{c, m_\mu^3}_{a, m_\lambda^2} } =  \sum_{a,c} \Lket{ ( a, m_\lambda^1 ), ( c, m_\mu^3 )} \Lbra{( c, m_\mu^4 ), ( a, m_\lambda^2 ) }. \label{eq:Case2-Eigenvector}
\end{align}
However, when restricted to \eqref{eq:Case3}, it has two independent unit eigenvectors up to scaling and multiplicity:
\begin{align}\label{eq:Case3-Eigenvectors}
	&  \sum_{a,b} E\indices{^{a, m_\lambda^1}_{a, m_\lambda^2} } \otimes  E\indices{^{b, m_\lambda^3}_{b, m_\lambda^4} } =  \sum_{a,b} \Lket{ ( a, m_\lambda^1 ), ( b, m_\lambda^3 )} \Lbra{( a, m_\lambda^2 ), ( b, m_\lambda^4 ) }, \\
	&  \sum_{a,b} E\indices{^{a, m_\lambda^1}_{b, m_\lambda^4} } \otimes  E\indices{^{b, m_\lambda^3}_{a, m_\lambda^2} } =  \sum_{a,b} \Lket{ ( a, m_\lambda^1 ), ( b, m_\lambda^3 )} \Lbra{( b, m_\lambda^4 ), ( a, m_\lambda^2 ) }.
\end{align}


As mentioned in the beginning of the section, there is no need to consider equivalent irrep sectors when we evaluate the eigenvalues because the sub-matrices representing $T_{2,\SU(d)}^{\mathcal{E}_{\CQA}}$ on tensor products of different multiples of sectors are similar to each other. We bound the spectral gap of $T_{2,\SU(d)}^{\mathcal{E}_{\CQA}}$ restricted to \eqref{eq:Case1} or \eqref{eq:Case2} in the next subsection. For \eqref{eq:Case3}, we take an arbitrary $S_n$ irrep $S^\lambda$, neglects its multiplicities and study $T_{2,\SU(d)}^{\mathcal{E}_{\CQA}} \vert_{(S^{\lambda})^{\otimes 4}}$. Recall by Example \ref{example:InvariantSubK=2}, $(S^{\lambda})^{\otimes 4}$ can be decomposed into four invariant subspaces spanned by basis vectors like Eq.~\eqref{eq:ExampleBasis}. To preserve this decomposition, we adopt a change of basis from Eq.~\eqref{eq:YJMBases3-1} and Eq.~\eqref{eq:YJMBases3-2} (neglecting multiplicities) and define the following subspaces:
\begin{align}
	& S_{\sbA} = \operatorname{span} \left\{ \frac{1}{2} \Big( E\indices{^a_a} E\indices{^b_b} + E\indices{^b_b} E\indices{^a_a}) +  (E\indices{^a_b} E\indices{^b_a} + E\indices{^b_a} E\indices{^a_b} \Big), E\indices{^a_a} E\indices{^a_a}; a < b  \right\} \subset \Ima\mathcal{P}_{\sbA} \subset (S^{\lambda})^{\otimes 4} \tag{D5$^\ast$}  \\
	& S_{\sbD} = \operatorname{span} \left\{ \frac{1}{2} \Big( E\indices{^a_a} E\indices{^b_b} + E\indices{^b_b} E\indices{^a_a}) -  (E\indices{^a_b} E\indices{^b_a} + E\indices{^b_a} E\indices{^a_b} \Big); a < b \right\} \subset \Ima\mathcal{P}_{\sbD} \subset (S^{\lambda})^{\otimes 4} \tag{D6$^\ast$},
\end{align}
and
\begin{align}
	& S_{\sbB} = \operatorname{span} \left\{ \frac{1}{2} \Big( E\indices{^a_a} E\indices{^b_b} - E\indices{^b_b} E\indices{^a_a}) - (E\indices{^a_b} E\indices{^b_a} - E\indices{^b_a} E\indices{^a_b} \Big); a < b \right\} \subset \Ima\mathcal{P}_{\sbB} \subset (S^{\lambda})^{\otimes 4} \tag{D7$^\ast$}, \\
	& S_{\sbC} = \operatorname{span} \left\{ \frac{1}{2} \Big( E\indices{^a_a} E\indices{^b_b} - E\indices{^b_b} E\indices{^a_a}) + (E\indices{^a_b} E\indices{^b_a} - E\indices{^b_a} E\indices{^a_b} \Big); a < b \right\} \subset \Ima\mathcal{P}_{\sbC} \subset (S^{\lambda})^{\otimes 4} \tag{D8$^\ast$}.
\end{align}
In Section \ref{sec:detailsSU(d)}, we will apply Markov chain theory to study 
\begin{align}
	T_{2,\SU(d),\sbA}^{\mathcal{E}_{\CQA}} = T_{2,\SU(d)}^{\mathcal{E}_{\CQA}} \vert_{S_{\sbA}}, \quad
	T_{2,\SU(d),\sbB}^{\mathcal{E}_{\CQA}} = T_{2,\SU(d)}^{\mathcal{E}_{\CQA}} \vert_{S_{\sbB}}, \quad
	T_{2,\SU(d),\sbC}^{\mathcal{E}_{\CQA}} = T_{2,\SU(d)}^{\mathcal{E}_{\CQA}} \vert_{S_{\sbC}}, \quad
	T_{2,\SU(d),\sbD}^{\mathcal{E}_{\CQA}} = T_{2,\SU(d)}^{\mathcal{E}_{\CQA}} \vert_{S_{\sbD}} \tag{D9$^\ast$}
\end{align}
Without restricting to $S_{\sbA},S_{\sbB},S_{\sbC},S_{\sbD}$, one can check, by methods presented in Claim \ref{claim:U(1)-1} \& \ref{claim:SU(d)-1}, that the matrix representing $T_{2,\SU(d)}^{\mathcal{E}_{\CQA}}$ acting directly on Eq.~\eqref{eq:YJMBases3-1} and Eq.~\eqref{eq:YJMBases3-2} even admits negative entries, invalidating most techniques from Markov chain theory.

The situation for $T_{2,\U(1)}^{\mathcal{E}_{\CQA}}$ is almost the same, we just change to view $a,b,c,d,...$ from above as the indices of computational basis from the corresponding $\U(1)$ charge sectors, which do not even have equivalent copies.


\subsection{Proof of \eqref{eq:Case1} \& \eqref{eq:Case2} }\label{sec:detailsC1C2}

\begin{theorem}\label{thm:Case1&2}
	Let $\mathcal{T}$ be a generating set of transpositions from $S_n$ defined in Section \ref{sec:Aldous} and let $\mathcal{E}_{\CQA}$ be defined by transpositions from $\mathcal{T}$. The second largest eigenvalue of  $T_{2,\SU(d)}^{\mathcal{E}_{\CQA}}$ or $T_{2,\U(1)}^{\mathcal{E}_{\CQA}}$, when restricted to subspaces like
	\begin{align}
		& (S^{\lambda} \otimes \mathrm{1}_{m_{\lambda}}) \otimes S^{\mu} \otimes \mathrm{1}_{m_{\mu}} \otimes (S^{\lambda} \otimes \mathrm{1}_{m_{\lambda}}) \otimes (S^{\mu} \otimes \mathrm{1}_{m_{\mu}}), \\
		& (S^{\lambda} \otimes \mathrm{1}_{m_{\lambda}}) \otimes S^{\mu} \otimes \mathrm{1}_{m_{\mu}} \otimes (S^{\mu} \otimes \mathrm{1}_{m_{\mu}}) \otimes (S^{\lambda} \otimes \mathrm{1}_{m_{\lambda}}) 
	\end{align}
	given by \eqref{eq:Case1} and \eqref{eq:Case2}, can be uniformly bounded by the second largest eigenvalue of the $S_n$ Cayley graph determined by $\mathcal{T}$. For instance, taking nearest-neighbour transpositions,
	\begin{align}
		1 - \frac{1}{4(n-1)}  \Big(1- \cos\frac{\pi}{n}\Big) \geq \lambda_2
	\end{align}
	for both $\SU(d)$ and $\U(1)$ cases and arbitrary tensor product of subspaces $S^\lambda \ncong S^\mu$ from above.
\end{theorem}
\begin{proof}
	We only prove \eqref{eq:Case1}.  \eqref{eq:Case2} follows by the same argumentation. Recall that we can also ignore the multiplicities when analyzing eigenvalues, thus we consider
	\begin{align}
		T_{2,\SU(d)}^{\mathcal{E}_{\CQA}} \vert_{S^\lambda \otimes S^\mu \otimes S^\lambda \otimes S^\mu }, \quad T_{2,\U(1)}^{\mathcal{E}_{\CQA}} \vert_{S^\lambda \otimes S^\mu \otimes S^\lambda \otimes S^\mu }
	\end{align} 
 	for some $S^\lambda \ncong S^\mu$. We study the $\SU(d)$ case first. Note that the operators in question are nontrivially supported on the first kind of basis elements from \eqref{eq:YJMBases1}. For any Young basis element $\ket{a_0}$ of $S^\lambda$ and $\ket{c_0}$ of $S^\mu$, since the representation of $\tau$ is always a real matrix, we assume
 	\begin{align}\label{eq:GeneralTauExapnsion}
 		\tau \ket{a_0} = \sum_i C_{\lambda,i} \ket{a_i}, \quad \tau \bra{a_0} = \sum_i C_{\lambda,i} \bra{a_i}, \quad 
 		\tau \ket{b_0} = \sum_j C_{\mu,j} \ket{b_j}, \quad \tau \bra{b_0} = \sum_j C_{\mu,j} \bra{b_j}.
 	\end{align}
 	These real coefficients can be found analytically by the Young orthogonal form introduced in Section \ref{sec:SnTheory} when $\tau$ is a nearest-neighbour transposition. They can be left as general as possible for our purpose here. By Eq.~\eqref{eq: Ttau-k2-integral},
 	\begin{align}
 		T_2^\tau = \frac{1}{8} (3 IIII + 3\tau\tau\tau\tau + I\tau I\tau + I\tau\tau I + \tau I I\tau + \tau I\tau I - I I \tau\tau - \tau\tau I I ).
 	\end{align}
 	Neglecting multiplicities and by Eq.~\eqref{eq:GeneralTauExapnsion} we find that
 	\begin{align}\label{eq:Vanishing1}
 		\big( I\tau\tau I + \tau I\tau I - I I \tau\tau - \tau\tau I I \big) E\indices{^{a_0}_{a_0} } \otimes  E\indices{^{c_0}_{c_0} } 
 		= \big( I\tau\tau I + \tau I\tau I - I I \tau\tau - \tau\tau I I \big)  \ket{ a_0 , c_0 } \bra{a_0,c_0}
 	\end{align}
 	is projected to zero by $T_2^{\YJM}$ for an arbitrary choice of basis element. Therefore,
 	\begin{align}
 		T_{2,\SU(d)}^{\mathcal{E}_{\CQA}} \vert_{S^\lambda \otimes S^\mu \otimes S^\lambda \otimes S^\mu }
 		= \Big[ \frac{1}{\vert \mathcal{T} \vert} \sum_{\tau \in \mathcal{T}}T_2^{\YJM} \frac{1}{8} (3 IIII + 3\tau\tau\tau\tau + I\tau I\tau + \tau I\tau I ) T_2^{\YJM} \Big]  \vert_{S^\lambda \otimes S^\mu \otimes S^\lambda \otimes S^\mu }.
 	\end{align}
 	
 	Even without the conjugate action of $T_2^{\YJM}$, we now know by Theorem \ref{thm:CayleyGap} and discussion afterwards that the spectral gap of
 	\begin{align}
 		\frac{1}{\vert \mathcal{T} \vert}  \sum_{\tau \in \mathcal{T}}  \frac{1}{8} (6 IIII + I\tau I\tau + \tau I\tau I )\vert_{S^\lambda \otimes S^\mu \otimes S^\lambda \otimes S^\mu }
 		\geq \frac{1}{\vert \mathcal{T} \vert} \sum_{\tau \in \mathcal{T}} \frac{1}{8} (3 IIII + 3\tau\tau\tau\tau + I\tau I\tau + \tau I\tau I )\vert_{S^\lambda \otimes S^\mu \otimes S^\lambda \otimes S^\mu }
 	\end{align}
 	can be obtained by Cayley graphs of $S_n$. For instance, when
 	\begin{enumerate}
 		\item $\mathcal{T} = \{(1,2),(2,3),...,(n-1,n)\}$ corresponding to the 1D chain, 
 		\begin{align}
 		\begin{aligned} 
 			& 1 - \frac{1}{4(n-1)} \Big(1- \cos\frac{\pi}{n}\Big) \approx 1 - \frac{1}{4(n-1)} \frac{\pi^2}{n^2} 
 			\geq \lambda_2 \Big( \frac{1}{\vert \mathcal{T} \vert}  \sum_{\tau \in \mathcal{T}}  \frac{1}{8} (6 IIII + I\tau I\tau + \tau I\tau I )\vert_{S^\lambda \otimes S^\mu \otimes S^\lambda \otimes S^\mu } \Big) \\
 			\geq & \lambda_2\Big( \frac{1}{\vert \mathcal{T} \vert} \sum_{\tau \in \mathcal{T}} \frac{1}{8} (3 IIII + 3\tau\tau\tau\tau + I\tau I\tau + \tau I\tau I )\vert_{S^\lambda \otimes S^\mu \otimes S^\lambda \otimes S^\mu } \Big) \\
 			\geq & \lambda_2 \big( T_{2,\SU(d)}^{\mathcal{E}_{\CQA}} \vert_{S^\lambda \otimes S^\mu \otimes S^\lambda \otimes S^\mu } \big),
 		\end{aligned}
 		\end{align}
 		where we used the Cauchy interlacing theorem in last inequality by noting that the last two operators share the same unique unit eigenvector, in the subspace in question, up to scaling (see Eq.~\eqref{eq:Case1-Eigenvector}).  
 		
 		\item  $\mathcal{T} = \{(1,n),(2,n),...,(n-1,n)\}$ corresponding to the star, the upper bound on the second largest eigenvalue is 
 		\begin{align}
 			1 - \frac{1}{8(n-1)}.
 		\end{align}
 		
 		\item $\mathcal{T} = \{(i,j); i < j\}$ corresponding to the complete graph, the upper bound on the second largest eigenvalue is 
 		\begin{align}
 			1 - \frac{1}{4(n-1)}.
 		\end{align}
 	\end{enumerate}
 	The $\U(1)$ case can be proved immediately by the same method.
\end{proof}


\subsection{Proof of Lemma \ref{lemma:Modified}}

Before proving \eqref{eq:Case3}, we first prove basic properties of the modified moment operator in Definition \ref{def:Modified}. The proof works for both $\SU(d)$ and $\U(1)$  because we only use the property that $\tau$ is both a unitary and Hermitian matrix and $\tau^2 = I$. Accordingly, we omit the subscript indicating the symmetries here. By definition,
\begin{align}
	& M_2^\tau = \frac{1}{8}(6IIII + I \tau I \tau + \tau I \tau I + \tau I I \tau + I \tau \tau I - II \tau \tau - \tau \tau II), \tag{D12$^\ast$} \\
	& T_2^\tau = \frac{1}{8} (3 IIII + 3\tau\tau\tau\tau + I\tau I\tau + I\tau\tau I + \tau I I\tau + \tau I\tau I - I I \tau\tau - \tau\tau I I ). \tag{E4$^\ast$}
\end{align}
The definition of $T_2^\tau$ in Eq.~\eqref{eq: Ttau-k2-integral}) implies that it is positive semidefinite. Moreover, 
\begin{align}
	M_2^\tau \geq T_2^\tau \geq 0.
\end{align} 
Simple computation shows that $(M_2^\tau)^2 = T_2^\tau$, which implies that these positive semidefinite operators are simultaneously diagonalizable with
\begin{align}
	\lambda_i( M_2^\tau ) = \sqrt{ \lambda_i( T_2^\tau ) }.
\end{align}
Since all these eigenvalues are nonnegative and no larger than one,
\begin{align}
	2 \geq 1 + \sqrt{\lambda_i( T_2^\tau ) }  = \frac{1 - \lambda_i( T_2^\tau ) }{1 - \sqrt{\lambda_i( T_2^\tau ) }} = \frac{1 - \lambda_i( T_2^\tau ) }{1 - \lambda_i( M_2^\tau ) } 
	\implies 2( I - M_2^\tau ) \geq I - T_2^\tau. 
\end{align}
Summing over $\mathcal{T}$, we obtain the bound on second largest eigenvalues as well as spectral gaps of $M_2^{\mathcal{E}_{\CQA}}$ and $T_2^{\mathcal{E}_{\CQA}}$. The above inequality also indicates that they share the same unit eigenspace. 

We have already shown the invariance of $T_2^{\mathcal{E}_{\CQA}}$ in the most general form in Example \ref{example:InvariantSubK=2}.  Replacing the integral of conjugate actions in Eq.~\eqref{eq:Invariant} by that of $\tau^{\otimes 4} =  \tau^{\otimes 2} \otimes \tau^{\otimes 2}$, we see that $\tau^{\otimes 4}$ is also invariant under $\Ima\mathcal{P}_{\sbA},\Ima\mathcal{P}_{\sbB},\Ima\mathcal{P}_{\sbC}$ and $\Ima\mathcal{P}_{\sbD}$ defined in Eq.~\eqref{eq: projector-invariant-subspace}. Conjugating by the action of $T_2^{\YJM}$ (or $T_2^Z$), we see that $M_2^{\mathcal{E}_{\CQA}}$ is also well-defined within $S_{\sbA},S_{\sbB},S_{\sbC}$ and $S_{\sbD}$. 


\subsection{Proof of \eqref{eq:Case3} under $\U(1)$ symmetry}\label{sec:detailsU(1)}

\subsubsection{Proof of Lemma \ref{lemma:CayletBound1}}

With respect to a fixed charge sector $S^\mu$, we now explicitly write out the transition matrices of $M_{2, \U(1)}^{\mathcal{E}_{\CQA}}$ in the invariant subspaces 
\begin{align}
    S_{\sbA} = \text{span} \left\{ \frac{1}{2} \Big( (E\indices{^a_a} E\indices{^b_b} + E\indices{^b_b} E\indices{^a_a}) +  (E\indices{^a_b} E\indices{^b_a} + E\indices{^b_a} E\indices{^a_b}) \Big), E\indices{^a_a} E\indices{^a_a}; a < b \right\} \tag{D5$^\ast$}
\end{align}
As before, we save the notation by removing indices of sectors.

\begin{claim}\label{claim:U(1)-1}
	The following facts hold by the modified CQA second-moment operator $M_2^{\mathcal{E}_{\CQA}}$:
	\begin{enumerate}
		\item Under bases given in Eq.~\eqref{eq:InvariantSubapce-states1}, $M_{2,\U(1),\sbA}^{\mathcal{E}_{\CQA}}$ is irreducible, doubly stochastic and symmetric.
		
		\item Any nonzero off-diagonal entry of  $M_{2,\U(1)}^{\mathcal{E}_{\CQA}}$ comes from the action of a unique $\tau \in \mathcal{T}$ when defining $\mathcal{E}_{\CQA}$.
	\end{enumerate}
\end{claim}
\begin{proof}
	The action of $M_{2,\U(1)}^\tau$ on \eqref{eq:InvariantSubapce-states1} can be decomposed into even smaller invariant subspaces classified as the following four types:
	\begin{enumerate}
		\item Suppose $\tau(a) = b$ for some fixed $a \neq b$, $M_{2,\U(1)}^\tau$ is restricted to the following $3 \times 3$ matrix 
		\begin{align}\label{eq:U(1)++StochasticMatrix1}
			\frac{1}{8} \begin{pmatrix}
				4 & 2 & 2 \\
				2 & 6 & 0 \\
				2 & 0 & 6 \\
			\end{pmatrix} 
		\end{align}
	    within
		\begin{align}
			\text{span}\Big\{ \frac{1}{2} \Big( (E\indices{^a_a} E\indices{^b_b} + E\indices{^b_b} E\indices{^a_a}) +  (E\indices{^a_b} E\indices{^b_a} + E\indices{^b_a} E\indices{^a_b}) \Big), E\indices{^a_a} E\indices{^a_a}, E\indices{^b_b} E\indices{^b_b} \Big\}.
		\end{align}
		A trick to obtain these matrix entries easily is, for example, simply computing the expansion coefficient of 
		\begin{align}
			\frac{1}{2} M_{2,\U(1)}^\tau \Big( (E\indices{^a_a} E\indices{^b_b} + E\indices{^b_b} E\indices{^a_a}) +  (E\indices{^a_b} E\indices{^b_a} + E\indices{^b_a} E\indices{^a_b}) \Big) 
		\end{align}
		at $E\indices{^a_a} E\indices{^b_b}, E\indices{^a_a} E\indices{^a_a}$ and $E\indices{^b_b} E\indices{^b_b}$, because each simple tensor appears in one and only one basis element from Eq.~\eqref{eq:InvariantSubapce-states1}. 
		
		\item Suppose $\tau(a) = c, \tau(b) = d$ for some fixed $a \neq b \neq c \neq d$, $M_{2,\U(1)}^\tau$ is restricted to the following $4 \times 4$ matrix
		\begin{align}\label{eq:U(1)++StochasticMatrix2}
			\frac{1}{8}  \begin{pmatrix}
				6 & 0 & 1 & 1 \\
				0 & 6 & 1 & 1 \\
				1 & 1 & 6 & 0 \\
				1 & 1 & 0 & 6 \\
			\end{pmatrix}
		\end{align}
		within
		\begin{align}
			\begin{aligned}
				\text{span}\Big\{ & \frac{1}{2} \Big( (E\indices{^a_a} E\indices{^b_b} + E\indices{^b_b} E\indices{^a_a}) +  (E\indices{^a_b} E\indices{^b_a} + E\indices{^b_a} E\indices{^a_b}) \Big), 
				\frac{1}{2} \Big( (E\indices{^c_c} E\indices{^d_d} + E\indices{^d_d} E\indices{^c_c}) +  (E\indices{^c_d} E\indices{^d_c} + E\indices{^d_c} E\indices{^c_d}) \Big), \\
				& \frac{1}{2} \Big( (E\indices{^a_a} E\indices{^d_d} + E\indices{^d_d} E\indices{^a_a}) +  (E\indices{^a_d} E\indices{^d_a} + E\indices{^d_a} E\indices{^a_d}) \Big),
				\frac{1}{2} \Big( (E\indices{^c_c} E\indices{^b_b} + E\indices{^b_b} E\indices{^c_c}) +  (E\indices{^c_b} E\indices{^b_c} + E\indices{^b_c} E\indices{^c_b}) \Big) \Big\}.
			\end{aligned}
		\end{align}
		
		\item Suppose $\tau(a) = c, \tau(b) = b$ for some fixed $a \neq b \neq c$, $M_{2,\U(1)}^\tau$ is restricted to the following $2 \times 2$ matrix
		\begin{align}\label{eq:U(1)++StochasticMatrix3}
			\frac{1}{8} \begin{pmatrix}
				7 & 1 \\
				1 & 7 \\
			\end{pmatrix}
		\end{align}
		within
		\begin{align}
			\text{span}\Big\{ \frac{1}{2} \Big( (E\indices{^a_a} E\indices{^b_b} + E\indices{^b_b} E\indices{^a_a}) +  (E\indices{^a_b} E\indices{^b_a} + E\indices{^b_a} E\indices{^a_b}) \Big), 
			\frac{1}{2} \Big( (E\indices{^c_c} E\indices{^b_b} + E\indices{^b_b} E\indices{^c_c}) +  (E\indices{^c_b} E\indices{^b_c} + E\indices{^b_c} E\indices{^c_b}) \Big) \Big\}.
		\end{align}
		Similarly, we can study the case for $\tau(a) = a, \tau(b) = d$ and $a \neq b \neq d$.
		
		\item Suppose $\tau(a) = a, \tau(b) = b$ for some fixed $a \neq b$, $M_{2,\U(1)}^\tau$ is trivially restricted to the following $1 \times 1$ matrix
		\begin{align}\label{eq:U(1)++StochasticMatrix4}
			\frac{1}{8} \begin{pmatrix}	8 \end{pmatrix}
		\end{align}
		within
		\begin{align}
			\text{span}\Big\{ \frac{1}{2} \Big( (E\indices{^a_a} E\indices{^b_b} + E\indices{^b_b} E\indices{^a_a}) +  (E\indices{^a_b} E\indices{^b_a} + E\indices{^b_a} E\indices{^a_b}) \Big) \Big\}.
		\end{align}
	\end{enumerate}
	Notice that each submatrix is symmetric and doubly stochastic and each subblock cannot overlap. After summing over a generating set $\mathcal{T}$ and divided by the normalization constant, $M_{2,\U(1),\sbA}^{\mathcal{E}_{\CQA}}$ is a doubly stochastic symmetric matrix. 
	
	To prove that it is irreducible, we first abbreviate the basis elements from \eqref{eq:InvariantSubapce-states1} as $v_{ab}$ for simplicity and we note that $\mathcal{T}$ is a generating set of $S_n$. For any $v_{cd} \neq v_{ab}$, we can find transpositions $\tau_1,...,\tau_r,\sigma_1,...,\sigma_s$ such that
	\begin{align}
		& \tau_1(a) = a_1,..., \tau_i(a_{i-1}) = a_i,...,\tau_r(a_{r-1}) = c, \\
		& \sigma_1(b) = {b_1},..., \sigma_j(\mathmakebox[6.6mm]{b_{j-1}}) = b_j,...,\sigma_s(\mathmakebox[7.5mm]{b_{s-1}}) = d.
	\end{align}
	By the previous discussion,
	\begin{align}
		\big\langle v_{a_i b}, M_2^{\tau_i} v_{a_{i-1} b} \big\rangle,  \quad \big\langle v_{c b_j}, M_2^{\sigma_j} v_{c b_{j-1}}, \big\rangle \neq 0.
	\end{align}
	This indicates that $M_{2,\U(1),\sbA}^{\mathcal{E}_{\CQA}}$, as a doubly stochastic matrix, is also irreducible. The unique stationary distribution is uniform:
	\begin{align}
		\pi_{\sbA}(v_{ab}) \equiv \frac{2}{d_\mu (d_\mu + 1)}, 
	\end{align}
	where $d_\mu = \dim S^\mu$ and $\frac{d_\mu (d_\mu + 1)}{2}$ is the total number of states in $S_{\sbA}$. Written as the unit eigenvector, it is
	\begin{align}
		\begin{aligned}
			& \frac{2}{d_\mu (d_\mu + 1)} \Big[ \sum_{a < b} \frac{1}{2}\Big( (E\indices{^a_a} E\indices{^b_b} + E\indices{^b_b} E\indices{^a_a}) + (E\indices{^a_b} E\indices{^b_a} + E\indices{^b_a} E\indices{^a_b}) \Big)
			+ \sum_a E\indices{^a_a} E\indices{^a_a} \Big] \\
			= & \frac{1}{d_{\mu }(d_\mu + 1)}(\hat{e} + \hat{\sigma})    
			= \frac{2}{d_{\mu }(d_\mu + 1)} \Pi_{+}.
		\end{aligned}
	\end{align}
	
	Finally, assume $M_{2,\U(1),\sbA}^{\mathcal{E}_{\CQA}}$ can transition $v_{ab}$ to $v_{cd}$ with $v_{ab} \neq v_{cd}$. Suppose $a \neq c$, in order to have this nonzero transition probability, there must be some $\tau$ satisfying $\tau(a) = c$. Since $\tau$ is a transposition and $a,b,c,d$ stand for the computation basis elements of binary strings, there cannot be any other transposition having the same property, which completes the proof. 
\end{proof}


\begin{claim}\label{claim:Cayley-1}
	Analogously, for the Cayley moment operator
	\begin{align}
		\Cay_{2,\U(1)} = \frac{1}{\vert \mathcal{T} \vert} \sum_{\tau \in \mathcal{T}} T_2^Z \frac{1}{8}(6IIII + \tau I \tau I + I \tau I \tau) T_2^Z, \tag{D14$^\ast$}
	\end{align}
	let
	\begin{align}\label{sym-cay-state-space-1}
		\Sym \vcentcolon = \operatorname{span} \left\{ \frac{1}{2}\Big( E\indices{^a_a} E\indices{^b_b} + E\indices{^b_b} E\indices{^a_a} \Big), E\indices{^a_a} E\indices{^a_a}; a < b \right\} \subset (S^{\mu})^{\otimes 4}.
	\end{align}
	Then
	\begin{enumerate}		
		\item The Cayley operator is invariant under this subspace. Let $\Cay_{2,\U(1), \Sym}$ denotes its restriction. It is an irreducible reversible row-stochastic transition matrix under basis from Eq.~\eqref{sym-cay-state-space-1}.
		
		\item Any nonzero off-diagonal entry of $\Cay_{2,\U(1),\Sym}$ comes from the action of a unique $\tau \in \mathcal{T}$ when defining $\mathcal{E}_{\CQA}$.
	\end{enumerate}
\end{claim}
\begin{proof}
	As a reminder, the first kind of basis vectors from \eqref{sym-cay-state-space-1} are not normalized. This is made intentionally, which allows the matrix representation of $\Cay_{2,\U(1), \Sym}$ to be row stochastic as follows.
	
	For each $\tau$, the action of $\Cay_{2,\U(1)}^\tau$ on \eqref{sym-cay-state-space-1} can be decomposed into even smaller invariant subspaces classified as the following four types:
	\begin{enumerate}
		\item Suppose $\tau(a) = b$ for some fixed $a \neq b$, $\Cay_{2,\U(1)}^\tau$ is restricted to the following $3 \times 3$ matrix 
		\begin{align}
			\frac{1}{8} \begin{pmatrix}
				6 & 1 & 1 \\
				2 & 6 & 0 \\
				2 & 0 & 6 \\
			\end{pmatrix} 
		\end{align}
		within
		\begin{align}
			\text{span}\Big\{ \frac{1}{2}\Big( E\indices{^a_a} E\indices{^b_b} +  E\indices{^b_b} E\indices{^a_a} \Big), E\indices{^a_a} E\indices{^a_a}, E\indices{^b_b} E\indices{^b_b} \Big\}.
		\end{align} 

		\item Suppose $\tau(a) = c, \tau(b) = d$ for some fixed $a \neq b \neq c \neq d$, $\Cay_{2,\U(1)}^\tau$ is restricted to the following $4 \times 4$ matrix
		\begin{align}
			\frac{1}{8}  \begin{pmatrix}
				6 & 0 & 1 & 1 \\
				0 & 6 & 1 & 1 \\
				1 & 1 & 6 & 0 \\
				1 & 1 & 0 & 6 \\
			\end{pmatrix}
		\end{align}
		within
		\begin{align}
			\begin{aligned}
				\text{span}\Big\{ \frac{1}{2}\Big( E\indices{^a_a} E\indices{^b_b} + E\indices{^b_b} E\indices{^a_a} \Big), 
				\frac{1}{2}\Big( E\indices{^c_c} E\indices{^d_d} + E\indices{^d_d} E\indices{^c_c} \Big), 
				\frac{1}{2}\Big( E\indices{^a_a} E\indices{^d_d} + E\indices{^d_d} E\indices{^a_a} \Big),
				\frac{1}{2}\Big( E\indices{^c_c} E\indices{^b_b} + E\indices{^b_b} E\indices{^c_c} \Big) \Big\}.
			\end{aligned}
		\end{align}
		
		\item Suppose $\tau(a) = c, \tau(b) = b$ for some fixed $a \neq b \neq c$, $\Cay_{2,\U(1)}^\tau$ is restricted to the following $2 \times 2$ matrix
		\begin{align}
			\frac{1}{8} \begin{pmatrix}
				7 & 1 \\
				1 & 7 \\
			\end{pmatrix}
		\end{align}
		within
		\begin{align}
			\text{span}\Big\{ \frac{1}{2}\Big( E\indices{^a_a} E\indices{^b_b} + E\indices{^b_b} E\indices{^a_a} \Big), 
			\frac{1}{2}\Big( E\indices{^c_c} E\indices{^b_b} + E\indices{^b_b} E\indices{^c_c} \Big) \Big\}.
		\end{align}
		We can similarly consider the case for $\tau(a) = a, \tau(b) = d$ and $a \neq b \neq d$.
		
		\item Suppose $\tau(a) = a, \tau(b) = b$ for some fixed $a \neq b$, $\Cay_{2,\U(1)}^\tau$ is trivially restricted to the following $1 \times 1$ matrix
		\begin{align}
			\frac{1}{8} \begin{pmatrix}	8 \end{pmatrix}
		\end{align}
		within
		\begin{align}
			\text{span}\Big\{ \frac{1}{2}\Big( E\indices{^a_a} E\indices{^b_b} + E\indices{^b_b} E\indices{^a_a} \Big) \Big\}.
		\end{align}
	\end{enumerate} 
	By the same argumentation used in proving Claim \ref{claim:U(1)-1}, $\Cay_{2,\U(1), \Sym}$ is irreducible and row-stochastic. It should be noted that for the first case when $\tau(a) = b$, the matrix representation ceases to the doubly stochastic. Let $d_\mu = \dim S^\mu$. Since
	\begin{align}
		\frac{1}{d_\mu^2} \Big( 2 \cdot \frac{d_\mu( 1-d_\mu)}{2} + 1 \cdot d_\mu \Big) = 1,
	\end{align}
	one can check, by examining $\pi_{\Cay} \cdot \Cay_{2,\U(1)}^\tau$ for arbitrary $\tau$, that the current stationary distribution is
	\begin{align}
		\pi_{\Cay}(v_{ab}) = \begin{cases}
			\frac{2}{d_\mu^2}, & a < b \\
			\frac{1}{d_\mu^2}, & a = b
		\end{cases}, 
	\end{align}
	where we use abbreviation $v_{ab}$ here for basis vectors from \eqref{sym-cay-state-space-1}. Written as a vector, it is
	\begin{align} 
		\begin{aligned}
			& \frac{2}{d_\mu^2} \sum_{a < b}  (E\indices{^a_a} E\indices{^b_b} + E\indices{^b_b} E\indices{^a_a})
			+ \frac{1}{d_\mu^2} \sum_a E\indices{^a_a} E\indices{^a_a}.
		\end{aligned}
	\end{align}
	The transition matrix is reversible because for any $\tau$,
	\begin{align}
		\pi_{\Cay}(v_{ab}) \Cay_{2,\U(1)}^\tau(v_{ab}, v_{cd}) = \pi_{\Cay}(v_{cd}) \Cay_{2,\U(1)}^\tau(v_{cd}, v_{ab}),
	\end{align}
	which concludes the proof.
\end{proof}


We are now going to bound the spectral gap of $\Cay_{2, \U(1),\Sym}$. As a caveat, the theory of $S_n$ Cayley graphs introduced in Section \ref{sec:Aldous} can be applied to evaluate the spectral gap of 
\begin{align}\label{eq:CayleySum}
	\Cay = \frac{1}{\vert \mathcal{T} \vert} \sum_{\tau \in \mathcal{T}} \frac{1}{8}(6IIII + \tau I \tau I + I \tau I \tau).
\end{align}
However, with additional projections given by second-order Pauli-$Z$ matrices on
\begin{align}
	\Cay_{2,\U(1)} = \frac{1}{\vert \mathcal{T} \vert} \sum_{\tau \in \mathcal{T}} T_2^Z \frac{1}{8}(6IIII + \tau I \tau I + I \tau I \tau) T_2^Z, \tag{D14$^\ast$}
\end{align}
or YJM elements, we have to study any potential influence on the unit eigenvalues and spectral gaps. 

To be specific, under $\U(1)$ symmetry, we explain at the end of in Section \ref{sec:Tabloid} that the decomposition of 2-fold tensor product of any $\U(1)$ charge sector $S^\mu = S^{(n-r, r)}$ admits $r+1$ many trivial irreps. As a result, the operator $\Cay$ defined in Eq.~\eqref{eq:CayleySum} has $(r+1)^2$ multiples of unit eigenvalues and the spectral gap is the difference between these unit eigenvalues and the $( (r+1)^2 + 1)$-th eigenvalue. On the contrary, after being projected through $T_2^Z$, $\Cay_{2,\U(1),\Sym}$ admits only one unit eigenvalue by Claim \ref{claim:U(1)-1}. We are going to show in the following that these differences will not affect our evaluation. Actually, the spectral gap of $\Cay_{2,\U(1)}$, as well as its restriction $\Cay_{2,\U(1),\Sym}$, is no less than that of $\Cay$. 

\begin{claim}\label{claim:Cayley-2}
	Let $S^\mu$ denote any $\U(1)$ charge sector with $\mu = (n-r, r)$ and $r = 0,...,n$. The second-largest eigenvalue of $\Cay_{2, \U(1)}$ defined under $(S^{\lambda})^{\otimes 4}$ is upper bounded by $\frac{1}{2} \Big( 1 + \frac{1}{\vert \mathcal{T} \vert} \lambda_2( \mathcal{G}(\mathcal{T}) ) \Big)$ with $\mathcal{G}(\mathcal{T})$ the Cayley graph generated by $\mathcal{T}$. 
\end{claim}
\begin{proof}
	We prove this lemma by assembling the following simple facts:
	\begin{enumerate}
		\item Under computational basis, the unit eigenspace of $\sum_{\tau \in \mathcal{T}} \tau\tau$ can be spanned by $r+1$ orthogonal eigenvectors of the following form
		\begin{align}
			u_{ab} = \sum_{\sigma \in S_n} E\indices{^{\sigma(a)}_{\sigma(b)}}.
		\end{align} 
		In particular, distinct orthogonal basis vectors can be differentiated by examining the Hamming weights $a + b = 0, 2, \cdots, 2r; \mod 2$ in $\mathbb{F}_2$ field.
		
		\item Then the $(r+1)^2$ unit eigenvectors of $\Cay$ are given as
		\begin{align}\label{eq:CayEigenvector}
			\Psi\indices{^a_b}\Psi\indices{^c_d} = \sum_{\sigma, \pi \in S_n} E\indices{^{\sigma(a)}_{\sigma(b)}} E\indices{^{\pi(c)}_{\pi(d)}} = u_{ab} \otimes u_{cd},
		\end{align}
		where $(a,b)$ and $(c,d)$ are taken from the $r+1$ equivalence classes defined in the above congruence equation. It is easy to see that
		\begin{align}
			\langle E\indices{^a_a} E\indices{^c_c}, \Psi\indices{^a_b}\Psi\indices{^c_d} \rangle = 0
		\end{align}
		unless $a = b$ and $c = d$. By Lemma \ref{lemma:CQAk}, $T_2^Z (\Psi\indices{^a_a}\Psi\indices{^c_c} ) = \Psi\indices{^a_a}\Psi\indices{^c_c}$. It is thus the unit eigenvector, unique up to scaling, of $\Cay_{2,\U(1),\Sym}$.
		
		\item There are still vectors $u$ orthogonal to the unit eigenspace of $\Cay$, yet capable of surviving under the projection of $T_2^Z$. Evidently, by definition,
		\begin{align}
			\Vert \Cay_{2,\U(1)} (u) \Vert = \Vert T_2^Z \Cay T_2^Z (u) \Vert \leq \lambda_2( \Cay ) \Vert u \Vert.
		\end{align}
	\end{enumerate}
	Since we have eliminated the concern of multiple unit eigenvectors, 
	\begin{align}
		\lambda_2 (\Cay_{2, \U(1)}) \leq \lambda_2( \Cay ) = \frac{1}{2} \Big( 1 + \frac{1}{\mathcal{T}} \sum_{\tau \in \mathcal{T}} \tau\tau \Big) 
		\leq \frac{1}{2} \Big( 1 + \frac{1}{\vert \mathcal{T} \vert} \lambda_2( \mathcal{G}(\mathcal{T}) ) \Big)
		= \frac{1}{\vert \mathcal{T} \vert} O\Big( \lambda_2( \mathcal{G}(\mathcal{T}) ) \Big),
	\end{align} 
	which completes the proof.
\end{proof}


We are now at the position to verify Lemma \ref{lemma:CayletBound1}. Let us recall the notion of Dirichlet form in Definition \ref{def:dirichlet-form}
\begin{align}
    \mathcal{E}(f) = \frac{1}{2} \sum_{x, y \in \mathcal{X}} (f(x) - f(y))^2 \pi_P(x) P(x, y)
\end{align}
on a reversible chain with a stationary distribution $(\mathcal{X}, P, \pi)$.

By Claim \ref{claim:U(1)-1} \& \ref{claim:Cayley-1}, let 
\begin{align}
	P_{\sbA} \vcentcolon = M_{2,\U(1),\sbA}^{\mathcal{E}_{\CQA}}, \quad P_{\Cay} \vcentcolon = \Cay_{2,\U(1),\Sym}.
\end{align}
We consider Markov processes $(S_{\sbA}, P_{\sbA}, \pi_{\sbA})$ and $(\Sym, P_{\Cay}, \pi_{\Cay})$. Taking a close look at Eq.~\eqref{eq:InvariantSubapce-states1} \& \eqref{sym-cay-state-space-1}, we notice that there is a natural one-to-one correspondence between basis states from $S_{\sbA}$ and $\Sym$. The bijection is not an isometry as it does not preserve the inner products. For our purpose it is sufficient to regard this mapping as a set level bijection because we only discuss probability theory instead of geometry within these state spaces. For what follows, we will simply identify them as the same state space, all denoted by $S_{\sbA}$.

Then for any function $f$ on $S_{\sbA}$, 
\begin{align}
    \begin{aligned}
    	\mathcal{E}_{\Cay}(f) = & \frac{1}{2} \sum_{x, y \in S_{\sbA}} (f(x) - f(y))^2 \pi_{\Cay}(x) P_{\Cay}(x, y)  \\
    	& \leq \frac{d_\mu + 1}{d_\mu} \mathcal{E}_{\sbA}(f) =  \frac{d_\mu + 1}{2 d_\mu} \sum_{x, y \in S_{\sbA}} (f(x) - f(y))^2 \pi_{\sbA}(x) P_{\sbA} (x, y).
    \end{aligned}
\end{align}
The reason can be see from Claim \ref{claim:U(1)-1} \& \ref{claim:Cayley-1}. Especially the fact that any nonzero off-diagonal entry of these transition matrices come from the action of a unique $\tau \in \mathcal{T}$. Then we can compare $\pi_{\Cay}(x) P_{\Cay}(x, y), \pi_{\sbA}(x) P_{\sbA} (x, y)$ simply by checking the matrix representations induced by each $\tau$ separately. For any $x = v_{ab} \neq v_{cd} = y$ with $v_{ab}, v_{cd}$ being abbreviations of state vectors in Eq.~\eqref{eq:InvariantSubapce-states1}, there are simply three cases when we have nonzero transition probabilities: 
\begin{align}
	\begin{cases}
		\pi_{\Cay}(x) P_{\Cay}(x, y) = \frac{2}{d_\mu^2} \frac{1}{8\vert \mathcal{T} \vert}
		> \pi_{\sbA}(x) P_{\sbA} (x, y) = \frac{2}{d_\mu (d_\mu + 1)} \frac{1}{8\vert \mathcal{T} \vert}, & a \neq b, c \neq d;  \\
		\pi_{\Cay}(x) P_{\Cay}(x, y) = \frac{2}{d_\mu^2} \frac{1}{8\vert \mathcal{T} \vert}
		< \pi_{\sbA}(x) P_{\sbA} (x, y) = \frac{2}{d_\mu (d_\mu + 1)} \frac{2}{\vert 8\mathcal{T} \vert}, & a \neq b, c = d; \\
		\pi_{\Cay}(x) P_{\Cay}(x, y) = \frac{1}{d_\mu^2} \frac{2}{8\vert \mathcal{T} \vert}
		< \pi_{\sbA}(x) P_{\sbA} (x, y) = \frac{2}{d_\mu (d_\mu + 1)} \frac{2}{8\vert \mathcal{T} \vert}, & a = b, c \neq d.
	\end{cases}
\end{align}
As a reminder, there is no transitions between states with $a = b$ and $c = d$. As a result, by Theorem \ref{thm:comparison-dirichlet}, we obtain that 
\begin{align}
     \Delta(\Cay_{2,\U(1),\Sym})
     \leq \max_{x \in S_{\sbA}} \frac{  \pi_{\Cay}(x)}{\pi_{\sbA}(x)} \frac{d_\mu + 1}{d_\mu}  \Delta(M_{2,\U(1),\sbA}) 
     = \Big(\frac{d_\mu + 1}{d_\mu} \Big)^2 \Delta(M_{2,\U(1),\sbA}) < 2\Delta(M_{2,\U(1),\sbA}), 
\end{align}
where we discard the trivial case when $d_\mu = 1$. The next smallest dimension of $\U(1)$ charge sector is $d_\mu = n$. The inequality holds when $n \geq 3$. We can also argue in the converse direction so that we obtain 
\begin{align}\label{eq:Inequality-1}
    & \Delta(M_{2,\U(1),\sbA})  
    \leq \Big(\frac{2 d_\mu}{d_\mu + 1} \Big)^2 \Delta(\Cay_{2,\U(1), \Sym}) \leq 4\Delta(\Cay_{2,\U(1),\Sym}) \\
    \implies & \frac{1}{4} \Delta(M_{2,\U(1),\sbA})   
    \leq \Delta(\Cay_{2,\U(1),\Sym}) 
    \leq 2 \Delta(M_{2,\U(1),\sbA}).  
\end{align}
Obviously, since $\Sym \subset (S^{\lambda})^{\otimes 4}$ is a subspace, $\lambda_2 (\Cay_{2, \U(1),\Sym}) \leq \lambda_2 (\Cay_{2, \U(1)})$ and we complete the proof of Lemma \ref{lemma:CayletBound1}. 


\subsubsection{Proof of Lemma \ref{lemma:CayletBound2}}

We now study $M_{2,\U(1), \ytableaushort{ {} {*(black)} , {*(black)} {} }}^{\mathcal{E}_{\CQA}})$ which is the modified moment operator restricted to the invariant subspace $S_{\sbB} \oplus S_{\sbC}$ with
\begin{align}
	& S_{\sbB} = \operatorname{span} \left\{ \frac{1}{2} \Big( E\indices{^a_a} E\indices{^b_b} - E\indices{^b_b} E\indices{^a_a}) - (E\indices{^a_b} E\indices{^b_a} - E\indices{^b_a} E\indices{^a_b} \Big); a < b \right\}, \tag{D7$^\ast$} \\
	& S_{\sbC} = \operatorname{span} \left\{ \frac{1}{2} \Big( E\indices{^a_a} E\indices{^b_b} - E\indices{^b_b} E\indices{^a_a}) + (E\indices{^a_b} E\indices{^b_a} - E\indices{^b_a} E\indices{^a_b} \Big); a < b \right\}. \tag{D8$^\ast$}
\end{align}
Checking its matrix representations by methods used in Claim \ref{claim:U(1)-1}, we find that, when acting on these basis elements,
\begin{align}\label{eq:Vanishing2}
	\Big[ T_2^{Z} \Big( \frac{1}{8\vert \mathcal{T} \vert} \sum_\tau I\tau\tau I + \tau I \tau I - \tau\tau II - II \tau\tau \Big) T_2^{Z} \Big] \vert_{S_{\sbB} \oplus S_{\sbC}} = 0.
\end{align}		
It is similar to the situation happens to Eq.~\eqref{eq:Vanishing1} in proving Theorem \ref{thm:Case1&2}. Accordingly, the operator $ M_{2,\U(1),\ytableaushort{ {} {*(black)} , {*(black)} {} }}^{\mathcal{E}_{\CQA}}$ degenerates to $\Cay_{2,\U(1)}$ within this subspace. By Eq.~\eqref{eq:CayEigenvector}, the unit eigenvector of $\Cay_{2,\U(1)}$ is obviously orthogonal to $S_{\sbB} \oplus S_{\sbC}$ and therefore,
\begin{align}\label{eq:Inequality-2}
	\lambda_1( M_{2,\U(1),\ytableaushort{ {} {*(black)} , {*(black)} {} }}^{\mathcal{E}_{\CQA}}) ) \leq \lambda_2( \Cay_{2,\U(1)}).
\end{align}
As mentioned at the end of Section \ref{sec:detailsT2CQA}, unit eigenvectors of $M_{2,\U(1)}$ are located in $S_{\sbA}$ and $S_{\sbD}$. The $\SU(d)$ case can be solved in the same way. To testify Eq.~\eqref{eq:Vanishing2} with the projection replaced by $T_2^{\YJM}$:
\begin{align}
	T_2^{\YJM} \Big( \frac{1}{8\vert \mathcal{T} \vert} \sum_\tau I\tau\tau I + \tau I \tau I - \tau\tau II - II \tau\tau \Big) T_2^{\YJM} = 0,
\end{align}
we make use of Young orthogonal form presented in Section \ref{sec:detailsSU(d)}.


\subsubsection{Proof of Lemma \ref{lemma:induced-markov-process} }

We now consider the case when $M_{2,\U(1)}^{\mathcal{E}_{\CQA}}$ is restricted to
\begin{align}
	S_{\sbD} = \operatorname{span} \left\{ \frac{1}{2} \Big( E\indices{^a_a} E\indices{^b_b} + E\indices{^b_b} E\indices{^a_a}) -  (E\indices{^a_b} E\indices{^b_a} + E\indices{^b_a} E\indices{^a_b} \Big); a < b \right\}. \tag{D6$^\ast$}
\end{align}

\begin{claim}\label{claim:U(1)-2}
	Under bases given in Eq.~\eqref{eq:InvariantSubapce-states2}, $M_{2,\U(1),\sbD}^{\mathcal{E}_{\CQA}}$ is an irreducible doubly stochastic symmetric matrix. Any nonzero off-diagonal entry of  $M_{2,\U(1),\sbD}^{\mathcal{E}_{\CQA}}$ comes from the action of a unique $\tau \in \mathcal{T}$ when defining $\mathcal{E}_{\CQA}$.
\end{claim}
\begin{proof}
	The proof proceeds nearly the same as Claim \ref{claim:U(1)-1}. The action of $M_{2,\U(1)}^\tau$ on Eq.~\eqref{eq:InvariantSubapce-states2} can be described using even smaller invariant subspaces classified as the following four types:
	\begin{enumerate}
		\item Suppose $\tau(a) = b$ for some fixed $a \neq b$, $M_{2,\U(1)}^\tau$ is restricted to the following $1 \times 1$ matrix 
		\begin{align}
			\frac{1}{8} \begin{pmatrix}
				8 \\
			\end{pmatrix} 
		\end{align}
		within
		\begin{align}
			\text{span}\Big\{ \frac{1}{2} \Big( (E\indices{^a_a} E\indices{^b_b} + E\indices{^b_b} E\indices{^a_a}) -  (E\indices{^a_b} E\indices{^b_a} + E\indices{^b_a} E\indices{^a_b}) \Big) \Big\}.
		\end{align}
		
		\item Suppose $\tau(a) = c, \tau(b) = d$ for some fixed $a \neq b \neq c \neq d$, $M_{2,\U(1)}^\tau$ is restricted to the following $4 \times 4$ matrix
		\begin{align}
			\frac{1}{8}  \begin{pmatrix}
				6 & 0 & 1 & 1 \\
				0 & 6 & 1 & 1 \\
				1 & 1 & 6 & 0 \\
				1 & 1 & 0 & 6 \\
			\end{pmatrix}
		\end{align}
		within
		\begin{align}
			\begin{aligned}
				\text{span}\Big\{ & \frac{1}{2} \Big( (E\indices{^a_a} E\indices{^b_b} + E\indices{^b_b} E\indices{^a_a}) -  (E\indices{^a_b} E\indices{^b_a} + E\indices{^b_a} E\indices{^a_b}) \Big), 
				\frac{1}{2} \Big( (E\indices{^c_c} E\indices{^d_d} + E\indices{^d_d} E\indices{^c_c}) - (E\indices{^c_d} E\indices{^d_c} + E\indices{^d_c} E\indices{^c_d}) \Big), \\
				& \frac{1}{2} \Big( (E\indices{^a_a} E\indices{^d_d} + E\indices{^d_d} E\indices{^a_a}) - (E\indices{^a_d} E\indices{^d_a} + E\indices{^d_a} E\indices{^a_d}) \Big),
				\frac{1}{2} \Big( (E\indices{^c_c} E\indices{^b_b} + E\indices{^b_b} E\indices{^c_c}) - (E\indices{^c_b} E\indices{^b_c} + E\indices{^b_c} E\indices{^c_b}) \Big) \Big\}.
			\end{aligned}
		\end{align}
		
		\item Suppose $\tau(a) = c, \tau(b) = b$ for some fixed $a \neq b \neq c$, $M_{2,\U(1)}^\tau$ is restricted to the following $2 \times 2$ matrix
		\begin{align}
			\frac{1}{8} \begin{pmatrix}
				7 & 1 \\
				1 & 7 \\
			\end{pmatrix}
		\end{align}
		within
		\begin{align}
			\text{span}\Big\{ \frac{1}{2} \Big( (E\indices{^a_a} E\indices{^b_b} + E\indices{^b_b} E\indices{^a_a}) -  (E\indices{^a_b} E\indices{^b_a} + E\indices{^b_a} E\indices{^a_b}) \Big), 
			\frac{1}{2} \Big( (E\indices{^c_c} E\indices{^b_b} + E\indices{^b_b} E\indices{^c_c}) - (E\indices{^c_b} E\indices{^b_c} + E\indices{^b_c} E\indices{^c_b}) \Big) \Big\}.
		\end{align}
		We can similarly consider the case for $\tau(a) = a, \tau(b) = d$ and $a \neq b \neq d$.
		
		\item Suppose $\tau(a) = a, \tau(b) = b$ for some fixed $a \neq b$, $M_{2,\U(1)}^\tau$ is trivially restricted to the following $1 \times 1$ matrix
		\begin{align}
			\frac{1}{8} \begin{pmatrix}	8 \end{pmatrix}
		\end{align}
		within
		\begin{align}
			\text{span}\Big\{ \frac{1}{2} \Big( (E\indices{^a_a} E\indices{^b_b} + E\indices{^b_b} E\indices{^a_a}) -  (E\indices{^a_b} E\indices{^b_a} + E\indices{^b_a} E\indices{^a_b}) \Big) \Big\}.
		\end{align}
	\end{enumerate}
	Notice that each submatrix is symmetric and doubly stochastic and each subblock cannot overlap. After summing over $\mathcal{T}$ and divided by the normalization constant, $M_{2,\U(1),\sbD}^{\mathcal{E}_{\CQA}}$ is a doubly stochastic symmetric matrix. Any its nonzero off-diagonal entry comes from the action of a unique $\tau \mathcal{T}$. The irreducibility can be verified similarly as in Claim \ref{claim:U(1)-1}. The stationary distribution is uniform:
	\begin{align}\label{eq:--Stationary}
		\pi_{\sbD}(v_{ab}) \equiv \frac{2}{d_\mu (d_\mu - 1)}, 
	\end{align}
	where $\frac{d_{\mu }(d_\mu - 1)}{2}$ is the total number of states in $S_{\sbD}$. As a vector, it can be written as 
	\begin{align}
		\begin{aligned}
			& \frac{2}{d_\mu (d_\mu - 1)} \Big[ \sum_{a < b} \frac{1}{2}\Big( (E\indices{^a_a} E\indices{^b_b} + E\indices{^b_b} E\indices{^a_a}) - (E\indices{^a_b} E\indices{^b_a} + E\indices{^b_a} E\indices{^a_b}) \Big) \Big] \\
			= & \frac{1}{d_\mu (d_\mu - 1)}(\hat{e} - \hat{\sigma})    
			= \frac{2}{d_\mu (d_\mu - 1)} \Pi_-,
		\end{aligned}
	\end{align}
	which completes the proof.
\end{proof}

Our next strategy is to bound the spectral gap of $M_{2,\U(1),\sbD}^{\mathcal{E}_{\CQA}}$ via that of $M_{2,\U(1),\sbA}^{\mathcal{E}_{\CQA}}$. However, not like the case when we work with Cayley moment operators having $S_{\sbA} \cong \Sym$, $S_{\sbA} \ncong S_{\sbD}$ by definition. This hinders the use of comparing Dirichlet forms directly. Fortunately, we can solve this problem by the notion of induced chain introduced in Definition~\ref{def:induced-markov-chain}. Let 
\begin{align}
	\tilde{S}_{\sbA} \vcentcolon = \operatorname{span} \left\{ \frac{1}{2} \Big( (E\indices{^a_a} E\indices{^b_b} + E\indices{^b_b} E\indices{^a_a}) + (E\indices{^a_b} E\indices{^b_a} + E\indices{^b_a} E\indices{^a_b}) \Big); a < b \right\}.
\end{align}
be the subspace of $S_{\sbA}$ by removing basis elements $E\indices{^a_a} E\indices{^a_a}$ for all $a$. Then there is a canonical bijection between $\tilde{S}_{\sbA}$ and $S_{\sbD}$. For brevity, we rewrite the these basis states by $s_{ab}$ with $a \leq b$.

Recall we set $P_{\sbA} = M_{2,\U(1),\sbA}^{\mathcal{E}_{\CQA}}$ and now we consider the induced Markov chain $(\widetilde{S}_{\sbA}, \widetilde{P}_{\sbA}, \widetilde{\pi})$. By Definition \ref{def:induced-markov-chain} and Example \ref{example:Markov}, transitions of a basis state $s_{ab}$ under $\widetilde{P}_{\sbA}$ can be summarized into the following cases:
\begin{enumerate}
	\item Transition to itself with
	\begin{align}\label{eq:tildeP-1}
		\widetilde{P}_{\sbA}(s_{ab}, s_{ab}) & = P_{\sbA}(s_{ab}, s_{ab})
		+ P_{\sbA}(s_{ab}, s_{aa}) \sum^{\infty}_{p=0} P_{\sbA}(s_{aa}, s_{aa})^p P_{\sbA}(s_{aa}, s_{ab})
		+ P_{\sbA}(s_{ab}, s_{bb}) \sum^{\infty}_{p=0} P_{\sbA}(s_{bb}, s_{bb})^p P_{\sbA}(s_{bb}, s_{ab}) \notag \\
		& = P_{\sbA}(s_{ab}, s_{ab})
		+ P_{\sbA}(s_{ab}, s_{aa}) \frac{P_{\sbA}(s_{bb}, s_{ab}) }{1 - P_{\sbA}(s_{aa}, s_{aa}) }
		+ P_{\sbA}(s_{ab}, s_{bb}) \frac{P_{\sbA}(s_{bb}, s_{ab}) }{1 - P_{\sbA}(s_{bb}, s_{bb}) },
	\end{align}
	where the last two terms is nonzero if and only if there is some SWAP $\tau$ such that $\tau(a) = b$ (and hence $\tau(b) = a$).
	
	\item Transition to $s_{cd}$ for which $c,d \neq a,b$. Then
	\begin{align}\label{eq:tildeP-2}
		\widetilde{P}_{\sbA}(s_{ab}, s_{cd}) = P_{\sbA}(s_{ab}, s_{cd}).
	\end{align}
	
	\item Transition to $s_{ac}$ with
	\begin{align}\label{eq:tildeP-3}
		\widetilde{P}_{\sbA}(s_{ab}, s_{ac}) & = P_{\sbA}(s_{ab}, s_{ac})
		+ P_{\sbA}(s_{ab}, s_{aa}) \sum^{\infty}_{p=0} P_{\sbA}(s_{aa}, s_{aa})^p P_{\sbA}(s_{aa}, s_{ac}) \notag \\
		& = P_{\sbA}(s_{ab}, s_{ac})
		+ P_{\sbA}(s_{ab}, s_{aa}) \frac{1}{1 - P_{\sbA}(s_{aa}, s_{aa}) } P_{\sbA}(s_{aa}, s_{ac}).
	\end{align}
	where the last term is nonzero if and only if there is some SWAPs $\tau_1,\tau_2$ such that $\tau_1(a) = b$ and $\tau_2(a) = c$. Diagram \eqref{diagram:paths1} illustrates this case with break lines standing for transitions outside $\widetilde{S}_{\sbA}$.
	
	\item Transition to $s_{bd}$ with
	\begin{align}\label{eq:tildeP-4}	
		\widetilde{P}_{\sbA}(s_{ab}, s_{bd}) & = P_{\sbA}(s_{ab}, s_{bd})
		+ P_{\sbA}(s_{ab}, s_{aa}) \sum^{\infty}_{p=0} P_{\sbA}(s_{bb}, s_{bb})^p P_{\sbA}(s_{bb}, s_{bd}) \notag \\
		& = P_{\sbA}(s_{ab}, s_{bd})
		+ P_{\sbA}(s_{ab}, s_{bb}) \frac{1}{1 - P_{\sbA}(s_{bb}, s_{bb}) } P_{\sbA}(s_{bb}, s_{bd}).
	\end{align}
	where the last term is nonzero if and only if there is some SWAPs $\tau_3,\tau_4$ such that $\tau_3(b) = a$ and $\tau_4(b) = d$.
\end{enumerate} 
These formulas encompass all possible transition paths starting from $s_{ab}$, possibly passing through states that outside $\widetilde{S}_{\sbA}$ and returning back. A crucial observation is that there can never be any transitions between states $s_{aa}$ and $s_{bb}$ with $a \neq b$ for $P_{\sbA}$, which largely simplifies the computation. 

It is also straightforward to see that $\widetilde{P}_{\sbA}$ is both irreducible and symmetric with stationary distribution $\widetilde{\pi}$ being uniform. Identifying $\widetilde{S}_{\sbA}$ and $S_{\sbD}$, $\widetilde{\pi}$ equals Eq.~\eqref{eq:--Stationary}. We are going to apply Theorem \ref{thm:induced-chain-spectral-gap} \& \ref{thm:path-comparison-theorem} to conclude that: 

\begin{claim}\label{claim:U(1)-3}
    The following inequalities about spectral gaps hold:
    \begin{align}\label{eq:Inequality-3}
    	\Delta(M_{2,\U(1),\sbA}^{\mathcal{E}_{\CQA}}) 
    	= \Delta (P_{\sbA})  
    	\leq \Delta(\widetilde{P}_{\sbA}) 
    	\leq 5 \Delta(P_{\sbD}) 
    	= 5 \Delta(M_{2,\U(1),\sbD}^{\mathcal{E}_{\CQA}})
    	\leq 5 \Delta(\widetilde{P}_{\sbA}). 
    \end{align}
\end{claim}
\begin{proof}
	By the Theorem \ref{thm:induced-chain-spectral-gap}, $\Delta(P_{\sbA}) \leq \Delta (\widetilde{P}_{\sbA})$. Since we have bounded $\Delta(P_{\sbA})$ in Lemma \ref{lemma:CayletBound1}, it suffices to show $\Delta(\widetilde{P}_{\sbA}) = O(\Delta(P_{\sbD}))$. We actually reach a stronger relation: 
	\begin{align}
		\Delta(\widetilde{P}_{\sbA}) = \Theta(\Delta(P_{\sbD}))
	\end{align}
	later. It can be easily seen from the above definition of $\widetilde{P}_{\sbA}$, together with the matrix representation of $P_{\sbD}$ in Claim \ref{claim:U(1)-2} that
	\begin{align}
		\begin{cases}
			\widetilde{P}_{\sbA}(s_{ab},s_{cd}) \geq P_{\sbD}(s_{ab},s_{cd}), & s_{ab} \neq s_{cd}, \\
			\widetilde{P}_{\sbA}(s_{ab},s_{ab}) \leq P_{\sbD}(s_{ab},s_{ab}), & \text{for any } s_{ab},
		\end{cases}
	\end{align}
	which intuitively indicates that $\widetilde{P}_{\sbA}$ has a better transition property. Indeed, the Dirichlet form of $\widetilde{P}_{\sbA}$ is thus larger than that of $P_{\sbD}$. Since they also share the same stationary distribution in $\tilde{S}_{\sbA} \cong S_{\sbD}$,
	\begin{align}
		\Delta(\widetilde{P}_{\sbA}) \geq \Delta(P_{\sbD}).
	\end{align}
	
	What we need is the converse part. We first set up conditions in order to apply Theorem \ref{thm:path-comparison-theorem} afterwards: given two basis states $x \neq y \in \tilde{S}_{\sbA} \cong S_{\sbD}$ such that $\widetilde{P}_{\sbA}(x,y) > 0$. We assign a path $\gamma_{xy}$ indicating transition from $x$ to $y$ via $P_{\sbD}(x, y)$ as follows:
	\begin{enumerate}
		\item Suppose $P_{\sbD}(x, y)$ is also nonzero, we just set $\gamma_{xy} = (x, y)$.
		
		\item Suppose $P_{\sbD}(x, y) = 0$. Since $\widetilde{P}_{\sbA} > 0$, there must be a length-two path starting from $x$, passing through a \emph{unique} state $z$ to $y$ like the path in red in Diagram \ref{diagram:paths1} for which $P_{\sbD}(x, z), P_{\sbD}(z, y) > 0$. Then we set $\gamma_{xy} = (x,z,y)$.
		
		\item Since $P_{\sbD}$ is symmetric, $\gamma_{yx}$ is defined by the same path with the reversed direction.
	\end{enumerate}
	
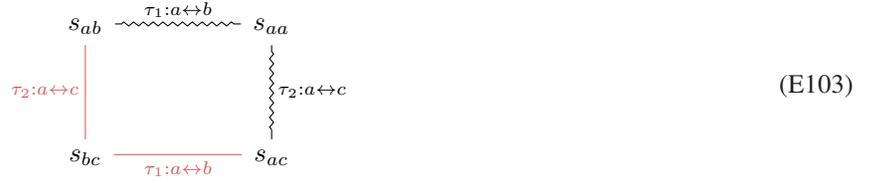
\begin{equation}\label{diagram:paths1}
		\begin{tikzcd}
			{s_{ab}} && {s_{aa}} \\
			\\
			{s_{bc}} && {s_{ac}}
			\arrow["{\tau_2:a \leftrightarrow c}"', color={rgb,255:red,214;green,92;blue,92}, no head, from=1-1, to=3-1]
			\arrow["{\tau_1:a \leftrightarrow b}", squiggly, no head, from=1-1, to=1-3]
			\arrow["{\tau_2:a \leftrightarrow c}", squiggly, no head, from=1-3, to=3-3]
			\arrow["{\tau_1:a \leftrightarrow b}"', color={rgb,255:red,214;green,92;blue,92}, no head, from=3-1, to=3-3]
		\end{tikzcd}
	\end{equation}
	
	Then we check the congestion ratio is defined in Definition \ref{def:Congestion}: 
	\begin{align}
		A = \max_{(p,q), P_{\sbD} (p,q) > 0} \frac{1}{\pi_{\sbD}(p) P_{\sbD}(p,q)} \sum_{\stackrel{x, y}{(p,q) \in \gamma_{xy}}} \widetilde{\pi}(x) \widetilde{P}_{\sbA}(x,y) \vert \gamma_{xy} \vert.
	\end{align}
	To determine the maximum, we discuss all three cases when choosing different $p,q$:
	\begin{enumerate}
		\item Suppose $p = s_{ab}, q = s_{cd}$ with $a,b \neq c,d$. Then $(p, q)$ is the only path containing itself and $\widetilde{P}_{\sbA}(p,q) = P_{\sbD}(p,q)$ by Eq.~\eqref{eq:tildeP-2} and Claim \ref{claim:U(1)-1}. As a result,
		\begin{align}
			\frac{1}{\pi_{\sbD}(p) P_{\sbD}(p,q)} \sum_{\stackrel{x, y}{(p,q) \in \gamma_{xy}}} \widetilde{\pi}(x) \widetilde{P}_{\sbA}(x,y) \vert \gamma_{xy} \vert 
			= \frac{1}{\pi_{\sbD}(p) P_{\sbD}(p,q)} \widetilde{\pi}(p) \widetilde{P}_{\sbA}(p,q) = 1.
		\end{align}
		
		\item Suppose $p = s_{ab}, q = s_{bc}$. That is, they share one common basis indices. Suppose there is no transposition exchanging either $a,b$ or $b,c$, then we cannot draw any 2-length path like that in red in Diagram \ref{diagram:paths1}. Consequently, we still $\widetilde{P}_{\sbA}(p,q) = P_{\sbD}(p,q)$ and recover the above case.
		
		\item Suppose $p = s_{ab}, q = s_{bc}$, but we assume that there exists the SWAP $\tau_1$ exchanging $a,b$. The assumption $P_{\sbD}(p,q) > 0$ from above indicates that there is also $\tau_2$ exchanging $a,c$. Consequently, both the trivial path $(s_{ab}, s_{bc})$ and the two-length path $(s_{ab}, s_{bc}, s_{ac})$ are defined and contain $(s_{ab}, s_{bc})$. One cannot find other different two-length paths that include $(s_{ab}, s_{bc})$. To check the this claim, we sketch in Diagram \eqref{diagram:paths2} with four more states $s_{ae},s_{bg},s_{bd},s_{cf}$ connecting $s_{ab}, s_{bc}$ respectively in dotted lines. It turns out that all these paths cannot exists, which we draw in dashed lines. Otherwise, assume we have the path $(s_{ab},s_{bc},s_{bd})$. Recall that we have conditions to set such a two-length path:
		\begin{align}
			\widetilde{P}_{\sbA}(s_{ab},s_{bd}) > 0, \quad 
			P_{\sbD}(s_{ab},s_{bd}) = 0,
		\end{align}
		which further implies that 
		\begin{align}
			P_{\sbA}(s_{ab},s_{bb}) > 0, \quad  P_{\sbD}(s_{bb},s_{bd}) > 0 
		\end{align}
		illustrated by break lines on the RHS of the diagram. This contradicts the criteria to assign two-length paths, which have to pass through $s_{ad}$ in green in the diagram. Another typical case like $(s_{ab},s_{bc},s_{cf})$ does not exist in the very beginning because this path involving four distinct labels $a,b,c,d$. If $\widetilde{P}_{\sbA}(s_{ab},s_{cf}) > 0$, we must have $P_{\sbD}(s_{ab},s_{cf}) > 0$ by Eq.~\eqref{eq:tildeP-2} and Claim \ref{claim:U(1)-1}. All other cases can be analyzed in the same way. 
		
		\item Reversing $p,q$ in the previous case, by the same argument, $(s_{bc},s_{ab})$ and the two-length path $(s_{ac}, s_{bc}, s_{ab})$ are the only two paths defined that include $(s_{bc},s_{ab})$. 
	\end{enumerate}
	
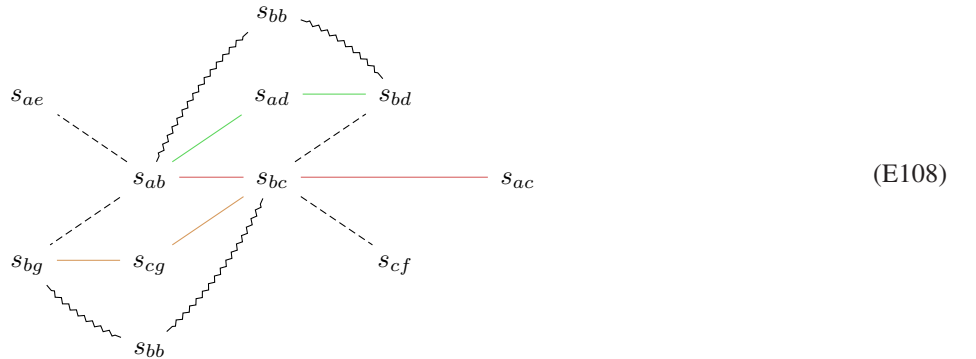
\begin{equation}\label{diagram:paths2}
		\begin{tikzcd}
			&& {s_{bb}} \\
			{s_{ae}} && {s_{ad}} & {s_{bd}} \\
			& {s_{ab}} & {s_{bc}} && {s_{ac}} \\
			{s_{bg}} & {s_{cg}} && {s_{cf}} \\
			& {s_{bb}}
			\arrow[color={rgb,255:red,214;green,92;blue,92}, no head, from=3-3, to=3-2]
			\arrow[color={rgb,255:red,214;green,92;blue,92}, no head, from=3-3, to=3-5]
			\arrow[color={rgb,255:red,92;green,214;blue,92}, no head, from=3-2, to=2-3]
			\arrow[color={rgb,255:red,92;green,214;blue,92}, no head, from=2-3, to=2-4]
			\arrow[curve={height=-6pt}, squiggly, no head, from=3-2, to=1-3]
			\arrow[draw=none, from=1-3, to=2-4]
			\arrow[curve={height=-6pt}, squiggly, no head, from=1-3, to=2-4]
			\arrow[dashed, no head, from=3-3, to=4-4]
			\arrow[dashed, no head, from=3-3, to=2-4]
			\arrow[dashed, no head, from=3-2, to=2-1]
			\arrow[dashed, no head, from=3-2, to=4-1]
			\arrow[curve={height=-6pt}, squiggly, no head, from=3-3, to=5-2]
			\arrow[curve={height=-6pt}, squiggly, no head, from=5-2, to=4-1]
			\arrow[color={rgb,255:red,214;green,153;blue,92}, no head, from=3-3, to=4-2]
			\arrow[color={rgb,255:red,214;green,153;blue,92}, no head, from=4-2, to=4-1]
		\end{tikzcd} 
	\end{equation} 
	For the third case above, we have
	\begin{align}\label{eq:Comparision}
	\begin{aligned}
		& \frac{1}{\pi_{\sbD}(p) P_{\sbD}(p,q)} \sum_{\stackrel{x, y}{(p,q) \in \gamma_{xy}}} \widetilde{\pi}(x) \widetilde{P}_{\sbA}(x,y) \vert \gamma_{xy} \vert \\
		= & \frac{1}{\pi_{\sbD}(s_{ab}) P_{\sbD}(s_{ab},s_{bc})} \widetilde{\pi}(s_{ab}) \widetilde{P}_{\sbA}(s_{ab},s_{bc}) 
		+ \frac{1}{\pi_{\sbD}(s_{ab}) P_{\sbD}(s_{ab},s_{bc})} \widetilde{\pi}(s_{ab}) \widetilde{P}_{\sbA}(s_{ab},s_{ac}) \cdot 2,
	\end{aligned}
	\end{align}
	with
	\begin{align}
		\begin{aligned}
			\widetilde{P}_{\sbA}(s_{ab},s_{ac}) 
			= P_{\sbA}(s_{ab}, s_{aa}) \frac{1}{1 - P_{\sbA}(s_{aa}, s_{aa}) } P_{\sbA}(s_{aa}, s_{ac}) 
			= \frac{2}{8 \vert \mathcal{T} \vert} \frac{1}{1 - P_{\sbA}(s_{aa}, s_{aa}) } \frac{2}{8 \vert \mathcal{T} \vert}
			\leq \frac{2}{8 \vert \mathcal{T} \vert},
		\end{aligned} 
	\end{align} 
	where the values of	$P_{\sbA}(s_{ab}, s_{aa}), P_{\sbA}(s_{aa}, s_{ac}),P_{\sbA}(s_{ab}, s_{bc})$ are obtained from \eqref{eq:U(1)++StochasticMatrix1} \& \eqref{eq:U(1)++StochasticMatrix3}. These matrix representations also provide us a simple upper bound on $P_{\sbA}(s_{aa}, s_{aa}) \leq 1 - \frac{2}{8 \vert \mathcal{T} \vert}$ because there is at least one $\tau$ acting on $a$ nontrivially. 
	
	For the last case from above, we have
	\begin{align}\label{eq:Comparision2}
	\begin{aligned}
		& \frac{1}{\pi_{\sbD}(q) P_{\sbD}(q,p)} \sum_{\stackrel{x, y}{(q,p) \in \gamma_{xy}}} \widetilde{\pi}(x) \widetilde{P}_{\sbA}(x,y) \vert \gamma_{xy} \vert \\
		= & \frac{1}{\pi_{\sbD}(s_{bc}) P_{\sbD}(s_{bc},s_{ab})} \widetilde{\pi}(s_{bc}) \widetilde{P}_{\sbA}(s_{ab},s_{bc}) 
		+ \frac{1}{\pi_{\sbD}(s_{bc}) P_{\sbD}(s_{bc},s_{ab})} \widetilde{\pi}(s_{ac}) \widetilde{P}_{\sbA}(s_{ac},s_{ab}) \cdot 2,
	\end{aligned}
	\end{align}
	By Claim \ref{claim:U(1)-2} and discussions afterwards, $\pi_{\sbD} = \widetilde{\pi}$ and they are uniformly distributed, which implies $\widetilde{\pi}(s_{ac})/ \pi_{\sbD}(s_{bc}) = 1$ from the second term in the above equation. Since all involved transition matrices are symmetric, this equation can be bounded in the same way as Eq.~\eqref{eq:Comparision}. Additionally, $P_{\sbD}(s_{ab}, s_{bc}) = \frac{1}{8 \vert \mathcal{T} \vert}$ by Claim \ref{claim:U(1)-2}. Assembling all these results, we conclude that
	\begin{align}
		A \leq 1 + \frac{1}{P_{\sbD}(s_{ab}, s_{bc})} \frac{2}{8 \vert \mathcal{T} \vert} \cdot 2 = 5
		\implies \Delta(\widetilde{P}_{\sbA})  
		\leq 5\Delta(P_{\sbD}).
	\end{align}
	by Theorem \ref{thm:path-comparison-theorem}.
\end{proof}


Since the spectral gap of $M_{2,\U(1)}^{\mathcal{E}_{\CQA}}$ is
\begin{align}
	\min_{\lambda} \Big\{ \Delta(M_{2,\U(1),\sbA}^{\mathcal{E}_{\CQA}}), \ \Delta(M_{2,\U(1),\sbD}^{\mathcal{E}_{\CQA}}), \ 1 - \lambda_1( M_{2,\U(1),\ytableaushort{ {} {*(black)} , {*(black)} {} }}^{\mathcal{E}_{\CQA}}) )  \Big\}, \tag{D21$^\ast$}
\end{align}
combining \eqref{eq:Inequality-1}, \eqref{eq:Inequality-2} and \eqref{eq:Inequality-3}, we conclude that
\begin{align}
	& \begin{cases}
		& \frac{1}{4} \Delta(M_{2,\U(1),\sbA}^{\mathcal{E}_{\CQA}})
		\leq \Delta( \Cay_{2,\U(1)} ) 
		\leq 2 \Delta(M_{2,\U(1),\sbA}^{\mathcal{E}_{\CQA}}) 
		\leq 10 \Delta(M_{2,\U(1),\sbD}^{\mathcal{E}_{\CQA}}) \\
		& \lambda_1( M_{2,\U(1),\ytableaushort{ {} {*(black)} , {*(black)} {} }}^{\mathcal{E}_{\CQA}}) ) \leq \lambda_2( \Cay_{2,\U(1)}),
	\end{cases} \\
	\implies & \frac{1}{40} \Delta( M_{2,\U(1)}^{\mathcal{E}_{\CQA}} ) \leq \frac{1}{10} \Delta( \Cay_{2,\U(1)} )  \leq \Delta( M_{2,\U(1)}^{\mathcal{E}_{\CQA}} ),
\end{align}
where $\Delta( \Cay_{2,\U(1)} )$ is bounded by Theorem \ref{thm:CayleyGap} and Claim \ref{claim:Cayley-2}. This completes the proof of $\U(1)$ case in Theorem \ref{thm:CQAGap}.


\subsection{Proof of \eqref{eq:Case3} under $\SU(d)$ symmetry}\label{sec:detailsSU(d)}

We do not need to prove from scratch for the $\SU(d)$ case since the Young orthogonal form introduced in Section \ref{sec:SnTheory} of any \emph{adjacent} transposition/SWAP $\tau = (j, j+1)$ admits a matrix representation fairly similar to the case when the SWAP acts on computational basis:
\begin{align}\label{eq:YoungOrthogonal2}
	\begin{pmatrix}
		\frac{1}{r} & \sqrt{1 - \frac{1}{r^2}} \\ \sqrt{1 - \frac{1}{r^2}} & -\frac{1}{r}
	\end{pmatrix}, \quad 
	\begin{pmatrix}
		0 & 1 \\ 1 & 0
	\end{pmatrix},
\end{align} 
where $r \in \{\pm 1,...,\pm n-1\}$ is determined by content vectors (see Section \ref{sec:SnTheory}). In a non-rigorous sense, the LHS matrix tends to the RHS as $r \to \infty$, but we have to keep in mind that these two matrices act on distinct representation spaces under different bases. Nevertheless, this observation indicates that many technique details we verified for the $\U(1)$ case still work here.

We now evaluate the spectral gap of $M^{\mathcal{E}_{\CQA}}_{2, \SU(d)}$, still using the notion of invariant subspaces as introduced in Section \ref{sec:kDesigns}. As a caveat, most notations like $S_{\sbA}$ are kept the same as before, but indices of basis states like $(E\indices{^a_a} E\indices{^b_b} + E\indices{^b_b} E\indices{^a_a}) +  (E\indices{^a_b} E\indices{^b_a} + E\indices{^b_a} E\indices{^a_b})$ stand for Young basis elements now. Moreover, the generating set $\mathcal{T}$ now exclusively refers to the collection of all nearest-neighbour, geometrically local transpositions because the proofs are largely based on Young orthogonal forms, which are only defined for these transpositions.

To prove Lemma \ref{lemma:CayletBound1} for the case under $\SU(d)$ symmetry, we first verify the following claims using Young orthogonal form:

\begin{claim}\label{claim:SU(d)-1}
	The following facts hold for the modified CQA second moment operator $M_{2,\SU(d)}^{\mathcal{E}_{\CQA}}$:
	\begin{enumerate}
		\item Within 
		\begin{align}
			S_{\sbA} = \operatorname{span} \left\{ \frac{1}{2} \Big( (E\indices{^a_a} E\indices{^b_b} + E\indices{^b_b} E\indices{^a_a}) +  (E\indices{^a_b} E\indices{^b_a} + E\indices{^b_a} E\indices{^a_b}) \Big), E\indices{^a_a} E\indices{^a_a}; a < b  \right\} \tag{D5$^\ast$}
		\end{align}
		and the under the above prescribed basis states, $M_{2,\SU(d),\sbA}^{\mathcal{E}_{\CQA}}$ is irreducible, doubly stochastic and symmetric.
		
		\item Any nonzero off-diagonal entry of  $M_{2,\SU(d),\sbA}^{\mathcal{E}_{\CQA}}$ comes from the action of a unique $\tau$ when defining $\mathcal{E}_{\CQA}$.
	\end{enumerate}
\end{claim}
\begin{proof}
	The proof is similar to that for Claim \ref{claim:U(1)-1} except we should consider the matrix representation of 
	\begin{align}
		M_{2,\SU(d)}^\tau = T^{\YJM}_2 \left( \frac{6}{8} IIII + \frac{1}{8}(I \tau I \tau + \tau I \tau I + \tau I I \tau + I \tau \tau I - II \tau \tau - \tau \tau II) \right) T^{\YJM}_2
	\end{align}
	using Young orthogonal forms. The action of $M_{2,\U(1)}^\tau$ on \eqref{eq:InvariantSubapce-states1} can be decomposed into even smaller invariant subspaces classified as the following four types (cf. $\U(1)$ cases in Claim \ref{claim:U(1)-1}):
	\begin{align}\label{eq:stochastic-matrix-invariant-++-SU(d)}
		\begin{cases}
			\frac{1}{8} \begin{pmatrix}
				6 + \frac{2}{r_{ab}^2} - 2\Big(1 - \frac{1}{r_{ab}^2} \Big) & 2\Big(1 - \frac{1}{r_{ab}^2} \Big) & 2\Big(1 - \frac{1}{r_{ab}^2} \Big) \\
				2\Big(1 - \frac{1}{r_{ab}^2} \Big) & 6 + \frac{2}{r_{ab}^2}  & 0 \\
				2\Big(1 - \frac{1}{r_{ab}^2} \Big) & 0 & 6 +  \frac{2}{r_{ab}^2}  \\
			\end{pmatrix}, & \text{if }\ \tau\vert_{\text{span}\{a,b\}} = \begin{pmatrix}
			\frac{1}{r_{ab}} & \sqrt{1 - \frac{1}{r_{ab}^2}} \\ \sqrt{1 - \frac{1}{r_{ab}^2}} & -\frac{1}{r_{ab}}
			\end{pmatrix};
			\vspace{3 mm} \\
			\frac{1}{8}  \begin{pmatrix}
				6 + \frac{1}{r_{ac}^2} +  \frac{1}{r_{bd}^2} & 0 & 1 - \frac{1}{r_{bd}^2} & 1 - \frac{1}{r_{ac}^2}  \\
				0 & 6 + \frac{1}{r_{ac}^2} +  \frac{1}{r_{bd}^2} & 1 - \frac{1}{r_{ac}^2} &  1 - \frac{1}{r_{bd}^2} \\
				 1 - \frac{1}{r_{bd}^2} & 1 - \frac{1}{r_{ac}^2} & 6 + \frac{1}{r_{ac}^2} +  \frac{1}{r_{bd}^2} & 0 \\
				1 - \frac{1}{r_{ac}^2} &  1 - \frac{1}{r_{bd}^2} & 0 & 6 + \frac{1}{r_{ac}^2} +  \frac{1}{r_{bd}^2} \\
			\end{pmatrix}, & \text{if } \begin{cases}
			 \tau\vert_{\text{span}\{a,c\}} = \begin{pmatrix}
				\frac{1}{r_{ac}} & \sqrt{1 - \frac{1}{r_{ac}^2}} \\ \sqrt{1 - \frac{1}{r_{ac}^2}} & -\frac{1}{r_{ac}}
			\end{pmatrix}, \\
			\tau\vert_{\text{span}\{b,d\}} = \begin{pmatrix} 
				\frac{1}{r_{bd}} & \sqrt{1 - \frac{1}{r_{bd}^2}} \\ \sqrt{1 - \frac{1}{r_{bd}^2}} & -\frac{1}{r_{bd}}
			\end{pmatrix};
			\end{cases}
			\vspace{3 mm} \\
			\frac{1}{8} \begin{pmatrix}
				7 + \frac{1}{r_{ac}^2} & 1 -\frac{1}{r_{ac}^2} \\
				1 - \frac{1}{r_{ac}^2} & 7 + \frac{1}{r_{ac}^2} \\
			\end{pmatrix}, & \text{if } \begin{cases}
			\tau\vert_{\text{span}\{a,c\}} = \begin{pmatrix}
				\frac{1}{r_{ac}} & \sqrt{1 - \frac{1}{r_{ac}^2}} \\ \sqrt{1 - \frac{1}{r_{ac}^2}} & -\frac{1}{r_{ac}}
			\end{pmatrix}, \\
			\tau(b) = \pm b;
			\end{cases}
			\vspace{3 mm} \\
			\frac{1}{8}\begin{pmatrix}
				8 \\
			\end{pmatrix}, & \text{if }\ \tau(a) = \pm a, \tau(b) = \pm b.
		\end{cases}
	\end{align}
	Note that a case similar to the third one happens if  $\tau\vert_{\text{span}\{b,d\}}$ restricts to an Young orthogonal form with $\tau(a) = \pm a$.
	
	Notice that each submatrix is symmetric and doubly stochastic and each subblock cannot overlap. After summing over the generating set $\mathcal{T}$ and divided by the normalization constant, $M_{2,\SU(d)}^{\mathcal{E}_{\CQA}}$ is a doubly stochastic symmetric matrix. The irreducibility is proved in the same way as Claim \ref{claim:U(1)-1} which guarantees a unique uniform stationary distribution:
	\begin{align}
		\pi_{\sbA}(v_{ab}) \equiv \frac{2}{d_{\lambda }(d_\lambda + 1)}, 
	\end{align}
	where $d_\lambda$ now refers to the dimension of the $S_n$ irrep $S^\lambda$.
	
	The second property is proved by the fact that $\tau$ is a nearest-neighbour transposition. Its action on Young basis, even though being different from the action on computational basis, is just given by transposition on components of the content vectors corresponding to the Young basis elements. Like the $\U(1)$ case, any nonzero off-diagonal entry $M_{2,\SU(d),\sbA}^{\mathcal{E}_{\CQA}}$ comes from the action of a unique  $M_{2,\SU(d),\sbA}^{\tau}$.
\end{proof}


\begin{claim}
	Analogously, for the Cayley moment operator 
	\begin{align}
		\Cay_{2,\SU(d)} = \frac{1}{\vert \mathcal{T} \vert} \sum_{\tau \in \mathcal{T}} T_2^{\YJM} \frac{1}{8}(6IIII + \tau I \tau I + I \tau I \tau) T_2^{\YJM}, 
	\end{align}
	let
	\begin{align}\label{sym-cay-state-space-2}
		\Sym \vcentcolon = \operatorname{span} \left\{ \frac{1}{2}\Big( E\indices{^a_a} E\indices{^b_b} + E\indices{^b_b} E\indices{^a_a} \Big), E\indices{^a_a} E\indices{^a_a}; a < b \right\} \subset (S^{\lambda})^{\otimes 4}.
	\end{align}
	Then
	\begin{enumerate}		
		\item The Cayley operator is invariant under this subspace. Let $\Cay_{2,\SU(d), \Sym}$ denotes its restriction. It is an irreducible reversible row-stochastic transition matrix under basis from Eq.~\eqref{sym-cay-state-space-2}.
		
		\item Any nonzero off-diagonal entry of $\Cay_{2,\SU(d), \Sym}$ comes from the action of a unique $\tau \in \mathcal{T}$ when defining $\mathcal{E}_{\CQA}$.
	\end{enumerate}
\end{claim}
\begin{proof}
	We only write down the matrix representation induced by each $\tau$ when defining $\Cay_{2,\SU(2)}^\tau$:
	\begin{align}\label{eq: stochastic-matrix-invariant-++-SU(d)}
		\begin{cases}
			\frac{1}{8} \begin{pmatrix}
				6 + \frac{2}{r_{ab}^2} & 1 - \frac{1}{r_{ab}^2}  & 1 - \frac{1}{r_{ab}^2} \\
				2\Big(1 - \frac{1}{r_{ab}^2} \Big) & 6 + \frac{2}{r_{ab}^2}  & 0 \\
				2\Big(1 - \frac{1}{r_{ab}^2} \Big) & 0 & 6 +  \frac{2}{r_{ab}^2}  \\
			\end{pmatrix}, & \text{if }\ \tau\vert_{\text{span}\{a,b\}} = \begin{pmatrix}
				\frac{1}{r_{ab}} & \sqrt{1 - \frac{1}{r_{ab}^2}} \\ \sqrt{1 - \frac{1}{r_{ab}^2}} & -\frac{1}{r_{ab}}
			\end{pmatrix};
			\vspace{3 mm} \\
			\frac{1}{8}  \begin{pmatrix}
				6 + \frac{1}{r_{ac}^2} +  \frac{1}{r_{bd}^2} & 0 & 1 - \frac{1}{r_{bd}^2} & 1 - \frac{1}{r_{ac}^2}  \\
				0 & 6 + \frac{1}{r_{ac}^2} +  \frac{1}{r_{bd}^2} & 1 - \frac{1}{r_{ac}^2} &  1 - \frac{1}{r_{bd}^2} \\
				1 - \frac{1}{r_{bd}^2} & 1 - \frac{1}{r_{ac}^2} & 6 + \frac{1}{r_{ac}^2} +  \frac{1}{r_{bd}^2} & 0 \\
				1 - \frac{1}{r_{ac}^2} &  1 - \frac{1}{r_{bd}^2} & 0 & 6 + \frac{1}{r_{ac}^2} +  \frac{1}{r_{bd}^2} \\
			\end{pmatrix}, & \text{if } \begin{cases}
				\tau\vert_{\text{span}\{a,c\}} = \begin{pmatrix}
					\frac{1}{r_{ac}} & \sqrt{1 - \frac{1}{r_{ac}^2}} \\ \sqrt{1 - \frac{1}{r_{ac}^2}} & -\frac{1}{r_{ac}}
				\end{pmatrix}, \\
				\tau\vert_{\text{span}\{b,d\}} = \begin{pmatrix} 
					\frac{1}{r_{bd}} & \sqrt{1 - \frac{1}{r_{bd}^2}} \\ \sqrt{1 - \frac{1}{r_{bd}^2}} & -\frac{1}{r_{bd}}
				\end{pmatrix};
			\end{cases}
			\vspace{3 mm} \\
			\frac{1}{8} \begin{pmatrix}
				7 + \frac{1}{r_{ac}^2} & 1 -\frac{1}{r_{ac}^2} \\
				1 - \frac{1}{r_{ac}^2} & 7 + \frac{1}{r_{ac}^2} \\
			\end{pmatrix}, & \text{if } \begin{cases}
				\tau\vert_{\text{span}\{a,c\}} = \begin{pmatrix}
					\frac{1}{r_{ac}} & \sqrt{1 - \frac{1}{r_{ac}^2}} \\ \sqrt{1 - \frac{1}{r_{ac}^2}} & -\frac{1}{r_{ac}}
				\end{pmatrix}, \\
				\tau(b) = \pm b;
			\end{cases}
			\vspace{3 mm} \\
			\frac{1}{8}\begin{pmatrix}
				8 \\
			\end{pmatrix}, & \text{if }\ \tau(a) = \pm a, \tau(b) = \pm b.
		\end{cases}
	\end{align}
	Note that a case similar to the third one happens if  $\tau\vert_{\text{span}\{b,d\}}$ restricts to an Young orthogonal form with $\tau(a) = \pm a$.
	
	As a result, $\Cay_{2,\SU(1),\Sym}$ is irreducible, reversible and row-stochastic, where the first case for above ceases to the doubly stochastic. The stationary distribution is given by 
	\begin{align}
		\pi_{\Cay}(s_{ab})  = \begin{cases}
			\frac{2}{d_\lambda}, & a < b; \\
			\frac{1}{d_\lambda}, & a = b.
		\end{cases}
	\end{align}
	Or as a vector, it is
	\begin{align}
		\begin{aligned}
			& \frac{2}{d_\lambda^2} \sum_{a < b}  (E\indices{^a_a} E\indices{^b_b} + E\indices{^b_b} E\indices{^a_a})
			+ \frac{1}{d_\lambda^2} \sum_a E\indices{^a_a} E\indices{^a_a}.
		\end{aligned}
	\end{align}
\end{proof}

Different from the $\U(1)$ case, the decomposition of the $2$-fold tensor product $S^\lambda \otimes S^\lambda$ of any $S_n$ irrep has a multiplicity-free trivial irrep~\cite{Fulton1997,Sagan01}. Therefore, the number of unit eigenvalues of
\begin{align}
	\frac{1}{\vert \mathcal{T} \vert}  \sum_{\tau \in \mathcal{T}}  \tau I \tau I + I \tau I \tau 
\end{align} 
as well as $\Cay_{2,\SU(d)}$ projected by $T_2^{\YJM}$, is just one within $(S^\lambda)^{\otimes 4}$. Consequently,
\begin{align}
	\lambda_2 (\Cay_{2, \SU(d),\Sym}) \leq \lambda_2( \Cay ) = \frac{1}{2} \Big( 1 + \frac{1}{\mathcal{T}} \sum_{\tau \in \mathcal{T}} \tau\tau \Big) 
	\leq \frac{1}{2} \Big( 1 + \frac{1}{\vert \mathcal{T} \vert} \lambda_2( \mathcal{G}(\mathcal{T}) ) \Big)
	= \frac{1}{\vert \mathcal{T} \vert} O\Big( \lambda_2( \mathcal{G}(\mathcal{T}) ) \Big).
\end{align}


Let
\begin{align}
	P_{\sbA} \vcentcolon = M_{2,\SU(d),\sbA }^{\mathcal{E}_{\CQA}}, \quad P_{\Cay} \vcentcolon = \Cay_{2,\SU(d), \Sym}
\end{align}
As before, we compare the Dirichlet forms of $(S_{\sbA}, P_{\sbA}, \pi_{\sbA})$ and $(S_{\sbA}, P_{\Cay}, \pi_{\Cay})$. Since any nonzero off-diagonal entry of these transition matrices come from the action of a unique $\tau \in \mathcal{T}$, for any $x = v_{ab} \neq v_{cd} = y$ with $v_{ab}, v_{cd}$ being abbreviations of state vectors in Eq.~\eqref{eq:InvariantSubapce-states1}, there are simply three cases when we have nonzero transition probabilities: 
\begin{align}
	\begin{cases}
		\pi_{\Cay}(x) P_{\Cay}(x, y) = \frac{2}{d_\lambda^2} \frac{1}{8\vert \mathcal{T} \vert} \Big(1 - \frac{1}{r_1^2} \Big)
		> \pi_{\sbA}(x) P_{\sbA} (x, y) = \frac{2}{d_\lambda(d_\lambda + 1)} \frac{1}{8\vert \mathcal{T} \vert} \Big(1 - \frac{1}{r_1^2} \Big), & a \neq b, c \neq d;  \\
		\pi_{\Cay}(x) P_{\Cay}(x, y) =  \frac{2}{d_\lambda^2} \frac{1}{8\vert \mathcal{T} \vert} \Big(1 - \frac{1}{r_2^2} \Big)
		< \pi_{\sbA}(x) P_{\sbA} (x, y) = \frac{2}{d_\lambda(d_\lambda + 1)} \frac{2}{\vert 8\mathcal{T} \vert} \Big(1 - \frac{1}{r_2^2} \Big), & a \neq b, c = d; \\
		\pi_{\Cay}(x) P_{\Cay}(x, y) = \frac{1}{d_\lambda^2} \frac{2}{8\vert \mathcal{T} \vert} \Big(1 - \frac{1}{r_3^2} \Big)
		< \pi_{\sbA}(x) P_{\sbA} (x, y) = \frac{1}{d_\lambda(d_\lambda + 1)} \frac{2}{8\vert \mathcal{T} \vert} \Big(1 - \frac{1}{r_3^2} \Big), & a = b, c \neq d,
	\end{cases}
\end{align}
where $r_1,r_2,r_3$ are determined by content vectors of the Young basis elements $a,b,c,d$ (see Section \ref{sec:SnTheory}) indexed the states. In any case, we still have
\begin{align}
\frac{1}{4} \Delta(M_{2,\SU(d),\sbA})   
\leq \Delta(\Cay_{2,\SU(d), \Sym}) 
\leq 2 \Delta(M_{2,\SU(d),\sbA}),
\end{align}
which is formally the same as the $\U(1)$ case.


We now consider the case when $M_{2,\SU(d)}^{\mathcal{E}_{\CQA}}$ is restricted to
\begin{align}
	S_{\sbD} = \operatorname{span} \left\{ \frac{1}{2} \Big( E\indices{^a_a} E\indices{^b_b} + E\indices{^b_b} E\indices{^a_a}) -  (E\indices{^a_b} E\indices{^b_a} + E\indices{^b_a} E\indices{^a_b} \Big); a < b \right\}.  \tag{D6$^\ast$}
\end{align}

\begin{claim}\label{claim:SU(d)-2}
	Under bases given in Eq.~\eqref{eq:InvariantSubapce-states2}, $M_{2,\U(1),\sbD}^{\mathcal{E}_{\CQA}}$ is an irreducible doubly stochastic symmetric matrix. Any nonzero off-diagonal entry of  $M_{2,\U(1),\sbD}^{\mathcal{E}_{\CQA}}$ comes from the action of a unique $\tau \in \mathcal{T}$ when defining $\mathcal{E}_{\CQA}$.
\end{claim}
\begin{proof}
    Analogous to Claim \ref{claim:U(1)-2}, we write down the matrix representation of $M_{2,\SU(d),\sbD}^\tau$ as the following four types:
	\begin{align}\label{eq:--StochasticMatrixSU(d)}
		\begin{cases}
			\frac{1}{8} \begin{pmatrix} 8  \\
			\end{pmatrix}, & \text{if }\ \tau\vert_{\text{span}\{a,b\}} = \begin{pmatrix}
				\frac{1}{r_{ab}} & \sqrt{1 - \frac{1}{r_{ab}^2}} \\ \sqrt{1 - \frac{1}{r_{ab}^2}} & -\frac{1}{r_{ab}}
			\end{pmatrix};
			\vspace{3 mm} \\
			\frac{1}{8}  \begin{pmatrix}
				6 + \frac{1}{r_{ac}^2} +  \frac{1}{r_{bd}^2} & 0 & 1 - \frac{1}{r_{bd}^2} & 1 - \frac{1}{r_{ac}^2}  \\
				0 & 6 + \frac{1}{r_{ac}^2} +  \frac{1}{r_{bd}^2} & 1 - \frac{1}{r_{ac}^2} &  1 - \frac{1}{r_{bd}^2} \\
				1 - \frac{1}{r_{bd}^2} & 1 - \frac{1}{r_{ac}^2} & 6 + \frac{1}{r_{ac}^2} +  \frac{1}{r_{bd}^2} & 0 \\
				1 - \frac{1}{r_{ac}^2} &  1 - \frac{1}{r_{bd}^2} & 0 & 6 + \frac{1}{r_{ac}^2} +  \frac{1}{r_{bd}^2} \\
			\end{pmatrix}, & \text{if } \begin{cases}
				\tau\vert_{\text{span}\{a,c\}} = \begin{pmatrix}
					\frac{1}{r_{ac}} & \sqrt{1 - \frac{1}{r_{ac}^2}} \\ \sqrt{1 - \frac{1}{r_{ac}^2}} & -\frac{1}{r_{ac}}
				\end{pmatrix}, \\
				\tau\vert_{\text{span}\{b,d\}} = \begin{pmatrix} 
					\frac{1}{r_{bd}} & \sqrt{1 - \frac{1}{r_{bd}^2}} \\ \sqrt{1 - \frac{1}{r_{bd}^2}} & -\frac{1}{r_{bd}}
				\end{pmatrix};
			\end{cases}
			\vspace{3 mm} \\
			\frac{1}{8} \begin{pmatrix}
				7 + \frac{1}{r_{ac}^2} & 1 -\frac{1}{r_{ac}^2} \\
				1 - \frac{1}{r_{ac}^2} & 7 + \frac{1}{r_{ac}^2} \\
			\end{pmatrix}, & \text{if } \begin{cases}
				\tau\vert_{\text{span}\{a,c\}} = \begin{pmatrix}
					\frac{1}{r_{ac}} & \sqrt{1 - \frac{1}{r_{ac}^2}} \\ \sqrt{1 - \frac{1}{r_{ac}^2}} & -\frac{1}{r_{ac}}
				\end{pmatrix}, \\
				\tau(b) = \pm b;
			\end{cases}
			\vspace{3 mm} \\
			\frac{1}{8}\begin{pmatrix}
				8 \\
			\end{pmatrix}, & \text{if }\ \tau(a) = \pm a, \tau(b) = \pm b.
		\end{cases}
	\end{align}
	Note that a case similar to the third one happens if  $\tau\vert_{\text{span}\{b,d\}}$ restricts to an Young orthogonal form with $\tau(a) = \pm a$.
	
	Using almost the same argument as before, we can show that $M_{2,\SU(d)}^{\mathcal{E}_{\CQA}}$ is a doubly stochastic, irreducible and symmetric matrix. All of its nonzero off-diagonal entries come from the action of a unique $\tau \in \mathcal{T}$. The stationary distribution is uniform:
	\begin{align}\label{eq:--StationarySU(d)}
		\pi_{\sbD}(v_{ab}) \equiv \frac{2}{d_{\lambda}(d_\lambda - 1)}, 
	\end{align}
	where $d_\lambda$ is the dimension of the $S_n$ irrep $S^\lambda$ now.
\end{proof}

As before, we are going to bound the spectral gap of $M_{2,\SU(d),\sbD}^{\mathcal{E}_{\CQA}}$ using the notion of induced Markov chain. Particularly, let
\begin{align}
	\tilde{S}_{\sbA} \vcentcolon = \operatorname{span} \left\{ \frac{1}{2} \Big( (E\indices{^a_a} E\indices{^b_b} + E\indices{^b_b} E\indices{^a_a}) + (E\indices{^a_b} E\indices{^b_a} + E\indices{^b_a} E\indices{^a_b}) \Big); a < b \right\} \cong S_{\sbD}.
\end{align}
We rewrite the these basis states by $s_{ab}$ with $a \leq b$ for brevity. Let the generating set $\mathcal{T}$ \emph{exclusively} contain nearest-neighbour SWAPs. Expanding by Young orthogonal forms, the transition properties of $\widetilde{P}_{\sbA}$ induced from $P_{\sbA} = M_{2,\SU(d),\sbA}^{\mathcal{E}_{\CQA}}$ has a similar behaviour as that in the $\U(1)$ case: 
\begin{enumerate}
	\item Transition to itself with
	\begin{align}\label{eq:tildeP-5}
		\widetilde{P}_{\sbA}(s_{ab}, s_{ab}) & = P_{\sbA}(s_{ab}, s_{ab})
		+ P_{\sbA}(s_{ab}, s_{aa}) \sum^{\infty}_{p=0} P_{\sbA}(s_{aa}, s_{aa})^p P_{\sbA}(s_{aa}, s_{ab})
		+ P_{\sbA}(s_{ab}, s_{bb}) \sum^{\infty}_{p=0} P_{\sbA}(s_{bb}, s_{bb})^p P_{\sbA}(s_{bb}, s_{ab}) \notag \\
		& = P_{\sbA}(s_{ab}, s_{ab})
		+ P_{\sbA}(s_{ab}, s_{aa}) \frac{P_{\sbA}(s_{bb}, s_{ab}) }{1 - P_{\sbA}(s_{aa}, s_{aa}) }
		+ P_{\sbA}(s_{ab}, s_{bb}) \frac{P_{\sbA}(s_{bb}, s_{ab}) }{1 - P_{\sbA}(s_{bb}, s_{bb}) },
	\end{align}
	where the last two terms is nonzero if and only if there is some nearest-neighbour SWAP $\tau$ such that $\tau\vert_{\text{span}\{a,b\}}$ is a nontrivial Young orthogonal form.
	
	\item Transition to $s_{cd}$ for which $c,d \neq a,b$. Then
	\begin{align}\label{eq:tildeP-6}
		\widetilde{P}_{\sbA}(s_{ab}, s_{cd}) = P_{\sbA}(s_{ab}, s_{cd}).
	\end{align}
	
	\item Transition to $s_{ac}$ with
	\begin{align}\label{eq:tildeP-7}
		\widetilde{P}_{\sbA}(s_{ab}, s_{ac}) & = P_{\sbA}(s_{ab}, s_{ac})
		+ P_{\sbA}(s_{ab}, s_{aa}) \sum^{\infty}_{p=0} P_{\sbA}(s_{aa}, s_{aa})^p P_{\sbA}(s_{aa}, s_{ac}) \notag \\
		& = P_{\sbA}(s_{ab}, s_{ac})
		+ P_{\sbA}(s_{ab}, s_{aa}) \frac{1}{1 - P_{\sbA}(s_{aa}, s_{aa}) } P_{\sbA}(s_{aa}, s_{ac}).
	\end{align}
	where the last term is nonzero if and only if there is some nearest-neighbour SWAPs $\tau_1,\tau_2$ such that $\tau_1\vert_{\text{span}\{a,b\}}, \tau_2\vert_{\text{span}\{a,c\}}$ are all nontrivial. Diagram \eqref{diagram:paths1} in the $\U(1)$ case can still be used to illustrates the situation intuitively with break lines standing for transitions outside $\widetilde{S}_{\sbA}$.
	
	\item Transition to $s_{bd}$ with
	\begin{align}\label{eq:tildeP-8}	
		\widetilde{P}_{\sbA}(s_{ab}, s_{bd}) & = P_{\sbA}(s_{ab}, s_{bd})
		+ P_{\sbA}(s_{ab}, s_{aa}) \sum^{\infty}_{p=0} P_{\sbA}(s_{bb}, s_{bb})^p P_{\sbA}(s_{bb}, s_{bd}) \notag \\
		& = P_{\sbA}(s_{ab}, s_{bd})
		+ P_{\sbA}(s_{ab}, s_{bb}) \frac{1}{1 - P_{\sbA}(s_{bb}, s_{bb}) } P_{\sbA}(s_{bb}, s_{bd}).
	\end{align}
	where the last term is nonzero if and only if there is some nearest-neighbour SWAPs $\tau_3,\tau_4$ such that $\tau_1\vert_{\text{span}\{b,a\}}, \tau_2\vert_{\text{span}\{b,d\}}$ are nontrivial.
\end{enumerate} 
As a caveat, the precise matrix entries of $\widetilde{P}_{\sbA}$ are different from those in the $\U(1)$ case. They need being recalculated using content vectors like we do in proving Claim \ref{claim:SU(d)-1} \& \ref{claim:SU(d)-2}.

\begin{claim}\label{claim:SU(d)-3}
	The following inequalities of the spectral gaps hold:
	\begin{align}
		\Delta(M_{2,\SU(d),\sbA}^{\mathcal{E}_{\CQA}}) 
		= \Delta (P_{\sbA})  
		\leq \Delta(\widetilde{P}_{\sbA}) 
		\leq 7 \Delta(P_{\sbD}) 
		= 7 \Delta(M_{2,\SU(d),\sbD}^{\mathcal{E}_{\CQA}})
		\leq 7 \Delta(\widetilde{P}_{\sbA}). 
	\end{align}
\end{claim}
\begin{proof}
	Like Claim \ref{claim:U(1)-3}, we need to define paths relative to transitions given by $P_{\sbD}$ in order to apply Theorem \ref{thm:comparison-dirichlet}. With the understanding on actions of nearest-neighbour SWAPs through Young orthogonal forms and content vectors, we can easily confirm that these paths can be defined by the same criteria as the $\U(1)$ case. Especially, Diagrams \ref{diagram:paths1} \& \ref{diagram:paths2} can be reused to indicate valid paths when the basis indices $a,b,c,d,...$ are treated as Young basis elements here.
	
	As a result, we only need to examine \eqref{eq:Comparision} \& \eqref{eq:Comparision2} for the $\SU(d)$ case. Suppose $p = s_{ab}, q = s_{bc}$ and suppose there is nearest-neighbour SWAPs $\tau_1,\tau_2$ such that $\tau_1\vert_{\text{span}\{a,b\}}, \tau_2\vert_{\text{span}\{a,c\}}$ are all nontrivial. We compute
	\begin{align}
	\begin{aligned}
		& \frac{1}{\pi_{\sbD}(p) P_{\sbD}(p,q)} \sum_{\stackrel{x, y}{(p,q) \in \gamma_{xy}}} \widetilde{\pi}(x) \widetilde{P}_{\sbA}(x,y) \vert \gamma_{xy} \vert \\
		= & \frac{1}{\pi_{\sbD}(s_{ab}) P_{\sbD}(s_{ab},s_{bc})} \widetilde{\pi}(s_{ab}) \widetilde{P}_{\sbA}(s_{ab},s_{bc}) 
		+ \frac{1}{\pi_{\sbD}(s_{ab}) P_{\sbD}(s_{ab},s_{bc})} \widetilde{\pi}(s_{ab}) \widetilde{P}_{\sbA}(s_{ab},s_{ac}) \cdot 2.
	\end{aligned}
	\end{align}
	By Claim \ref{claim:SU(d)-3},
	\begin{align}
		P_{\sbA}(s_{ab}, s_{bc}) = \frac{1}{8 \vert \mathcal{T} \vert} \Big( 1 - \frac{1}{r_{2,ac}^2} \Big), \quad 
		P_{\sbA}(s_{ab}, s_{aa}) = \frac{1}{8 \vert \mathcal{T} \vert} 2\Big( 1 - \frac{1}{r_{1,ab}^2} \Big), \quad
		P_{\sbA}(s_{aa}, s_{ac}) = \frac{1}{8 \vert \mathcal{T} \vert} 2\Big( 1 - \frac{1}{r_{2,ac}^2} \Big),
	\end{align}
	where $r_{1,ab}, r_{2,ac}$ appear from the Young orthogonal forms of $\tau_1,\tau_2$ respectively (cf. Diagram \eqref{diagram:paths1}). Then
	\begin{align}
		\begin{aligned}
			\widetilde{P}_{\sbA}(s_{ab},s_{ac}) 
			= & P_{\sbA}(s_{ab}, s_{aa}) \frac{1}{1 - P_{\sbA}(s_{aa}, s_{aa}) } P_{\sbA}(s_{aa}, s_{ac}) \\
			= & \frac{2}{8 \vert \mathcal{T} \vert} \Big( 1 - \frac{1}{r_{1,ab}^2} \Big) \frac{1}{1 - P_{\sbA}(s_{aa}, s_{aa}) } \frac{2}{8 \vert \mathcal{T} \vert} \Big( 1 - \frac{1}{r_{2,ac}^2} \Big)
			\leq \frac{1}{8 \vert \mathcal{T} \vert} \Big( 1 - \frac{1}{r_{2,ac}^2} \Big) \cdot \frac{8}{3}.
		\end{aligned} 
	\end{align} 
	To find the above upper bound, we set $r_{1,ab}^2 \to \infty$ (even they are upper bounded by $n^2$), and we take a simple upper bound on $P_{\sbA}(s_{aa}, s_{aa})$:
	\begin{align}
		P_{\sbA}(s_{aa}, s_{aa}) \leq 1 - \frac{2}{8 \vert \mathcal{T} \vert} \Big( 1 - \frac{1}{r^2} \Big) < 1 - \frac{2}{8 \vert \mathcal{T} \vert} \cdot \frac{3}{4}
	\end{align}
	where $r$ appears from the Young orthogonal form of a certain $\tau$ acting on $a$ nontrivially and we set $r^2 = 4$. 
	
	On the other hand,
	\begin{align}
		\begin{aligned}
			& \frac{1}{\pi_{\sbD}(q) P_{\sbD}(q,p)} \sum_{\stackrel{x, y}{(q,p) \in \gamma_{xy}}} \widetilde{\pi}(x) \widetilde{P}_{\sbA}(x,y) \vert \gamma_{xy} \vert \\
			= & \frac{1}{\pi_{\sbD}(s_{bc}) P_{\sbD}(s_{bc},s_{ab})} \widetilde{\pi}(s_{bc}) \widetilde{P}_{\sbA}(s_{ab},s_{bc}) 
			+ \frac{1}{\pi_{\sbD}(s_{bc}) P_{\sbD}(s_{bc},s_{ab})} \widetilde{\pi}(s_{ac}) \widetilde{P}_{\sbA}(s_{ac},s_{ab}) \cdot 2
		\end{aligned}
	\end{align}
	can be bounded in the same way as earlier one because all involved transition matrices are symmetric and we still have $\widetilde{\pi}(s_{ac})/ \pi_{\sbD}(s_{bc}) = 1$ from the second term in the above equation, just like the case in Claim \ref{claim:U(1)-3}. By Claim \ref{claim:U(1)-2},
	\begin{align}
		P_{\sbD}(s_{ab}, s_{bc}) = \frac{1}{8 \vert \mathcal{T} \vert} \Big( 1 - \frac{1}{r_{2,ac}^2} \Big).
	\end{align}
	Assembling all these results, we conclude that
	\begin{align}
		A \leq 1 + \frac{1}{P_{\sbD}(s_{ab}, s_{bc})} \frac{1}{8 \vert \mathcal{T} \vert} \Big( 1 - \frac{1}{r_{2,ac}^2} \Big)  \cdot \frac{8}{3} \cdot 2 < 7 
		\implies  \Delta(\widetilde{P}_{\sbA})  
		\leq 7\Delta(P_{\sbD}) 
	\end{align}
	by Theorem \ref{thm:path-comparison-theorem}.
\end{proof}


Since the spectral gap of $M_{2,\SU(d)}^{\mathcal{E}_{\CQA}}$ is
\begin{align}
\min_{\lambda} \Big\{ \Delta(M_{2,\SU(d),\sbA}^{\mathcal{E}_{\CQA}}), \ \Delta(M_{2,\SU(d),\sbD}^{\mathcal{E}_{\CQA}}), \ 1 - \lambda_1( M_{2,\SU(d),\ytableaushort{ {} {*(black)} , {*(black)} {} }}^{\mathcal{E}_{\CQA}}) )  \Big\}, \tag{D21$^\ast$}
\end{align}
we conclude that
\begin{align}
& \begin{cases}
	& \frac{1}{4} \Delta(M_{2,\SU(1),\sbA}^{\mathcal{E}_{\CQA}})
	\leq \Delta( \Cay_{2,\SU(d)} ) 
	\leq 2 \Delta(M_{2,\SU(d),\sbA}^{\mathcal{E}_{\CQA}}) 
	\leq 14 \Delta(M_{2,\SU(d),\sbD}^{\mathcal{E}_{\CQA}}) \\
	& \lambda_1( M_{2,\SU(d),\ytableaushort{ {} {*(black)} , {*(black)} {} }}^{\mathcal{E}_{\CQA}}) ) \leq \lambda_2( \Cay_{2,\SU(d)}),
\end{cases} \\
\implies & \frac{1}{56} \Delta( M_{2,\U(1)}^{\mathcal{E}_{\CQA}} ) \leq \frac{1}{14} \Delta( \Cay_{2,\U(1)} )  \leq \Delta( M_{2,\U(1)}^{\mathcal{E}_{\CQA}} ),
\end{align}
where $\Delta( \Cay_{2,\SU(d)} )$ is bounded by Theorem \ref{thm:CayleyGap}. This completes the proof of $\SU(d)$ case in Theorem \ref{thm:CQAGap}.


\end{document}